\documentclass{article}
\usepackage{kr}

\usepackage[usenames,dvipsnames,tikz]{xcolor}
\usepackage{times}
\usepackage{soul}
\usepackage{url}
\usepackage[hidelinks]{hyperref}
\usepackage[utf8]{inputenc}
\usepackage[small]{caption}
\usepackage{graphicx}
\usepackage{amsmath}
\usepackage{amsthm}
\usepackage{booktabs}
\usepackage{algorithm}
\usepackage{algorithmic}
\newtheorem{example}{Example}
\newtheorem{theorem}{Theorem}
\newtheorem{definition}{Definition}
\newtheorem{lemma}{Lemma}
\newtheorem{proposition}{Proposition}
\newtheorem{corollary}{Corollary}
\newtheorem{remark}{Remark}
\newcommand{\KripkeStructureTI}{\K_{\mathit{TI}}\xspace}

\usepackage{enumerate}
\usepackage{paralist}
\usepackage{mathtools}

\newcommand\ie{i.e.,\xspace} 

\usepackage{amsmath,amsthm}\usepackage[capitalize]{cleveref}
\usepackage[usenames,dvipsnames]{xcolor}
\usepackage{xspace}

\def \ifempty#1{\def\temp{#1} \ifx\temp\empty }

\usepackage{xparse}

\NewDocumentCommand{\advice}{om}{%
\par {\color{OliveGreen} \textbf{Advice:} #2  
\ifempty{#1} {} \else {\bf In this case: #1} \fi 
} \par
}

\newsavebox\VerbTIPcommandforconsistentformatting
\begin{lrbox}{\VerbTIPcommandforconsistentformatting}
\begin{minipage}{\textwidth}
\begin{verbatim}
\newcommand\uriname[1]{\texttt{#1}\xspace}
\end{verbatim}
\end{minipage}
\end{lrbox}

\newsavebox\VerbTIPnobracketsarounddefinitionsWRONG
\begin{lrbox}{\VerbTIPnobracketsarounddefinitionsWRONG}
\begin{minipage}{\textwidth}
\begin{verbatim}
 \newcommand{\leqp}{{\leq_p}}
 \end{verbatim}
\end{minipage}
\end{lrbox}
\newsavebox\VerbTIPnobracketsarounddefinitionsRIGHT
\begin{lrbox}{\VerbTIPnobracketsarounddefinitionsRIGHT}
\begin{minipage}{\textwidth}
\begin{verbatim}
 \newcommand{\leqp}{\leq_p}
\end{verbatim}
\end{minipage}
\end{lrbox}

\usepackage{ifthen}
\usepackage{url}
\usepackage[usenames,dvipsnames,tikz]{xcolor}
\usepackage{amssymb}
\usepackage{amsmath}
\usepackage{import}
\usepackage{xparse}
\usepackage{amsmath}
\usepackage{xspace}
\usepackage{datatool}
\usepackage{glossaries}

\providecommand\m[1]{\ensuremath{#1}\xspace}
\renewcommand{\m}[1]{\ensuremath{#1}\xspace}
\newcommand{\trval}[1]{\m{\mathbf{#1}}}

	\newcommand{\limplies}{\Rightarrow}
	
	\newcommand{\lequiv}{\Leftrightarrow}

	\newcommand{\lrule}{\leftarrow}
	\newcommand{\cause}{\stackrel{c}{\lrule}}
	
	\newcommand{\ltrue}{\trval{t}}
	\newcommand{\lfalse}{\trval{f}}
	\newcommand{\lunkn}{\trval{u}}
	
	\newcommand{\false}{\m{\bot}}

	\newcommand{\Tr}{\ltrue}
	\newcommand{\Fa}{\lfalse}
	\newcommand{\Un}{\lunkn}

	\newcommand{\voc}{\m{\Sigma}}

	\newcommand{\struct}{\m{I}}
	
	\newcommand{\I}{\m{\mathcal{I}}}

	\newcommand{\WW}{\m{\mathcal{W}}}

	\newcommand{\f}{\m{\varphi}}

	\NewDocumentCommand\inter{g+g}{%
	  \IfNoValueTF{#1}
	    {\struct}
	    {\m{#1^{#2}}}}

	\renewcommand{\int}{\m{\mathbb{Z}}}

	\newcommand{\leqp}{\m{\leq_p}}
	\newcommand{\geqp}{\m{\geq_p}}

	\newcommand{\leqt}{\m{\leq_t}}

	\DeclareMathOperator\glb{glb}
	\DeclareMathOperator\lub{lub}

	\NewDocumentCommand\subs{g+g}{%
	  \IfNoValueTF{#1}
	    {\m{/}}
	    {\m{#1/ #2}}}

\newcommand{\ouracronym}[3]{%
	\newacronym{#1}{#2}{#3}
	\expandafter\newcommand\csname #1\endcsname{\gls{#1}\xspace}%
}
	\ouracronym{FO}{FO}{first-order logic}
	\ouracronym{PC}{PC}{propositional calculus}
	\ouracronym{MX}{MX}{Model Expansion}
	\ouracronym{MO}{MO}{Model Optimization}
	\ouracronym{ASP}{ASP}{Answer Set Programming}
	\ouracronym{TP}{TP}{Theorem Proving}
	\ouracronym{LP}{LP}{Logic Programming}
	\ouracronym{CP}{CP}{Constraint Programming}
	\ouracronym{FP}{FP}{Functional Programming}
	\ouracronym{KR}{KR}{Knowledge Representation}
	\ouracronym{CSP}{CSP}{Constraint Satisfaction Problem}
	\ouracronym{SMT}{SMT}{SAT Modulo Theories}
	\ouracronym{KBS}{KBS}{knowledge base system}
	\ouracronym{NNF}{NNF}{Negation Normal Form}
	\ouracronym{FNNF}{FNNF}{Flat Negation Normal Form}
	\ouracronym{DefNNF}{DefNNF}{Definition Negation Normal Form}
	\ouracronym{DEFNF}{DEFNF}{Definition Normal Form}
	\ouracronym{CDCL}{CDCL}{Conflict-Driven Clause-Learning}
	\ouracronym{WFS}{WFS}{Well-Founded Semantics}
	\ouracronym{LCG}{LCG}{Lazy Clause Generation}
	\ouracronym{AEL}{AEL}{Autoepistemic Logic}
	\ouracronym{OEL}{OEL}{Ordered Epistemic Logic}
	\ouracronym{AFT}{AFT}{Approximation Fixpoint Theory}

	\makeatletter
	\def\ifenv#1{
	\def\@tempa{#1}%
	\def\@ttempa{#1*}%
	\ifx\@tempa\@currenvir
	\expandafter\@firstoftwo
	\else
	\expandafter\@secondoftwo
	\fi
	}
	\makeatother

	\newcommand{\ddrule}[4]{\ensuremath{#1 \leftarrow #2 & \{#3\} & #4}}
	\newcommand{\drule}[2]{\ensuremath{#1 & \leftarrow & #2}}

	\newcommand{\darule}[4]{\ensuremath{#1 \leftarrow #2 & \{#3\} & #4}}
	\newcommand{\arule}[2]{\ensuremath{#1 \, &\leftarrow \, #2}}

	\newcommand{\LNDRule}[2]{
	\ifenv{array}
	{\drule{#1}{#2}}
	{ \ifenv{align}
		{\arule{#1}{#2}}
		{\ifenv{align*}
		{\arule{#1}{#2}}
		{ERROR: using LDRule in unsupported environment: \@currenvir}
		}
	}
	}

	\newcommand{\LDRule}[4]{
	\ifenv{array}
	{\ddrule{#1}{#2}{#3}{#4}}
	{ \ifenv{align}
		{\darule{#1}{#2}{#3}{#4}}
		{\ifenv{align*}
		{\darule{#1}{#2}{#3}{#4}}
		{ERROR: using LDRule in unsupported environment: \@currenvir}
		}
	}
	}

	\NewDocumentCommand\LRule{m+g+g+g}{%
		\IfNoValueTF{#2}%
		{#1.&}{%
		\IfNoValueTF{#3}
		{\LNDRule{#1}{#2.}}
		{\LDRule{#1}{#2.}{#3}{#4}}%
		}
	}

	\NewDocumentCommand\CLRule{m+g}{%
	\ifenv{array}
	{\cdrule{#1}{#2}}
	{ \ifenv{align}
		{\carule{#1}{#2}}
		{\ifenv{align*}
			{\carule{#1}{#2}}
			{ERROR: using CLRule in unsupported environment: \@currenvir}
		}
	}
	}

	\NewDocumentCommand\carule{m+g}{%
		\IfNoValueTF{#2}
			{\ensuremath{#1.}}
			{\ensuremath{#1 \, &\cause \, #2}}}
	\NewDocumentCommand\cdrule{m+g}{%
		\IfNoValueTF{#2}
			{\ensuremath{#1.}}
			{\ensuremath{#1 & \cause & #2}}}

	\newcommand{\algrule}[4]{
	\hbox{{#1}:}& 
	\quad #2 ~\longrightarrow~ #3 
	\hbox{~ if } #4\\
	}

	\newcommand{\AlgoRule}[4]{
	\ifenv{array}
	{\algrule{#1}{#2}{#3}{#4}}
		{ERROR: using AlgoRule in unsupported environment: \@currenvir}
	}

	\newcommand{\ignore}[1]{}

	\newboolean{nocomments}
	\setboolean{nocomments}{false}

	\newboolean{commentmargin}
	\setboolean{commentmargin}{true}

	\newcommand{\namedcomment}[3]{%
		\ifthenelse{\boolean{nocomments}}%
		{}
		{
			\ifthenelse{\boolean{commentmargin}}%
				{ {\color{#3} \marginpar{\color{#3}\sc #2}#1}  }
				{  {\color{#3} {\sc #2}: #1}  }
		}%
	}
	\newcommand{\mnamedcomment}[3]{\ifthenelse{\boolean{nocomments}}{}{{\marginpar{\tiny \color{#3}{\sc #2}:#1}}}}

	\newcommand{\todo}[1]{\namedcomment{#1}{TODO}{blue}}

	\newcommand{\bart}[1]{\namedcomment{#1}{bb}{red}}

	\usepackage{soul}

\usepackage[normalem]{ulem} 
\makeatletter
\font\uwavefont=lasyb10 scaled 700
\def\spelling{\bgroup\markoverwith{\lower3.5\p@\hbox{\uwavefont\textcolor{Red}{\char58}}}\ULon}
\def\grammar{\bgroup\markoverwith{\lower3.5\p@\hbox{\uwavefont\textcolor{LimeGreen}{\char58}}}\ULon}
\def\phrasing{\bgroup\markoverwith{\lower3.5\p@\hbox{\uwavefont\textcolor{RoyalBlue}{\char58}}}\ULon}

\newcommand\remove{\bgroup\markoverwith{\textcolor{red}{\rule[0.5ex]{2pt}{0.4pt}}}\ULon}
\makeatother

\usepackage{etoolbox}

\ExplSyntaxOn
\newcommand\setcitation[2]{%
	\csdef{mycommoncitation\text_uppercase:n{#1}}{#2}}
\newcommand\getcitation[1]{%
	\csuse{mycommoncitation\text_uppercase:n{#1}}}

\ExplSyntaxOff

\setcitation{IDP}{WarrenBook/DeCatBBD18}
\setcitation{fodot}{tocl/DeneckerT08}
\setcitation{CPsupport}{ictai/DeCat13}
\setcitation{minisatid}{ictai/DeCat13}
\setcitation{FunctionDetection}{iclp/DeCatB13}
\setcitation{lazyGrounding}{jair/CatDBS15} 
\setcitation{Bootstrapping}{ngc/BogaertsJDJBD16}
\setcitation{BootstrappingIDP}{ngc/BogaertsJDJBD16}
\setcitation{FOSymmetry}{tplp/DevriendtBBD16}
\setcitation{IDP2}{lash08/WittocxMD08}
\setcitation{FOLASP}{tplp/VanDesselDV21}

\setcitation{foid}{tocl/DeneckerT08}
\setcitation{cplogic}{journal/tplp/VennekensDB10}

\setcitation{fodot2asp}{corr/DeneckerLTV19} 
\setcitation{Tarskian}{corr/DeneckerLTV19} 
\setcitation{TarskianSemanticsASP}{corr/DeneckerLTV19} 

\setcitation{SAT}{faia/336}
\setcitation{HandbookOfSAT}{faia/336}

\setcitation{DPLLT}{cav/GanzingerHNOT04}
\setcitation{SMT}{faia/BarrettSST21}

\setcitation{CP}{fai/Rossi06}
\setcitation{Inca}{iclp/DrescherW12}
\setcitation{csp2asp}{ijcai/DrescherW11}
\setcitation{EZCSP}{lpnmr/Balduccini11}

\setcitation{ASPComp2}{lpnmr/DeneckerVBGT09}
\setcitation{ASPComp3}{journals/tplp/CalimeriIR14}
\setcitation{ASPComp4}{conf/lpnmr/AlvianoCCDDIKKOPPRRSSSWX13}
\setcitation{ASPComp5}{journals/ai/CalimeriGMR16}
\setcitation{ASPComp6}{jair/GebserMR17}
\setcitation{ASPComp7}{tplp/GebserMR20}
\setcitation{AspInPractice}{synthesis/2012Gebser}
\setcitation{clasp}{ai/GebserKS12}
\setcitation{oclingo}{kr/GebserGKOSS12}
\setcitation{clingo}{iclp/GebserKKOSW16}
\setcitation{gringo}{lpnmr/GebserST07}
\setcitation{cmodels}{aaai/GiunchigliaLM04}
\setcitation{DLV}{tocl/LeonePFEGPS06}
\setcitation{GroundingBottleneck}{padl/SonPL14}

\setcitation{MaxSAT}{faia/LiM21}

\setcitation{KR}{Baral:2003}

\setcitation{inputster}{tplp/Jansen13}
\setcitation{LearningPaper}{TPLP/BruynoogheBBDDJLRDV} 
\setcitation{clog}{iclp/BogaertsVDV14} 
\setcitation{foc}{iclp/BogaertsVDV14}
\setcitation{inferenceClog}{ecai/BogaertsVDV14}
\setcitation{inferenceFOC}{ecai/BogaertsVDV14}
\setcitation{examplesClog}{nmr/BogaertsVDV14b} 
\setcitation{AFT}{DeneckerMT00}
\setcitation{KBS}{iclp/DeneckerV08}
\setcitation{KBS-invitedtalk}{jelia/Denecker16}
\setcitation{KBPE}{inap/DePooterWD11}
\setcitation{lazygroundingASP}{ijcai/BogaertsW18} 
\setcitation{justifications}{lpnmr/DeneckerBS15} 
\setcitation{justificationTheory}{lpnmr/DeneckerBS15} 
\setcitation{justificationsAlpha}{ijcai/BogaertsW18} 
\setcitation{ASP}{marek99stable} 
\setcitation{satid}{sat/MarienWDB08}
\setcitation{lazyclausegeneration}{constraints/OhrimenkoSC09}
\setcitation{FP}{ACMCS/Hudak89}
\setcitation{GroundingWithBounds}{jair/WittocxMD10}
\setcitation{GroundWithBounds}{jair/WittocxMD10}

\setcitation{LTC}{iclp/Bogaerts14}
\setcitation{SPSAT}{ictai/DevriendtBMDD12}
\setcitation{BreakID}{sat/DevriendtBBD16}
\setcitation{LCG}{stuckeyLCG}
\setcitation{MiniZinc}{conf/cp/NethercoteSBBDT07}
\setcitation{GroundedFixpoints}{ai/BogaertsVD15}
\setcitation{PartialGroundedFixpoints}{ijcai/BogaertsVD15}
\setcitation{LogicBlox}{datalog/GreenAK12}
\setcitation{proB}{journals/sttt/LeuschelB08}
\setcitation{NaturalInductions}{KR/DeneckerV14} 
\setcitation{LP}{jacm/EmdenK76}

\setcitation{AF}{ai/Dung95}
\setcitation{ADF}{kr/BrewkaW10}
\setcitation{af}{ai/Dung95}
\setcitation{adf}{kr/BrewkaW10}
\setcitation{ADFRevisited}{ijcai/BrewkaSEWW13}
\setcitation{DefaultLogic}{ai/Reiter80}
\setcitation{DL}{ai/Reiter80}
\setcitation{AEL}{mo85}
\setcitation{minisat}{sat/EenS03}
\setcitation{completion}{adbt/Clark78}
\setcitation{ClarkCompletion}{adbt/Clark78}
\setcitation{wasp}{lpnmr/AlvianoDFLR13}
\setcitation{lcg}{stuckeyLCG}
\setcitation{CEGAR}{jacm/ClarkeGJLV03}
\setcitation{CuttingPlane}{or/DantzigFJ54}
\setcitation{CuttingPlanes}{or/DantzigFJ54}
\setcitation{kodkod}{tacas/TorlakJ07}
\setcitation{CDCL}{Marques-SilvaS99}
\setcitation{1UIP}{iccad/ZhangMMM01}
\setcitation{relevance}{ijcai/JansenBDJD16}
\setcitation{relevance-implementation}{aspocp/JansenBDJD16}
\setcitation{WFS}{GelderRS91}
\setcitation{UnfoundedSet}{GelderRS91}
\setcitation{UFS}{GelderRS91}
\setcitation{stablesemantics}{iclp/GelfondL88}
\setcitation{shatter}{Shatter}
\setcitation{sbass}{drtiwa11a}
\setcitation{lparsemanual}{url:lparsemanual}
\setcitation{AIC}{ppdp/FlescaGZ04}
\setcitation{templates}{tplp/DassevilleHJD15}
\setcitation{templates2}{iclp/DassevilleHBJD16}
\setcitation{sat-to-sat}{aaai/JanhunenTT16}
\setcitation{sat-to-sat-qbf}{bnp/BogaertsJT16}
\setcitation{sat-to-sat-SO}{kr/BogaertsJT16}
\setcitation{XSB}{SwiW12}
\setcitation{KCmap}{jair/DarwicheM02}
\setcitation{KnowledgeCompilationMap}{jair/DarwicheM02}
\setcitation{TLA}{DBLP:books/aw/Lamport2002}
\setcitation{EventB}{BookAbrial2010}
\setcitation{MX}{MitchellT05}
\setcitation{MIP}{Sierksma96}
\setcitation{PerfectModel}{minker88/Przymusinski88}
\setcitation{SafeInductions}{ijcai/BogaertsVD17}
\setcitation{AIC}{ppdp/FlescaGZ04}
\setcitation{alpha}{lpnmr/Weinzierl17}
\setcitation{omiga}{jelia/Dao-TranEFWW12}
\setcitation{gasp}{fuin/PaluDPR09}
\setcitation{asperix}{lpnmr/LefevreN09a}
\setcitation{CTL}{lop/ClarkeE81}
\setcitation{AFT-AIC}{ai/BogaertsC18}
\setcitation{UltimateApproximator}{DeneckerMT04}
\setcitation{KripkeKleene}{Fitting85}
\setcitation{AFT-HO}{corr/CharalambidisRS18} 
\setcitation{HereThere}{Heyting30}
\setcitation{dAEL}{ijcai/HertumCBD16}
\setcitation{SDD}{ijcai/Darwiche11}
\setcitation{HEX}{ijcai/EiterIST05}
\setcitation{wADF}{aaai/BrewkaSWW18}
\setcitation{wADFfix}{corr/BrewkaSWW18}
\setcitation{TransitionSystems}{jacm/NieuwenhuisOT06}
\setcitation{galliwasp}{lopstr/MarpleG12}
\setcitation{GalliWasp}{lopstr/MarpleG12}
\setcitation{clingcon}{tplp/BanbaraKOS17}
\setcitation{lp2sat}{birthday/JanhunenN11}
\setcitation{lp2mip}{LIU12}
\setcitation{lp2diff}{lpnmr/JanhunenNS09}
\setcitation{lp2acyc}{ecai/GebserJR14}
\setcitation{PB}{faia/RousselM09}
\setcitation{pb}{faia/RousselM09}
\setcitation{CuttingPlanes}{dam/CookCT87}
\setcitation{RoundingSAT}{ijcai/ElffersN18}
\setcitation{PRS}{aaai/DixonG02}
\setcitation{sat4j}{jsat/BerreP10}
\setcitation{mingo}{LIU12}
\setcitation{pbmodels}{lpnmr/LiuT05}
\setcitation{HEF-LP}{lpnmr/GebserLL07}
\setcitation{HCF-LP}{amai/Ben-EliyahuD94}
\setcitation{aspcore2}{tplp/CalimeriFGIKKLM20}
\setcitation{AspCore2}{tplp/CalimeriFGIKKLM20}
\setcitation{SemanticWeb}{book/AntoniouGHH12}
\setcitation{QBF}{faia/BuningB09}
\setcitation{SMUS}{cp/IgnatievPLM15}
\setcitation{CDCLsym}{tacas/MetinBCK18}
\setcitation{SEL}{sat/DevriendtBB17}
\setcitation{Coq}{txtcs/BertotC04}
\setcitation{Isabelle}{book/NipkowPW02}
\setcitation{HOL}{tphol/SlindN08}
\setcitation{lean}{cade/MouraU21}
\setcitation{CakeML}{jfp/TanMKFON19}
\setcitation{FRAT}{lmcs/BaekCH22}
\setcitation{DRUP}{date/GoldbergN03}
\setcitation{RUP}{date/GoldbergN03}
\setcitation{GRIT}{tacas/Cruz-FilipeMS17}
\setcitation{TraceCheck}{url:TraceCheck}
\setcitation{LRAT}{cade/Cruz-FilipeHHKS17}
\setcitation{PR}{cade/HeuleKB17}
\setcitation{DRAT}{sat/WetzlerHH14}
\setcitation{satcomp2016}{aaai/BalyoHJ17}
\setcitation{satcomp2013}{url:SATcomp2013}
\setcitation{Zinc}{constraints/MarriottNRSBW08}
\setcitation{MiniZinc}{cp/NethercoteSBBDT07}
\setcitation{Essence}{constraints/FrischHJHM08}
\setcitation{smodels}{lpnmr/SyrjanenN01}
\setcitation{lparse}{Syrjanen98}
\setcitation{IDPZ3}{corr/CarbonelleVVD22} 
\setcitation{IDP-Z3}{corr/CarbonelleVVD22} 
\setcitation{Z3}{tacas/MouraB08}

\setcitation{DMN}{DMN} 
\setcitation{Planning}{book/GhallabNT04}
\setcitation{AutomatedPlanning}{book/GhallabNT04}
\setcitation{RC2}{jsat/IgnatievMM19}
\setcitation{enfragmo}{phd/Aavani14}
\setcitation{}{}
  
 \ExplSyntaxOn
\newcommand\refto[1]{%
      \ifcsname  mycommoncitation\text_uppercase:n{#1}\endcsname%
      \getcitation{#1}%
      \else%
      #1%
      \fi%
      }

\newcommand\mycite[1]{%
      \ifcsname mycommoncitation\text_uppercase:n{#1}\endcsname%
   \cite{\getcitation{#1}}%
  \else%
    \cite{#1}%
  \fi%
}	
  
\newcommand\mycitet[1]{%
      \ifcsname mycommoncitation\text_uppercase:n{#1}\endcsname%
   \citet{\getcitation{#1}}%
  \else%
    \citet{#1}
  \fi%
}	
  
\ExplSyntaxOff

\usepackage {tikz}
\usetikzlibrary{positioning}

\newcommand{\natnrs}{{\mathbb{N}}}

\newcommand{\Agents}{\mc{A}}
\newcommand{\MD}{\mathit{MD}}

\newcommand{\K}{\mc{K}^*}

\newcommand{\mc}[1]{\mathcal{#1}}
\newcommand{\cael}{\ensuremath{\mc{L}}}
\newcommand{\cocael}{\ensuremath{\mc{C}\hspace{-0.25mm}\mc{O}\hspace{-0.4mm}\cael}}
 \newcommand{\A}{\m{\mc{A}}}

\newcommand{\marcos}[1]{\namedcomment{#1}{mc}{OliveGreen}}

\newcommand\world{\m{w}}
\newcommand\otherworld{\m{v}}
\newcommand\agent{\m{A}}
\newcommand\possible[2]{\m{#1^{#2}}}
\newcommand\impossible[2]{\m{\bar{#1}^{#2}}}
\newcommand\ordinal{\m{\mu}}
\newcommand\smallerordinal{\m{\alpha}}
\newcommand\evensmallerordinal{\m{\beta}}
\newcommand\limitordinal{\m{\lambda}}
\newcommand\objworld[1]{\m{#1^{\mathit{obj}}}}
\newcommand\depth{\m{d}}
\newcommand\restr[2]{\m{#1|_{#2}}}
\newcommand\ext{\textit{ext}}
\newcommand\IH{\textnormal{IH}}

\newcommand\evalLim[2]{\m{{#1\!\!\uparrow\!\!#2}}}
\newcommand\evalLimBar[2]{\m{{\bar{#1}\!\!\uparrow\!\!#2}}}

\newcommand\restrictability{restrictability\xspace}

\newcommand\restrictabilityHoldsAt[1][\mu]{\restrictability holds at $#1$\xspace}

\newcommand\biworld[1]{\ensuremath{#1}-biworld}
\newcommand\biworlds[1]{\ensuremath{#1}-biworlds}
\newcommand\Biworld[1]{\ensuremath{#1}-Biworld}
\newcommand\prebiworld[1]{$#1$-prebiworld}
\newcommand\Prebiworld[1]{$#1$-Prebiworld}

\newcommand{\monotonicityCompl}{monotonicity of completedness\xspace}

\newcommand\monotonicityComplHoldsAt[1][\mu]{\monotonicityCompl holds at $#1$\xspace}

\newcommand\completability{completability\xspace}

\newcommand\completabilityHoldsAt[1][\mu]{\completability holds at $#1$\xspace}

\newcommand\complCond{the completedness condition\xspace}

\newcommand\complCondHoldsAt[1][\mu]{\complCond holds at $#1$\xspace}

\newcommand\onlyKnows[1]{\m{O_{\!#1}}}
\newcommand\knows[1]{\m{K_{\!#1}}}

\newcommand\pl[1]{\textit{pl}(#1)}

\renewcommand\implies{\limplies}
\renewcommand\iff{\lequiv}
\newcommand\citet[1]{\citeauthor{#1}~\shortcite{#1}\xspace}

\newcommand\entailsti{\models_{\mathit{PI}}}

 \newcommand{\conf}[1]{}
 \newcommand{\extended}[1]{#1}

\setboolean{nocomments}{true}

\title{
Mathematical Foundations for Joining Only Knowing and Common Knowledge\\ \extended{(Extended Version)}}

\author{%
	Marcos Cramer$^1$\and
	Samuele Pollaci$^2$\and
	Bart Bogaerts$^2$ \\
	\affiliations
	$^1$Institute of Artificial Intelligence, TU Dresden, Germany \\
	$^2$Department of Computer Science, Vrije Universiteit Brussel, Belgium\\
	\emails
	marcos.cramer@tu-dresden.de,
	\{samuele.pollaci, bart.bogaerts\}@vub.be
}

\renewcommand\phi{\varphi}
\begin{document}
	
	\maketitle
	\begin{abstract}
		Common knowledge and only knowing capture two intuitive and natural notions that have proven to be useful in a variety of settings, for example to reason about coordination or agreement between agents, or  to analyse the knowledge of knowledge-based agents. 
		While these two epistemic operators have been extensively studied in isolation, the approaches made to encode their complex interplay failed to capture  some essential properties of only knowing. 
		We propose a novel solution by defining a notion of \biworld{\ordinal} for countable ordinals \ordinal, which approximates not only the worlds that an agent deems possible, but also those deemed impossible. 
    This approach allows us to 
		define a multi-agent epistemic logic with common knowledge and only knowing operators, and a three-valued model semantics for it. 
		Moreover, we show that we only really need biworlds of depth at most $\omega^2+1$. 
		Based on this observation, we define a Kripke semantics on a canonical Kripke structure and show that this semantics coincides with the model semantics.
		Finally, we discuss issues arising when combining negative introspection or truthfulness with only knowing and show how positive introspection can be integrated into our logic. 
	\end{abstract}
	
	\extended{

	\subsection*{Note about Extended Version}
	
	This is the extended version of the paper ``Mathematical Foundations for Joining Only Knowing and Common Knowledge'' accepted for KR 2023. Each of the additional lemmas and proofs is included at the relevant point in the text.
	}
	
	\section{Introduction}
  When developing intelligent agents, it is important that they can reason not just about the world they are placed in, but also about the knowledge or beliefs of other agents in their environment.
  Consider for instance a traffic situation where two cars meet at a crossing. 
  When the driver of the first car observes their light is green, the typical action to be taken would be to keep on driving. 
  The reason for this is two-fold: on the one hand, the driver knows the traffic regulations and the fact that the other car should stop. 
  But on the other hand, the driver also knows (or believes) that the other driver also knows the regulations (and will likely respect them). 
  It is this knowledge about the other driver's knowledge that allows the first driver to conclude that it is safe to continue. 
  
  The formal study of knowledge and how to reason correctly about 
it 
  has a long history in knowledge representation, dating back at least to the 1960s \cite{Hintikka1962-HINKAB}. 
  This study becomes particularly interesting when, as in our 
	example, 
	multiple agents are involved. 
  Next to the standard knowledge operator $K$, we are concerned with two epistemic operators, namely \textbf{common knowledge} and \textbf{only knowing}, and their intricate interplay. 
  We say that a statement $\phi$ is \emph{common knowledge} among a group $G$ of agents (and denote this $C_G\phi$) if each agent in $G$ knows $\phi$, and also knows that every agent in $G$ knows $\phi$, and knows that everyone in $G$ knows that everyone in $G$ knows $\phi$, and so on.
  This operator is useful, for instance, for reasoning about coordination or agreement between agents.
  We say that an agent $A$ \emph{only knows} a statement $\phi$ (and denote this $\onlyKnows{A} \phi$) if the agent knows $\phi$ (denoted $\knows{A}\phi$) and moreover \emph{everything} they know follows from $\phi$ (so whenever $\knows{A}\psi$ holds, it must be the case that $\phi$ entails $\psi$). 
  This operator is useful for instance when we consider that knowledge-based agents do not know anything except for what follows from their knowledge base and we might want to reason about their knowledge as well. 

  There have been many papers studying these operators in isolation \cite{FaginHMV95,book/MeyerH95,logcom/HalpernL01,tark/WaalerS05,corr/BelleL10,ai/BelleL15} and some authors have  even studied the combination of the two \cite{ijcai/AucherB15,ijcai/BelleL15,phd/VanHertum16}.  
  However, we will argue that there are some essential properties of only knowing that none of these approaches captures.
  We will start with the good news: it is easy to develop a Kripke semantics for a logic with these two operators: given a Kripke structure $\mc{K}$
  with set of worlds $W$ and an accessibility relations $R_A$ for every agent $A$, 
the semantics for $C_G$ and $\onlyKnows{A} $ is given by
  \begin{compactitem}
  	\item $\mc{K}, w\models C_G \phi$ if $\mc{K},w'\models \phi$ for all $w'$ reachable from $w$ with edges in $\bigcup_{A\in G} R_A$, and 
  	\item $\mc{K}, w\models \onlyKnows{A}  \phi$ if for every world $w'\in W$, $\mc{K},w'\models \phi$ if and only if $(w,w')\in R_A$.  	
  \end{compactitem}
  Intuitively, the if-part in the definition of the semantics of $\onlyKnows{A}$ states that $A$ knows $\phi$ (in all worlds $A$ deems possible, $\phi$ holds), and the only-if-part ensures that the agent doesn't know anything else (all worlds in which $\phi$ holds are indeed deemed possible by the agent in question).

  If this is so easy, then where's the catch, one might wonder. 
  Well, the problem lies in the choice of the set $W$ of worlds. 
  The question we tackle is: can we construct a set of worlds $W$ that is rich enough such that, for instance,  $\onlyKnows{A} \top $ really means that that $A$ only knows tautologies in the language (\ie that agent $A$ ``knows nothing'')?
  In other words, can we construct a canonical Kripke structure for this logic? 
	A naive first attempt at doing so would be as follows. 
  \begin{definition}[World --- incorrect definition]
  	Given a propositional vocabulary $\voc$, a world $w$ consists of 
  	\begin{compactitem}
  		\item a (classical propositional) interpretation $\objworld{\world}$ over $\voc$ (the \emph{objective interpretation} of $w$), and 
  		\item for each agent $A\in\Agents$, a set of worlds $A^w$.
  	\end{compactitem}
  \end{definition}
  However, the attentive reader might have noticed that this 
circular definition 
breaks the basic rules of set theory (a world is defined to consist of, among others, a set of worlds). 
	We are not the first to observe this issue.
	The most common approach to alleviate it is to approximate the knowledge of agents up to a certain level and defining $k{+}1$-worlds as consisting, among others, of a set $A^w$ of $k$-worlds for each agent $A$. 
	There are two challenges with this 
	approach. 
	The first is that	as soon as we add common knowledge to our language, there is a strong need for these worlds to be \emph{infinitely deep}. 
	One way to achieve this is to consider an infinite \emph{precision-increasing} sequence of $k$-worlds, as is done for instance by 
	\citet{jacm/FaginHV91},  or by \citet{ijcai/BelleL15}.
	However, that in itself does not suffice: 
  we show that to evaluate certain formulas there is a need to have worlds that are even deeper than this. 
	The second is that in such approximations, no matter how deep one goes, there is never enough information to conclude that this is ``all we know'': it might always be that by making the approximation more precise, more knowledge at some later level comes in. 
	
	This brings us to the main contribution of this paper: a solution to the above problem. 
	First we will define a notion of \biworld{\mu}, where $\mu$ can be any countable ordinal (thereby resolving the first challenge), and where the ``bi'' in biworld stands for the fact that we do not just approximate the set of worlds an agent deems possible, but also the set of worlds an agent deems \emph{impossible}. 
	Intuitively, with each \biworld{\mu{+}1} $w$, we will associate a set $A^w$ of \biworlds{\mu} representing the set of \biworlds{\mu} of which the agent deems some extension possible, and a set $\bar A^w$ of \biworlds{\mu} of which the agent deems some extension impossible. 
	This immediately allows us to see when all of agent $A$'s knowledge is captured by such a biworld: this is precisely when $A^w\cap \bar  A^w$ is empty. 
	For limit ordinals, the situation is more complex, and this makes the construction of biworlds highly technical and mathematical.  
	\conf{Due to space limitations, proofs are not included in this paper, but all propositions and lemmas are carefully proven in {\color{red}a technical report accompanying this paper \cite{CKOK_ARXIV}}.} 
	This transfinite construction of our biworlds is given in Section~\ref{sec:definitions}; Section~\ref{sec:properties} then shows several properties of them, essentially showing that they are well-behaved, in a precise sense. 
	
	In Section~\ref{sec:vals} we define our logic of common knowledge and only knowing as a simple multi-agent epistemic logic extended with the operators $\onlyKnows{A} $ and $C_G$ as described above. 
	We show that the biworlds possess enough information so that the formulas in our logic can be evaluated in them. 
	More specifically, we define a three-valued model semantics for our logic: given a formula $\phi$ and a biworld $w$ of any depth, we define what $\phi^w$ is. This can be true, false, or unknown, where the last case represents the fact that the biworld is not sufficiently ``deep'' to evaluate the formula. 
	Moreover, we show that we do not need arbitrarily deep biworlds: if a biworld has depth at least $\omega^2+1$, then every formula will evaluate in it either to true or to false. 
	Inspired by this observation, we are able to define our canonical Kripke structure: 
	the \emph{worlds}, are precisely the $\omega^2{+}1$-biworlds that are \emph{completed}, which is a technical term to denote the fact that it identifies for every other biworld of any other level whether or not the agent believes it is possible, \ie that it characterizes complete knowledge.  
	The accessibility relations can directly be obtained from the definition of the worlds. 
	We then proceed to show that the semantics obtained from this canonical Kripke-structure actually coincides with the valuation we started from. 
	Finally, this allows us to prove several desirable properties the resulting logic satisfies (see Theorem~\ref{thm:properties} for an extensive list), including the following two: 
	\begin{compactitem}
		\item For any $\phi$ and $\psi$, if $\phi\not\models\psi$, then $ \onlyKnows{A} \phi \models \neg \knows{A}\psi$.
		\item For any $\phi$, $\onlyKnows{A} \phi\not\models \bot$.
	\end{compactitem}

The first property states precisely that whenever $\onlyKnows{A} \phi$ holds, all formulas not entailed by $\phi$ are not known (and in fact we also have that all properties entailed by $\phi$ are known by $A$). 
The second property states that for \emph{any} formula $\phi$, there is a world in which $\onlyKnows{A} \phi$ holds. 
In fact, our results are even stronger than this: we show that there is a unique state-of-mind of agent $A$ in which they know precisely $\phi$. 
These two properties of the $\onlyKnows{A} $ operator, while quite simple, are --- to the best of our knowledge --- not satisfied by any other paper combining common knowledge and only knowing. 

While developing our worlds, and our logic, we do not enforce any properties that are often associated to \emph{knowledge}. 
For instance, our worlds do not guarantee that our agents are \emph{truthful} or \emph{introspective}. 
The main focus of the paper is on how to create a 
semantic structure that allows defining a rich enough set of worlds. 
However, once this semantic structure is in place, it is possible to 
use it to define a logic that satisfies such properties as well. 
To illustrate this, we show in Section \ref{sec:TIworlds} how we can hard-code into the logic the fact that agents are positively introspective, which in fact means that it is common knowledge that all agents satisfy this property. 
We also discuss the possibility of adding truthfulness and negative introspection, highlighting a problem that arises from combining them with only knowing. We use this problem to uncover a major mistake in previous work on combining only knowing and common knowledge.

Finally, in \cref{sec:relatedWork} we 
discuss 
related work and explain why previous approaches could not achieve the two properties set out above. We conclude in \cref{sec:conclusion}.

\section{Construction of Biworlds} \label{sec:defs}\label{sec:definitions}
	
In this section, we define the concept of biworlds. Intuitively, a biworld \world consists of \emph{(1)} an objective interpretation, and \emph{(2)} for each agent \agent two sets of biworlds $\possible{\agent}{\world}$ and  $\impossible{\agent}{\world}$, where the former set contains all the biworlds the agent deems possible, and the latter the biworlds the agent deems impossible. 
	As explained before, using this definition brings us into set-theoretic problems. 
	To resolve this issue, we define the notion of \biworld{\ordinal}s,  for all countable ordinals~$\ordinal$. 
	Intuitively, a \biworld{\ordinal} does not describe the complete epistemic state of the different agents, but only  their belief about the world up to a certain depth.  
	For some formulas, having a certain depth will suffice to evaluate the formula. 
	For instance, a $1$-biworld will suffice to determine whether or not an agent knows $p$, but might not suffice to determine whether an agent knows that some other agent knows $p$. 
	Moreover, we will later show that for 
 our logic, the ordinal $\omega^2+1$ suffices, in the sense that every set of formulas is either true or false in every $\omega^2{+}1$-biworld. 
	While we only need the theory of biworlds up to $\omega^2+1$, we will still develop the theory in more generality for all \emph{countable} ordinals. From now on, when we say \emph{ordinal}, we mean \emph{countable ordinal}. \extended{We repeat this assumption only in the proofs in which we require it.}

	If $\ordinal$ is a successor ordinal $\smallerordinal+1$, then a \biworld{\ordinal} $\world$ associates with each agent two sets of $\smallerordinal$-biworlds, where
	$\possible{\agent}{\world}$ represents the set of \biworld{\smallerordinal}s that can be extended (in depth; for some sufficiently large ordinal) to a biworld the agent deems possible, and 
	$\impossible{\agent}{\world}$ represents the set of \biworld{\smallerordinal}s that can be extended to a biworld the agent deems impossible. 
	Clearly, a natural condition will then be that the union of these two sets is the set of all \biworld{\smallerordinal}s. 
	
	If there is any biworld in the intersection of $\possible{\agent}{\world}$ and $\impossible{\agent}{\world}$, this biworld must be extendible both to a biworld the agent deems possible, and to a biworld the agent deems impossible. If for all agents $A$, the intersection of $\possible{\agent}{\world}$ and $\impossible{\agent}{\world}$ is empty, this means that the information in $\world$ fully specifies which biworlds the agents deem possible and which ones impossible. Once this is fully specified, there is only one way to extend $\world$ to a biworld of a greater depth. 
	When there are multiple ways in which a $\ordinal$-biworld $\world$ can be extended to a $\ordinal{+}1$-biworld, then we call $\world$ an \emph{incompleted} biworld. 

    Since a biworld in the intersection of $\possible{\agent}{\world}$ and $\impossible{\agent}{\world}$ must be extendible (for some sufficiently large depth) both to a biworld the agent deems possible and to a biworld the agent deems impossible, it must be incompleted.

We start with the auxiliary notion of a \prebiworld{\ordinal}, which approximates the more complex notion of a \biworld{\ordinal} while leaving out some of the structural conditions that we impose on \biworld{\ordinal}s. We define \prebiworld{\ordinal}s, restrictions of prebiworlds and a precision order on the prebiworlds through 
\textbf{simultaneous recursion},  
but for 
readability, we split  
it  
into Definitions \ref{def:new:prebiworld}, \ref{def:new:restriction} and~\ref{def:precision}.

\begin{definition}[\Prebiworld{\ordinal}]
	\label{def:new:prebiworld}
	Let \ordinal be an ordinal. 
	We define the set of $\ordinal$\emph{-prebiworlds} over a propositional vocabulary $\Sigma$, and a set of agents $\Agents$ by transfinite induction:
	
	\begin{compactitem}
		\item A \prebiworld{0} $\world$ is an interpretation $\I$ over $\Sigma$. We define $\objworld{\world} = \I$ and $\depth(\world)=0$. 
		
		\item \mbox{A \prebiworld{\ordinal{+}1} $\world$ is a triple  $(\I,(\possible{\agent}{\world})_{\agent\in\Agents}, (\impossible{\agent}{\world})_{\agent\in\Agents})$,} where $\I$ is an interpretation over $\Sigma$, and for each $\agent \in \Agents$, $\possible{\agent}{\world}$ and $\impossible{\agent}{\world}$ are sets of \prebiworld{\ordinal}s.
		We define $\depth(\world)= \ordinal+1$ and $\objworld{\world}=\I$. 
		
		\item A \prebiworld{\limitordinal} $w$ for a limit ordinal $\limitordinal$ is an increasing transfinite sequence $(\world_\smallerordinal)_{\smallerordinal<\limitordinal}$ of prebiworlds, i.e.\
		 for each $\smallerordinal<\limitordinal$, $\world_\smallerordinal$ is an \prebiworld{\smallerordinal} and for each  $\evensmallerordinal\leq\smallerordinal$ it must be that $\world_\evensmallerordinal\leqp \world_\smallerordinal$. 
		 We define $\depth(\world)=\limitordinal$ and $\objworld{\world}=\world_0$.
	\end{compactitem}
\end{definition}

For each ordinal $\ordinal$, we denote by $\WW_p^\ordinal$ the set of \prebiworld{\ordinal}s, and we call the integer $d(w)$ the \emph{depth of} $w$. If $\world\in\WW_p^\limitordinal$ with $\limitordinal$ a limit ordinal, then for each $\alpha<\limitordinal$, we denote by $(\world)_\alpha$ the \prebiworld{\smallerordinal} which is the element in position $\smallerordinal$ in the transfinite sequence represented by $\world$.

The following definition captures what it means to restrict a prebiworld of depth $\ordinal$ to a prebiworld of smaller depth $\smallerordinal$.

\begin{definition}[Restriction of a prebiworld]
	\label{def:new:restriction}
	Assume \world is a \prebiworld{\ordinal} and $\smallerordinal\leq\ordinal$. 
	The \emph{restriction} of $\world$ to $\smallerordinal$ is the \prebiworld{\smallerordinal} $\restr{\world}{\smallerordinal}$ defined as follows: 
	\begin{compactitem}
		\item If $\smallerordinal=0$, then $\restr{\world}{\smallerordinal}=\objworld{\world}$. 
		\item If $\smallerordinal$ is a limit ordinal, then $\restr{\world}{\smallerordinal} = (\restr{\world}{\evensmallerordinal})_{\evensmallerordinal<\smallerordinal}$
		\item If $\smallerordinal=\smallerordinal'+1$, we distinguish two cases: 
		\begin{compactitem}
			\item If $\ordinal$ is a limit ordinal, then $\restr{\world}{\smallerordinal} = (\world)_\smallerordinal$. 
			\item If $\ordinal$ is a successor ordinal $\ordinal=\ordinal'+1$, then $\objworld{\restr{\world}{\smallerordinal}}=\objworld{\world}$, 
				$\possible{\agent}{\restr{\world}{\smallerordinal}} = \{\restr{\world'}{\smallerordinal'} \mid \world'\in \possible{\agent}{\world}\}$, and 
				$\impossible{\agent}{\restr{\world}{\smallerordinal}} = \{\restr{\world'}{\smallerordinal'} \mid \world'\in \impossible{\agent}{\world}\}$.
		\end{compactitem}
	\end{compactitem}
\end{definition}

The precision order on the prebiworlds is based on the notion of restriction:

\begin{definition}[Precision order]\label{def:precision}
	If $\world$ is a \prebiworld{\ordinal} and \otherworld an \prebiworld{\smallerordinal} with $\smallerordinal\leq\ordinal$, we say that $\otherworld$ is less precise than \world (and denote this $\otherworld\leqp\world$) if $\restr{\world}{\smallerordinal}=\otherworld$. 
\end{definition}
If $\world\leqp\otherworld$ for some prebiworlds $\world$ and $\otherworld$, we call $\world$  a \emph{restriction of}  $\otherworld$ and $\otherworld$  an \emph{extention of}  $\world$.

Now we define  the set of incompleted prebiworlds and the set of biworlds by simultaneous transfinite recursion. For better readability, we separate this simultaneous definition into Definitions \ref{def:new:incompleted} and~\ref{def:new:biworld} and explain afterwards why it is a successful definition.

\begin{definition}[Incompleted prebiworld] \label{def:new:incompleted}
	A \prebiworld{\ordinal} \world is \emph{incompleted} if there exists two distinct \biworld{\ordinal{+}1}s $\otherworld_1$ and $\otherworld_2$ such that $\otherworld_1,\otherworld_2\geqp \world$ . 
\end{definition}

A prebiworld is called \emph{completed} if it is not incompleted.

\begin{definition}[\Biworld{\ordinal}] \label{def:new:biworld} 
	A \ordinal-\emph{biworld} is a  \prebiworld{\ordinal} \world such that one of the following conditions holds:
	\begin{compactenum}[(a)]
		\item $\ordinal=0$;
		\item \label{case:def:biworld_successor} $\ordinal=\ordinal'+1$ is a successor ordinal,  and for each agent $\agent\in\A$ the following hold: 
		\begin{compactenum}[(1)]
			\item \label{item:union_AandAbar_is_all_biworlds} the union $\possible{\agent}{\world}\cup\impossible{\agent}{\world}$ is the set of all $\ordinal'$-biworlds;
			\item for each $\otherworld\in\possible{\agent}{\world}\cap\impossible{\agent}{\world}$, $\otherworld$ is incompleted;
		\end{compactenum}
		\item $\ordinal$ is a limit\extended{ ordinal}, and $\world_\smallerordinal$ is an \biworld{\smallerordinal} for each $\smallerordinal<\ordinal$.
	\end{compactenum}
\end{definition}

For each ordinal $\ordinal$, the set of \biworld{\ordinal}s is denoted as $\WW^\ordinal$. We often call a \biworld{0} an \emph{objective world}.
 
Since the definition of \emph{incompleted prebiworlds} of depth $\ordinal$ depends on the existence of biworlds of depth $\ordinal+1$, an attentive reader may be worried whether the simultaneous definition of \emph{incompleted prebiworlds} and \emph{biworlds} is really a successful definition. We therefore now explain how this definition is to be understood. 

\begin{figure}[h] 
\vspace{-1mm}
		\centering
		\includegraphics[width=0.9\columnwidth]{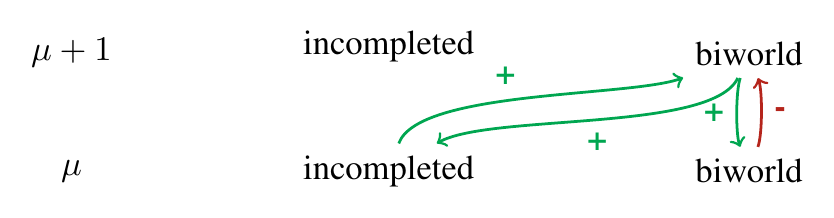}
\vspace{-3mm}
	\caption{Dependencies between the notions of \emph{incompleted prebiworlds} and \emph{biworlds} at levels $\ordinal$ and $\ordinal+1$}
	\label{fig:dependencies}
	\vspace{-4mm}
\end{figure}	

For this Figure~\ref{fig:dependencies} depicts the dependencies between the notions of \emph{incompleted prebiworlds} and \emph{biworlds} at levels $\ordinal$ and $\ordinal+1$.
Here, a green arrow from one notion to another indicates that the second notion positively dependends on the first, \ie will include more objects when the first notion includes more objects. A red arrow, on the other hand, indicates a negative dependency, so that the second notion will include fewer objects when the first notion includes more objects. For example, the green arrow from \emph{biworld} at level $\mu+1$ to \emph{incompleted} at level $\mu$ indicates that determining additional biworlds at level $\mu+1$ can lead to determining a prebiworld $w$ at level $\mu$ to be incompleted, as one of the newly determined biworlds at level $\mu+1$ may 
imply there is 
more than one way of extending $w$ to a $\m{+}1$-biworld. The arrow from \emph{biworld} at level $\mu$ to \emph{biworld} at level $\mu+1$ is red, because given a $\mu{+}1$-prebiworld $w$, determining more biworlds at level $\mu$ can show that the union $\possible{\agent}{\world}\cup\impossible{\agent}{\world}$ is not the set of all $\ordinal$-biworlds, \ie that $w$ is not a biworld.

Now the simultaneous definition of \emph{incompleted prebiworlds} and \emph{biworlds} is a transfinite recursion over $\mu$, where at each level $\mu$ of the induction, the notions of an \emph{incompleted $\mu$-prebiworld} and a \emph{$\mu{+}1$-biworld} are defined simultaneously by choosing the minimal set of incompleted $\mu$-prebiworlds and $\mu{+}1$-biworlds that satisfy Definitions \ref{def:new:incompleted} and~\ref{def:new:biworld}(b). Such a minimal set exists, because the definitions of \emph{incompleted $\mu$-prebiworld} and  \emph{$\mu{+}1$-biworld} only depend positively on each other. Here the negative dependency of \emph{$\mu{+}1$-biworld} on \emph{$\mu$-biworld} is not a problem, because at this stage in the transfinite recursion over $\mu$, the set of $\mu$-biworlds has already been determined: If $\mu$ is a successor ordinal $\mu'+1$, it has been determined in the previous step $\mu'$ of the transfinite recursion. If $\mu$ is a limit ordinal, it has been determined by Definition~\ref{def:new:biworld}(c) and the fact that for every $\alpha < \mu$, the set of $\alpha$-biworlds has already been determined. If $\mu=0$, the set of $\mu$-biworlds just coincides with the set of $\mu$-prebiworlds by Definition~\ref{def:new:biworld}(a).


\begin{example}\label{ex:compl_notpostcompl_omegabiw}
	We consider one of the simplest settings, namely we suppose to have just one agent $\agent$, and one propositional variable $p$ in the vocabulary $\Sigma$. The set $\WW^0$ of \biworld{0}s coincides with the set of all \prebiworld{0}s by \cref{def:new:biworld}, and it is equal to $\WW^0=\{\{p\},\{\emptyset\}\}$.
	The set $\WW^1$ already counts considerably more elements, eighteen to be precise: nine with $\{p\}$ as objective world, and other nine with $\{\emptyset\}$. This follows from the limitations imposed by Item \ref{item:union_AandAbar_is_all_biworlds} in Definition \ref{def:new:biworld} of biworlds\footnote{In the single-agent setting, $|\WW^1|=n3^n$, where $n=|\WW^0|$.}. For the sake of conciseness, we omit the explicit description of all the \biworld{1}s but one:
	\vspace{-4.5mm}
	\begin{align*}
		v_1:=(\{p\}, \{\{p\}\}, \{\{p\}, \{\emptyset\}\}).
	\end{align*}
	\vspace{-4.5mm}\\
	Since $\{p\}$ is the objective world of $v_1$, $p$ is true in $v_1$. Moreover, the sets $\possible{\agent}{v_1}$ and $\impossible{\agent}{v_1}$  provide some information on the beliefs of \agent in $v_1$. Recall that $\possible{\agent}{v_1}$ (resp. $\impossible{\agent}{v_1}$)   is the set of biworlds that have an extension \agent deems possible (resp. impossible). Since $\{p\}$ is the only biworld in $\possible{\agent}{v_1}$, $p$ is true in any biworld \agent deems possible. As we will see later, this means that \agent knows $p$. On the contrary, since $\impossible{\agent}{v_1}$ contains both $\{p\}$ and $\{\emptyset\}$, $p$ is true in some of the biworlds \agent deems impossible, and false in others.
	
	Starting with the objective world $v_0:=\{p\}$ and continuing with the biworld $v_1$, we can inductively build an \prebiworld{\omega} $v:=(v_\alpha)_{\alpha<\omega}$ as follows
	\vspace{-1.5mm}
	\begin{equation*}
		v_\alpha:=\begin{cases}
			\{p\} & \text{if }\alpha=0
			\\ (v_0, \{v_{\alpha'}\}, \WW^{\alpha'}) & \text{if }\alpha=\alpha'+1.
		\end{cases}
		\vspace{-1.5mm}
	\end{equation*}
	Since for all successor ordinals $\alpha=\alpha'+1<\omega$ we have $\possible{\agent}{v_{\alpha}}\cap\impossible{\agent}{v_\alpha}=\{v_{\alpha'}\}$, \extended{in order }to prove that $v$ is an \biworld{\omega}, it suffices to show   that for all $\alpha<\omega$, $v_\alpha$ is incompleted\conf{, which can be done by an inductive proof.}\extended{. We achieve this with a proof by induction: 
		\begin{compactitem}
			\item $\alpha=0$. The \prebiworld{1}s
			\begin{align*}
				u^1_0:=&(\{p\}, \{\{p\}\}, \{\{\emptyset\}\})
				\\u^2_0:=&(\{p\}, \{\{\emptyset\}\}, \{\{p\}\})
			\end{align*}
			are clearly distinct \biworld{1}s extending $v_0$. Hence $v_0$ is incompleted.
			\item $\alpha=\alpha'+1$. By induction hypothesis, $v_{\alpha'}$ is incompleted, \ie it has two distinct \biworld{\alpha} extensions $u^1_{\alpha'}$ and $u^2_{\alpha'}$. It is easy to see that
			\begin{align*}
				u^1_\alpha:=&(\{p\}, \{u^1_{\alpha'}\}, \WW^\alpha\setminus \{u^1_{\alpha'}\})
				\\u^2_\alpha:=&(\{p\}, \{u^2_{\alpha'}\}, \WW^\alpha\setminus \{u^2_{\alpha'}\})
			\end{align*}
			are two distinct \biworld{\alpha{+}1}s extending $v_\alpha$. Hence, $v_a$ is incompleted as desired.
		\end{compactitem}  
		We conclude that $v$ is a \biworld{\omega}.}  Moreover, $v$ 
	\conf{can be shown to be}\extended{is} completed. 
	\extended{Let us explain why this holds. 
	
	Taking into account the second item of Definition~\ref{def:new:restriction} and the incompletedness of each $v_\alpha$, one can easily see that $v$ has exactly two \prebiworld{\omega{+}1}s extending it, namely:
	\vspace{-1mm}
	\begin{align*}
		u^1:=(\{p\}, \{v\}, \WW^\omega\setminus \{v\}), \quad u^2:=(\{p\}, \{v\}, \WW^\omega).
	\end{align*}
	\vspace{-4.5mm}\\
	Since $\possible{\agent}{u^1}\cup\impossible{\agent}{u^1}=\WW^\omega$ and $\possible{\agent}{u^1}\cap\impossible{\agent}{u^1}=\emptyset$, $u^1$ is an \biworld{\omega{+}1}. However, $u^2$ is a biworld only if $v$ is incompleted, by \cref{case:def:biworld_successor} of \cref{def:new:biworld}. Hence, $u^2$ is a biworld if and only if $v$ is incompleted. Since we construct the concepts of incompleted \prebiworld{\omega} and \biworld{\omega{+}1} by simultaneous induction by choosing the minimal sets, $u^2$ is not a biworld and $v$ is not incompleted.}
	
\end{example}

We now dive deeper in certain features of (pre)biworlds with limit ordinal depth.
Note that for successor ordinals $\ordinal$, and  for any \prebiworld{\ordinal}, each 
agent is equipped with two sets of prebiworlds. 
For limit ordinals $\ordinal$, however, this is not the case.  
The following definition aims at retrieving a similar concept for a prebiworld with limit ordinal depth. 

\begin{definition}\label{def:new:arrow_sets}
	Given a limit ordinal $\limitordinal$, a \prebiworld{\limitordinal} $\world$ and an agent $\agent \in \A$, we define the following sets
	\vspace{-1mm}
	\begin{align*}
		\evalLim{\agent}{\world} := \{\otherworld \in \WW_p^\limitordinal \mid \forall \ordinal<\limitordinal: (\otherworld)_{\ordinal} \in \possible{\agent}{(\world)_{\ordinal+1}}\}\\
		\evalLimBar{\agent}{\world} := \{\otherworld \in \WW_p^\limitordinal \mid \forall \ordinal<\limitordinal: (v)_{\ordinal} \in \bar A^{(w)_{\ordinal+1}}\}
	\vspace{-1mm}
	\end{align*} 
\end{definition}

It is clear that if $\world$ in the definition above is a biworld, then $\evalLim{\agent}{\world}$ and $\evalLimBar{\agent}{\world}$ are sets of biworlds. It is important to notice that the sets defined in \cref{def:new:arrow_sets} do not carry the exact same meaning as the successor-ordinal counterpart: if $v\in\evalLim{\agent}{\world}$ for some $\lambda$-(pre)biworld \world, then in \world, for each approximation $v_\alpha$ ($\alpha<\lambda$) of $v$, the agent $\agent$ deems some extension of $v_\alpha$ as possible. Analogously, if $v\in\evalLimBar{\agent}{\world}$, then in \world, for each approximation $v_\alpha$ ($\alpha<\lambda$) of $v$, the agent $\agent$ believes some extension of $v_\alpha$ is impossible.

%
%


\section{Properties of Biworlds}\label{sec:properties}
In this section, we show that the formal definitions stated in Section \ref{sec:definitions} behave well and that they indeed correspond to the intuitive ideas introduced above. First, we present two propositions regarding the notion of restriction, which show that  the induced relation $\leqp$ is  a non-strict partial order. Second, we provide some additional insight into the sets defined in \cref{def:new:arrow_sets}.  Finally, we focus on certain properties concerning biworlds. In particular, we will show the following fundamental facts in \cref{thm:four_properties}:
\bart{I changed compactenum to compactdesc. Saves us a line. But... I'm not super happy with it so feel free to revert if we have a spare lien}
\begin{compactdesc}
	\item[Restrictability:] the restriction of a biworld is a biworld.
	\item[Monotonicity of completedness:] an extension of a completed biworld is a completed biworld. 
	\item[Completability:] all biworlds have a a completed extension.
	\item[Completedness at successor ordinals:] a \biworld{\ordinal{+}1} $w$ is completed if and only if $\possible{\agent}{\world}\cap\impossible{\agent}{\world}=\emptyset$ for all agents $\agent$. 
\end{compactdesc}
Given Definitions \ref{def:new:restriction} and \ref{def:new:incompleted}, the first two properties are sensible to have. Completability will be fundamental in Section \ref{sec:twoValuations} and it is clearly desirable: a complete biworld characterizes complete knowledge, providing a full description of the epistemic state of the agents. 
The last property provides a  simple characterization of what it means to be completed for biworlds of successor ordinal depth. 
While the four properties listed above may seem natural and straightforward, several intermediate results are required to prove them to hold. 


As anticipated, we start by showing that Definition \ref{def:precision} defines a non-strict partial order on the set of prebiworlds.

\begin{proposition} \label{prop:restriction}
	Let $\world$ be a \prebiworld{\ordinal} and let $\evensmallerordinal\leq\smallerordinal\leq\ordinal$. Then $\restr{\world}{\ordinal}=\world$ and $\restr{\restr{\world}{\smallerordinal}}{\evensmallerordinal}=\restr{\world}{\evensmallerordinal}$.
\end{proposition}
\extended{
\begin{proof}
	We first prove $\restr{\world}{\ordinal}=\world$ by induction on $\ordinal$.
	
	\begin{compactenum}
		\item $\ordinal=0$. Then $\restr{\world}{\ordinal}=\objworld{\world}=\world$ by Definitions \ref{def:new:restriction} and \ref{def:new:prebiworld}.
		
		\item $\ordinal=\ordinal'+1$ is a successor ordinal. By induction hypothesis, we have that   $\restr{\world'}{\ordinal'}=\world'$ for each \ordinal'-prebiworld $\world'$. By Definition \ref{def:new:restriction}, we conclude this case.
		\item $\ordinal$ is a limit ordinal. Then $\restr{\world}{\ordinal}=(\restr{\world}{\smallerordinal})_{\smallerordinal<\ordinal}$, and we need to show $\restr{\world}{\smallerordinal}=\world_\smallerordinal$. We proceed by induction on $\smallerordinal$.
		\begin{compactenum}
			\item If $\smallerordinal=0$ or $\smallerordinal$ is a successor ordinal, then we conclude by Definition \ref{def:new:restriction}.
			\item If $\smallerordinal$ is a limit ordinal, then $\restr{\world}{\smallerordinal}=(\restr{\world}{\evensmallerordinal})_{\evensmallerordinal<\smallerordinal}=(\world_\evensmallerordinal)_{\evensmallerordinal<\smallerordinal}=\world_\smallerordinal$, where the equalities hold by Definition \ref{def:new:restriction}, induction hypothesis on $\smallerordinal$, and Definition \ref{def:new:prebiworld}, respectively.
		\end{compactenum}
	\end{compactenum}
	Now we show that $\restr{\restr{\world}{\smallerordinal}}{\evensmallerordinal}=\restr{\world}{\evensmallerordinal}$ by induction on $\ordinal$.
	\begin{compactenum}
		\item $\ordinal=\smallerordinal$. This is a consequence of the already proven first claim of this proposition.
		\item $\ordinal=\ordinal'+1$ is a successor ordinal and $\smallerordinal<\ordinal$. We proceed by induction on $\evensmallerordinal$.
		\begin{compactenum}
			\item $\evensmallerordinal=0$. Then $\restr{\world}{\evensmallerordinal}=\objworld{\world}$, $\restr{\restr{\world}{\smallerordinal}}{\evensmallerordinal}=\objworld{\restr{\world}{\smallerordinal}}$, and it is easy to see that $\objworld{\restr{\world}{\smallerordinal}}=\objworld{\world}$.
			\item $\evensmallerordinal=\evensmallerordinal'+1$ is a successor ordinal. Then $\objworld{\restr{\restr{\world}{\smallerordinal}}{\evensmallerordinal}}=\objworld{\restr{\world}{\smallerordinal}}=\objworld{\world}=\objworld{\restr{\world}{\evensmallerordinal}}$. Moreover, if $\smallerordinal=\smallerordinal'+1$ is a successor ordinal, then for each agent $\agent\in\A$
			
			\begin{align*}
				\possible{\agent}{\restr{\restr{\world}{\smallerordinal}}{\evensmallerordinal}}&=\left\{\restr{\world'}{\evensmallerordinal'}\mid \world'\in\possible{\agent}{\restr{\world}{\smallerordinal}}\right\}
				\\&=\left\{\restr{\restr{\world'}{\smallerordinal'}}{\evensmallerordinal'}\mid \world'\in\possible{\agent}{\world}\right\}
				\\&=\left\{\restr{\world'}{\evensmallerordinal'}\mid \world'\in\possible{\agent}{\world}\right\}=\possible{\agent}{\restr{\world}{\evensmallerordinal}},
			\end{align*}
			
			where the third equality holds by induction hypothesis in $\ordinal$. Analogously, we have $\impossible{\agent}{\restr{\restr{\world}{\smallerordinal}}{\evensmallerordinal}}=\impossible{\agent}{\restr{\world}{\evensmallerordinal}}$. If $\smallerordinal$ is a limit ordinal, then $\restr{\restr{\world}{\smallerordinal}}{\evensmallerordinal}=\restr{\left(\restr{\world}{\gamma}\right)_{\gamma<\smallerordinal}}{\evensmallerordinal}=\restr{\world}{\evensmallerordinal}$.
			
			\item $\evensmallerordinal$ is a limit ordinal. Then $\restr{\restr{\world}{\smallerordinal}}{\evensmallerordinal}=\left(\restr{\restr{\world}{\smallerordinal}}{\gamma}\right)_{\gamma<\evensmallerordinal}=\left(\restr{\world}{\gamma}\right)_{\gamma<\evensmallerordinal}=\restr{\world}{\evensmallerordinal}$, where the second equality holds by the induction hypothesis on $\evensmallerordinal$.
		\end{compactenum}
		
		\item $\ordinal$ is a limit ordinal. We can proceed by nested inductions on $\evensmallerordinal$ and $\smallerordinal$, as we did above. The proof is exactly the same as for $\ordinal$ being a successor ordinal, except for the case where both $\evensmallerordinal$ and $\smallerordinal$ are successor ordinals. In this case we have $\restr{\restr{\world}{\smallerordinal}}{\evensmallerordinal}=\restr{\world_\smallerordinal}{\evensmallerordinal}=\world_\evensmallerordinal=\restr{\world}{\evensmallerordinal}$, where the second equality holds because $\world_\evensmallerordinal\leqp\world_\smallerordinal$.
	\end{compactenum}
\end{proof}
}

\begin{proposition}
	The relation $\leqp$ is a non-strict partial order (a reflexive, antisymmetric and transitive relation). 
	The induced strict order $<_p$ is a well-founded order. 
\end{proposition}


Before proceeding to the properties of biworlds, we focus on the sets of \cref{def:new:arrow_sets}. Analogously to what happens in \cref{case:def:biworld_successor} of \cref{def:new:biworld} for biworlds of successor ordinal depth, the union of $\evalLim{\agent}{\world}$ and $\evalLimBar{\agent}{\world}$ is the whole set of biworlds of that depth.

\begin{proposition}\label{prop:new:limitworld_union}
	Let $\limitordinal$ be a limit ordinal, and $\world$ be a \biworld{\limitordinal}. For all $\agent\in\A$, we have  $\evalLim{\agent}{\world}\cup\evalLimBar{\agent}{\world}=\WW^\limitordinal$. 
\end{proposition}
\extended{
	\begin{proof}
		Let $\agent \in\A$. Clearly, $\evalLim{\agent}{\world}\cup\evalLimBar{\agent}{\world}\subseteq\WW^\limitordinal$. To show the other inclusion, let $\world'\in\WW^\limitordinal\setminus \bar\agent\uparrow\world$. Then there exists $\ordinal'<\limitordinal$ such that $(\world')_{\ordinal'}\notin\impossible{\agent}{(\world)_{\ordinal'+1}}$. Since \world and $\world'$ are biworlds, we must have $(\world')_{\ordinal'}\in\possible{\agent}{(\world)_{\ordinal'+1}}$.  By the precision order of limit biworlds, for all $\ordinal<\ordinal'$, we get $(\world')_{\ordinal}\in\possible{\agent}{(\world)_{\ordinal+1}}$, and for all $\ordinal>\ordinal'$ we get $(\world')_{\ordinal}\notin\impossible{\agent}{(\world)_{\ordinal+1}}$. In particular, this implies $(\world')_{\ordinal}\in\possible{\agent}{(\world)_{\ordinal+1}}$ for all $\ordinal>\ordinal'$. Hence, $\world'\in\evalLim{\agent}{\world}$ as desired.
	\end{proof}
}

It is interesting to notice that the second part of  \cref{case:def:biworld_successor} of \cref{def:new:biworld} does not hold in general for the sets $\evalLim{\agent}{\world}$ and $\evalLimBar{\agent}{\world}$, \ie their intersection might contain completed biworlds (see \cref{ex:intersection_with_completed}). Nevertheless, this intersection tells us something about the completedness of $\world$, as stated in the following proposition.

\begin{proposition}\label{prop:empty_uparrow_then_completed}\label{prop:empty_then_completed_limit}
	Let $\limitordinal$ be a limit ordinal and $\world$ a \biworld{\limitordinal}. If $\evalLim{\agent}{\world}\cap\evalLimBar{\agent}{\world}=\emptyset$ for all $\agent\in\A$, then \world is completed.
\end{proposition}
\extended{
	\begin{proof}
		Suppose $w' \geqp w$ is a $\lambda+1$-biworld. Then $\objworld{w} = (w)_0$. Furthermore,for each $\smallerordinal<\limitordinal$ and for each $\agent\in\A$, we have $\possible{\agent}{\world_{\smallerordinal+1}}=\possible{\agent}{\restr{\world'}{\smallerordinal+1}}$ and  $\impossible{\agent}{\world_{\smallerordinal+1}}=\impossible{\agent}{\restr{\world'}{\smallerordinal+1}}$. Hence, for each $\agent\in\A$, $\possible{A}{w'} \subseteq \evalLim{\agent}{\world}$ and $\impossible{A}{w'} \subseteq \evalLimBar{\agent}{\world}$. But since $\evalLim{\agent}{\world}\cap\evalLimBar{\agent}{\world}=\emptyset$ and $\possible{\agent}{\world'}\cup\impossible{\agent}{\world'}=\WW^\limitordinal$, the only way this can hold is when $\possible{A}{w'} = \evalLim{\agent}{\world}$ and $\impossible{A}{w'} = \evalLimBar{\agent}{\world}$. 
		This fully determines $w'$, so there cannot be two distinct $\lambda+1$-biworlds $w',w'' \geqp w$.
	\end{proof}
}

\begin{example}\label{ex:intersection_with_completed}
	Consider the completed \biworld{\omega} $v$ defined in \cref{ex:compl_notpostcompl_omegabiw}. It is easy to see that $\evalLim{\agent}{v}=\{v\}$ and $\evalLimBar{\agent}{v}=\WW^\omega$. Hence, we get $\evalLim{\agent}{v}\cup \evalLimBar{\agent}{v}=\WW^\omega$ and $\evalLim{\agent}{v}\cap\evalLimBar{\agent}{v}=\{v\}$, which agrees with \cref{prop:new:limitworld_union} and shows that the converse of \cref{prop:empty_then_completed_limit} does not hold. Moreover, notice that $\evalLim{\agent}{v}\cap\evalLimBar{\agent}{v}$ contains a completed biworld.
\end{example}

\extended{
We present one last result regarding the sets defined in Definition \ref{def:new:arrow_sets}, that turns out to be fundamental to prove the lemmas leading to \cref{thm:four_properties}. 

\begin{proposition}\label{prop:ext_to_limit}
	Let $\limitordinal$ be a limit ordinal, $w$ be a \biworld{\limitordinal}, $\agent\in\A$ be an agent, and $\ordinal<\limitordinal$ be an ordinal. If $\otherworld\in\possible{\agent}{(w)_{\ordinal+1}}$, then there exists $\otherworld'\in\evalLim{\agent}{\world}$ such that $\otherworld'\geqp\otherworld$.  If $\otherworld\in\impossible{\agent}{(w)_{\ordinal+1}}$, then there exists $\otherworld'\in\evalLimBar{\agent}{\world}$ such that $\otherworld'\geqp\otherworld$.
\end{proposition}
	\begin{proof}
		Recall that we only consider countable ordinals. We prove the part of the statement with $\otherworld\in\possible{\agent}{(w)_{\ordinal+1}}$, and the second part for $\otherworld\in\impossible{\agent}{(w)_{\ordinal+1}}$ is analogous.
		By the precision order of limit biworlds, for each ordered pair $(\beta, \gamma)$ of ordinals with $\evensmallerordinal<\gamma<\limitordinal$ and for each $u\in\possible{\agent}{\world_{\evensmallerordinal+1}}$, there exists $u'\in\possible{\agent}{\world_{\gamma+1}}$ such that $u'\geqp u$. 
		So by the axiom of choice, there exists a family of functions $f_A^\gamma\colon W_A^\gamma\to \possible{\agent}{\world_{\gamma+1}}$ for $\gamma<\limitordinal$, where $W_A^\gamma:=\bigcup_{\beta<\gamma}\possible{\agent}{\world_{\beta+1}}$, such that for each $u\in W_A^\gamma$, we have $f_A^\gamma(u)\geqp u$.
		Since \limitordinal is a countable limit ordinal, 
		 we can choose a strictly increasing sequence $(\rho_n)_{n<\omega}$ of ordinals which has limit \limitordinal. W.l.o.g.\ we assume $\rho_0>\ordinal$.
		Then, we define $\tilde{v}\in\WW^\smallerordinal$ as follows: 
		\begin{align}\label{going-up_construction}
			\tilde{v}_\evensmallerordinal:=
			\begin{cases}
				\restr{v}{\evensmallerordinal} & \text{ if }\evensmallerordinal\leq\ordinal
				\\  \restr{f_A^{\rho_0}(v)}{\evensmallerordinal} & \text{ if }\ordinal<\evensmallerordinal\leq \rho_0
				\\ \restr{f_A^{\rho_{n+1}}(\tilde{v}_{\rho_n})}{\evensmallerordinal} & \text{ if }\rho_n<\evensmallerordinal\leq \rho_{n+1}
			\end{cases}.
		\end{align}
		By the definition of the functions $f_A^{\rho_i}$ and of restriction of biworlds, it is clear that $\tilde{v}\in\evalLim{\agent}{\world}$ and $\tilde{v}\geqp v$.
	\end{proof}
}
\extended{
The remaining definitions capture the four properties listed at the beginning of this section, but only at limited depth. They are used to inductively prove \cref{lem:well-behaved}. 

The first property is \restrictability, \ie whenever we take the restriction of a biworld, the result is also a biworld. A priori, it is only known that this is a prebiworld. 

\begin{definition}\label{def:restrictability_holds}
	Let $\mu$ be an ordinal. We say that \emph{\restrictabilityHoldsAt[\mu]} ($\mathrm{R}_\ordinal$) if for every \biworld{\ordinal} $\world$, and every $\alpha \leq \mu$, $\restr{\world}{\alpha}$ is again a biworld.
\end{definition}

The second property is \monotonicityCompl, which states that any extension of a completed biworld is itself completed. Intuitively, we expect this to hold: if one biworld has complete knowledge, refining it can only result in structures that still have complete knowledge. 

\begin{definition}\label{def:monotonicity_compl_holds}
	Let $\mu$ be an ordinal. We say that \emph{\monotonicityComplHoldsAt[\mu]} ($\mathrm{MoC}_\ordinal$) if for each \biworld{\ordinal} $\world$ and each ordinal $\alpha \leq \mu$, whenever $\world|_\alpha$ is completed,  so is $\world$.  
\end{definition}

The third useful property is \completability. It states that any biworld can always be extended to a completed one. This is in line with or view of incompleted biworlds as approixmations of some completed biworlds. 
\begin{definition}\label{def:completability_holds}
	Let $\mu$ be an ordinal. We say that \emph{\completabilityHoldsAt[\mu]} ($\mathrm{C}_\ordinal$) if for every $\alpha<\mu$, and every \biworld{\smallerordinal} \world, there exists a completed \biworld{\ordinal} $\world'$ with $\world'\geqp\world$. 
\end{definition}

The fourth and last property is a characterization of completedness for a biworld whose depth is a successor ordinal. 

\begin{definition}\label{def:compl_condition_holds}
	Let $\mu$ be a successor ordinal. We say that \emph{\complCondHoldsAt[\mu]} ($\mathrm{CC}_\ordinal$) if for each \biworld{\ordinal} $\world$, $\world$ is completed if and only if for each agent $A$, $A^w\cap \bar{A}^w = \emptyset$. 
\end{definition}


\begin{definition}
	We say that biworlds are \emph{well-behaved until $\mu$} if restrictability, completability and monotonicity of completedness hold at all ordinals smaller than $\mu$ and the completedness condition holds at all successor ordinals smaller than $\pl{\mu}$, where $\pl{\mu}$ denotes the largest limit ordinal smaller or equal to $ \mu$.
\end{definition}


In order to prove that every biworld enjoys the four properties listed at the beginning of this section, we need to prove \cref{lem:well-behaved} below.}
\conf{The following theorem formally states the fundamental properties of biworlds presented at the beginning of this section.}
The proof of it requires several intermediate technical lemmas\conf{, which we omit due to space limitations}\extended{, whose interplay is sketched in Figure~\ref{fig:relations_lemmas}}.
\ignore{Nevertheless, we provide a sketch of the proof.}

\extended{
\begin{figure}[t]
	\centering
	\vspace{-3mm}
	\includegraphics[width=0.8\columnwidth]{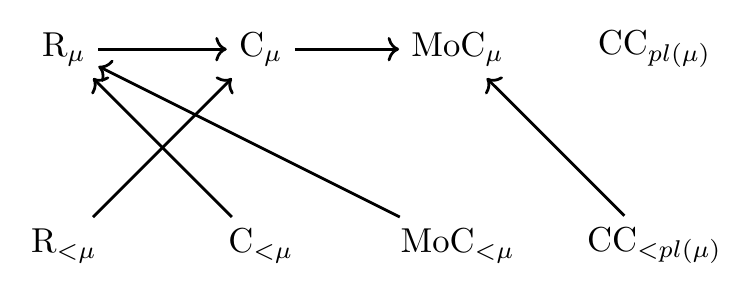}
	\vspace{-4mm}
	\caption{The figure depicts the relations \conf{between the concepts introduced in the sketch of the proof of \cref{thm:four_properties}, and proved in the omitted lemmas}\extended{expressed in the Lemmas \ref{lem:restrict}, \ref{lem:completability}, \ref{lem:monotonicity}, and \ref{lem:cc}}. By $\mathrm{R}_{<\mu}$, we denote that $\mathrm{R}_\alpha$ holds for all $\alpha<\ordinal$; analogously for $\mathrm{C}_{<\mu}$, and $\mathrm{MoC}_{<\mu}$. By $\mathrm{CC}_{<pl(\mu)}$ we denote that $\mathit{CC}_\alpha$ holds for all successor ordinals $\alpha<pl(\mu)$. A concept is true if the bodies of all the arrows pointing at it are true. For example: if $\mathrm{C}_{<\mu}$ and $\mathrm{MoC}_{<\mu}$, then $\mathrm{R}_{\mu}$. For the limit ordinal case, we need an additional implication not represented in this figure for the sake of simplicity, namely $\mathrm{C}_{\mu+2} \implies \mathrm{CC}_{\mu+1}$ (Lemma~\ref{lem:cc}).}
	\label{fig:relations_lemmas}
	\vspace{-4mm}
\end{figure}
}

\extended{
\begin{lemma}\label{lem:empty_then_completed}
	Let $\world$ be a \biworld{\ordinal{+}1}. If $\possible{\agent}{\world}\cap\impossible{\agent}{\world}=\emptyset$ for all $\agent\in\A$, then \world is completed.
\end{lemma}
\begin{proof}
	By contradiction, suppose there exist \biworld{\ordinal{+}2}s $\otherworld_1$ and $\otherworld_2$, such that $\otherworld_1\neq\otherworld_2$ and $\otherworld_1,\otherworld_2\geqp\world$. Then, there exists $\agent\in\A$ such that either $\possible{\agent}{\otherworld_1}\neq\possible{\agent}{\otherworld_2}$ or $\impossible{\agent}{\otherworld_1}\neq\impossible{\agent}{\otherworld_2}$. W.l.o.g., we can assume that $\possible{\agent}{\otherworld_1}\neq\possible{\agent}{\otherworld_2}$ and that there  exists $u\in\possible{\agent}{\otherworld_1}$ such that $u\notin\possible{\agent}{\otherworld_2}$. Since $\otherworld_2$ and $u$ are biworlds, $u$ must be in $\impossible{\agent}{\otherworld_2}$. Then, by definition of restriction, we have that $\restr{u}{\ordinal}\in \possible{\agent}{\restr{\otherworld_1}{\ordinal+1}}\cap\impossible{\agent}{\restr{\otherworld_2}{\ordinal+1}}=\possible{\agent}{\world}\cap\impossible{\agent}{\world}=\emptyset$, a contradiction.
\end{proof}
}

\extended{
\begin{lemma}
	\label{lem:compl}
	Suppose \monotonicityCompl holds at all ordinals up to and including $\mu$. Suppose $\alpha \leq \mu$ and $w$ is a completed \biworld{\alpha}. Then there is at most one \biworld{\ordinal} $w' \geqp w$. Moreover, if there is a \biworld{\ordinal} $w' \geqp w$, $w'$ is completed.
\end{lemma}

\begin{proof}
	Fix $\alpha$. We prove this lemma by induction over $\mu$.
	\begin{compactitem}
		\item $\mu = \alpha$. Trivial.
		\item $\mu = \mu'+1$ is a successor ordinal greater than $\alpha$. The induction hypothesis is that there is at  most one \biworld{\ordinal'} $v \geqp w$, and any such $v$ is completed. Suppose for a contradiction that there are two distinct \biworld{\ordinal}s $w',w'' \geqp w$. Then $w'|_{\mu'}$ and $w''|_{\mu'}$ are \biworld{\ordinal}s, so by the induction hypothesis $w'|_{\mu'} = w''|_{\mu'}$ and $w'|_{\mu'}$ is complete, which contradicts Definition~\ref{def:new:incompleted}. So there is at most one \biworld{\mu} $w' \geqp w$. Now assume that there is some \biworld{\mu} $w' \geqp w$. Since $w$ is completed and $w'|_\alpha = w$, the monotonicity of completedness for $\mu$ implies that $w$ is completed.
		\item $\mu$ is a limit ordinal. By induction hypothesis, for every $\beta<\mu$ with $\beta \geq \alpha$, there is at most one \biworld{\beta} $w_\beta \geqp w$, and any such $w_\beta$ is completed. Suppose for a contradiction that there are two distinct \biworld{\mu} s$w', w'' \geqp w$. Then for some ordinal $\beta$ with $\alpha < \beta < \mu$, $w''|_\beta \neq w'|_\beta$, contradicting the induction hypothesis (since $w''|_\beta, w'|_\beta \geqp w$). So there is at most one $\mu$-biworld $\geqp w$. Assume that $w' \geqp w$ is such a $\mu$-biworld. Since $w'|_\alpha = w$ and $w$ is completed, the monotonicity of completedness for $\mu$ implies that $w$ is completed.\qedhere
	\end{compactitem}
\end{proof}
}

\extended{
\begin{lemma}
	\label{lem:restrict}
	If monotonicity of completedness and completability hold at all ordinals strictly smaller than $\mu$, then restrictability holds at $\mu$. 
\end{lemma}
\begin{proof}
	We prove this by induction over $\mu$. The induction hypothesis $\IH_0$ is that restrictability holds at every ordinal $\beta < \mu$.
	\begin{compactitem}
		\item $\mu = 0$. Trivial.
		\item $\mu=\mu'+1$ is a successor ordinal. Let $w$ be a $\mu$-biworld. We need to show that for every $\alpha \leq \mu$, $w|_\alpha$ is a biworld. We prove this by induction over $\alpha$. The induction hypothesis $\IH_1$ is that for any $\alpha' < \alpha$, $w|_{\alpha'}$ is a biworld.
		\begin{compactitem}
			\item $\alpha = 0$. $w|_0$ is a biworld by Definition~\ref{def:new:biworld}. 
			\item $\alpha=\alpha'+1$ is a successor ordinal. First we show that for each agent $\agent\in\A$, $\possible{\agent}{\restr{\world}{\alpha}}\cap \impossible{\agent}{\restr{\world}{\alpha}}$ does not contain completed biworlds. For this, suppose $v \in \possible{\agent}{w|_\alpha} \cap \impossible{\agent}{w|_\alpha}$. By Definition~\ref{def:new:restriction}, there are worlds $w' \in \possible{\agent}{w}$ and $w'' \in \impossible{\agent}{w}$ such that $w'|_\alpha = w''|_\alpha = v$. Suppose for a contradiction that $v$ is completed. Since monotonicity of completedness holds for all ordinals $\leq \mu'$, Lemma~\ref{lem:compl} implies that $w' = w''$ and $w'$ is completed. But $w'=\world'' \in \possible{\agent}{w}\cap \impossible{\agent}{w}=\emptyset$ and $\world'$ being completed contradict the fact that $\world$ is a biworld.\\
			What remains to be shown is that $\possible{\agent}{w|_\alpha} \cup \impossible{\agent}{w|_\alpha} = \WW^{\alpha'}$.\\
			$\subseteq$: Assume $v \in \possible{\agent}{w|_\alpha} \cup \impossible{\agent}{w|_\alpha}$. Without loss of generality, assume $v \in \possible{\agent}{w|_\alpha}$. Since $\mu$ is a successor ordinal, $\possible{\agent}{w|_\alpha} = \{w'|_{\alpha'} \mid w' \in \possible{\agent}{w} \}$. So there is some $w' \in \possible{\agent}{w}$ such that $v = w'|_{\alpha'}$. By Definition~\ref{def:new:biworld}, the fact that $w' \in \possible{\agent}{w}$ implies that $w'$ is a biworld. By induction hypothesis $\IH_0$, restrictability holds at $\mu'$. Since $w'$ is a $\mu'$-biworld and $v = w'|_{\alpha'}$, it follows that $v$ is a biworld, as required.\\
			$\supseteq$: Assume $v \in \WW^{\alpha'}$, \ie $v$ is an $\alpha'$-biworld. Since completability holds at $\mu'$, there is a $\mu'$-biworld $v' \geq v$. Since $w$ is a $\mu$-biworld, $v' \in \possible{\agent}{w} \cup \impossible{\agent}{w}$. Without loss of generality, assume $v' \in \possible{\agent}{w}$. Then $v = V'|_\alpha \in \possible{\agent}{w|_\alpha}$, as required.
			\item $\alpha$ is a limit ordinal. By induction hypothesis $\IH_1$, we have that for any $\alpha' < \alpha$, $w|_{\alpha'}$ is a biworld. By Definition~\ref{def:new:restriction}, $w|_\alpha = (w|_{\alpha'})_{\alpha'<\alpha}$. So by Definition~\ref{def:new:biworld}, $w|_\alpha$ is a biworld, as required.
		\end{compactitem}
		\item $\mu$ is a limit ordinal. This case directly follows from Definition~\ref{def:new:biworld}. \qedhere
	\end{compactitem}
\end{proof}
}

\extended{
\begin{lemma}\label{lem:monotonicity}
	If completability holds at $\mu$ and the completedness condition holds at all successor ordinals smaller than $\pl{\mu}$, then monotonicity of completedness holds at $\mu$. 
\end{lemma}
\begin{proof}
	Fix $\alpha$. We need to prove that for each $\mu \geq \alpha$ and each $\mu$-biworld $\world$ such that $\world|_\alpha$ is completed, $\world$ is also completed. We prove this lemma by induction over $\mu$.
	\begin{compactitem}
		\item $\mu = \alpha$. Trivial.
		\item $\mu = \mu'+1$ is a successor ordinal. Let $w$ be a \biworld{\mu} such that $w|_\alpha$ is completed. By induction hypothesis, $w|_{\mu'}$ is completed. Since completability holds at $\mu$, there is a completed \biworld{\mu} $w' \geqp w|_{\mu'}$. Since $w,w' \geqp w|_{\mu'}$ and $w|_{\mu'}$ is completed, $w=w'$. So $w$ is completed, as required.
		\item $\mu$ is a limit ordinal. Let $w$ be a $\mu$-biworld such that $w|_\alpha$ is completed. By induction hypothesis, $w|_{\alpha+1}$ is completed. Since the completedness condition holds at $\alpha+1$, $\possible{\agent}{w|_{\alpha+1}} \cap \possible{\agent}{w|_{\alpha+1}} = \emptyset$. Now we can show by a proof by contradiction that $A \uparrow w \cap \bar A \uparrow w = \emptyset$: If $v \in A \uparrow w \cap \bar A \uparrow w$, then $v|_\alpha \in \possible{\agent}{w|_{\alpha+1}} \cap \possible{\agent}{w|_{\alpha+1}}$, a contradiction. Thus by \cref{prop:empty_then_completed_limit}, $w$ is completed, as required.
	\end{compactitem}
\end{proof}
}

\extended{
\begin{lemma}\label{lem:completability}
	Let  \ordinal be an ordinal. If restrictability holds at each $\smallerordinal\leq\ordinal$, then  \completabilityHoldsAt. 
\end{lemma}
\begin{proof}
	First, we show that for any $\smallerordinal<\ordinal$ and for each \smallerordinal-biworld \world there exists a completed \biworld{\smallerordinal{+}1} $\world'\geqp\world$. 
	
	We proceed by induction on \smallerordinal.
	\begin{compactenum}
		\item $\smallerordinal=0$. We take $\world':=\left(\world, (\WW^0)_{\agent\in\A}, (\emptyset)_{\agent\in\A} \right)$. Clearly, $\world'$ is a \biworld{1} extending \world. By Lemma \ref{lem:empty_then_completed}, $\world'$ is a completed \biworld{1}.
		
		\item $\smallerordinal=\smallerordinal'+1$ is a successor ordinal.  Fix $\agent\in\A$. Since $\world$ is a biworld, every biworld $\otherworld$ in $\possible{\agent}{\world}\cap\impossible{\agent}{\world}$ is incompleted, \ie there exist distinct \biworld{\smallerordinal} $\otherworld_1,\otherworld_2 \geqp\otherworld$. By the axiom of choice, we can define a function $f_A\colon \possible{\agent}{\world}\cap\impossible{\agent}{\world} \to \WW^\smallerordinal $ by sending each $\otherworld$ in $\possible{\agent}{\world}\cap\impossible{\agent}{\world}$  to an extension $f_A(\otherworld)\in\WW^\smallerordinal$. Let $S_A:=\{f_A(\otherworld)\mid\otherworld\in\possible{\agent}{\world}\cap\impossible{\agent}{\world}\}$. 
		Notice that, since biworlds in $\possible{\agent}{\world}\cap\impossible{\agent}{\world}$ are incompleted,  each \biworld{\smallerordinal'} $\otherworld\in\possible{\agent}{\world}\cap\impossible{\agent}{\world}$ has an extension in $S_A$ and an extension not in $S_A$.
		Then we take 
		\begin{align*}
			\world':=\Big(\objworld{\world}, &(\ext(\possible{\agent}{\world})\setminus S_A)_{\agent\in\A}, \\&((\WW^\smallerordinal\setminus \ext(\possible{\agent}{\world}))\cup S_A)_{\agent\in\A}\Big),
		\end{align*}
		where 
		$\ext(\possible{\agent}{\world}) := \{ u \in\WW^{\alpha} \mid \restr{u}{\alpha'}\in \possible{\agent}{\world} \}$.
		 It is easy to see that $\world'$ is an \biworld{\smallerordinal{+}1}, and by Lemma \ref{lem:empty_then_completed}, $\world'$ is completed. It remains to show that $\world'\geqp\world$. For each agent $\agent \in \A$, we need to prove two equalities:
		\begin{compactenum}
			\item $\possible{\agent}{\restr{\world'}{\smallerordinal}}=\possible{\agent}{\world}$. 
			
			Let $\otherworld\in\possible{\agent}{\restr{\world'}{\smallerordinal}}$. Then, by Defintion \ref{def:new:restriction}, there exists $\otherworld'\in\possible{\agent}{\world'}\subseteq\ext(\possible{\agent}{\world})$ such that $\otherworld'\geqp\otherworld$. Hence, $\otherworld=\restr{\otherworld'}{\smallerordinal'}\in\possible{\agent}{\world}$, as desired. 
			Finally, for $\possible{\agent}{\world}\subseteq\possible{\agent}{\restr{\world'}{\smallerordinal}}$, let $\otherworld\in\possible{\agent}{\world}$. By induction hypothesis, there exists an \smallerordinal-biworld $\otherworld'$ such that $\otherworld'\geqp\otherworld$. If  $\otherworld\notin \possible{\agent}{\world}\cap\impossible{\agent}{\world}$, then by uniqueness of restriction $\otherworld'\in\ext(\possible{\agent}{\world})\setminus S_A$. Hence $\otherworld=\restr{\otherworld'}{\smallerordinal}\in \possible{\agent}{\restr{\world'}{\smallerordinal}}$, as desired. If  $\otherworld\in \possible{\agent}{\world}\cap\impossible{\agent}{\world}$, then there exists $\otherworld''\in \ext(\possible{\agent}{\world})$, such that $\otherworld''\geqp\otherworld$ and $\otherworld''\neq f_A(\otherworld)$. Hence, $\otherworld=\restr{\otherworld''}{\smallerordinal'}\in\possible{\agent}{\restr{\world'}{\smallerordinal}}$, as desired. 
			
			\item $\impossible{\agent}{\restr{\world'}{\smallerordinal}}=\impossible{\agent}{\world}$. 
			
			Let $\otherworld\in\impossible{\agent}{\restr{\world'}{\smallerordinal}}$. Then, by Defintion \ref{def:new:restriction}, there exists $\otherworld'\in\impossible{\agent}{\world'}$ such that $\otherworld'\geqp\otherworld$. If $\otherworld'\in S_A$, then $\otherworld=\restr{\otherworld'}{\smallerordinal'}\in\impossible{\agent}{\world}$, as desired. Otherwise, $\otherworld'\in\WW^\smallerordinal\setminus \ext(\possible{\agent}{\world})$ and $\otherworld=\restr{\otherworld'}{\smallerordinal'}\notin \possible{\agent}{\world}$. Since restrictability holds at level \smallerordinal, $\otherworld=\restr{\otherworld'}{\smallerordinal'}$ is a biworld. This together with the facts that $\world$ is a biworld and that $\otherworld\notin \possible{\agent}{\world}$ implies that $\otherworld\in\impossible{\agent}{\world}$, as desired. Hence, $\impossible{\agent}{\restr{\world'}{\smallerordinal}}\subseteq\impossible{\agent}{\world}$. For the other inclusion, let $\otherworld\in \impossible{\agent}{\world}$. If $\otherworld\in\possible{\agent}{\world}$, then $\otherworld=\restr{f_A(\otherworld)}{\smallerordinal'}\in\impossible{\agent}{\restr{\world'}{\smallerordinal}}$, as desired. If $\otherworld\notin\possible{\agent}{\world}$, then by the induction hypothesis there exists $\otherworld'\in \WW^\smallerordinal\setminus\ext(\possible{\agent}{\world})$ such that $\otherworld'\geqp\otherworld$. Hence, $\otherworld=\restr{\otherworld'}{\smallerordinal'}\in\impossible{\agent}{\restr{\world'}{\smallerordinal}}$.
			
		\end{compactenum}
		
		\item $\smallerordinal$ is a limit ordinal.  Fix $A \in \A$. By induction hypothesis, for each $\otherworld\in \evalLim{\agent}{\world}\cap\evalLimBar{\agent}{\world}$ and for each $\evensmallerordinal<\smallerordinal$, there exists a completed \biworld{\evensmallerordinal{+}1}  $\otherworld_{\evensmallerordinal+1}^\ast\geqp(\otherworld)_\evensmallerordinal$. Using this notation, we define some sets of biworlds to construct the desired completed biworld $\world'$:
		\begin{align*}
			S_0^A:=&\evalLim{\agent}{\world}\setminus\evalLimBar{\agent}{\world}
			\\ S^A_1:=&\left\{\otherworld\in \evalLim{\agent}{\world}\cap\evalLimBar{\agent}{\world}\middle|\  \begin{array}{l}
				\exists \delta<\smallerordinal \textnormal{ s.t.\ for every }$\evensmallerordinal$
				\\ \textnormal{ s.t.\ $\delta < \evensmallerordinal < \smallerordinal$},\\
				\otherworld_{\evensmallerordinal+1}^\ast \in \impossible{\agent}{\world_{\evensmallerordinal+2}}
			\end{array}\right\}. 
		\end{align*} 
		Now define $w' := (\objworld{\world},(S^A_0 \cup S^A_1)_{A \in \A},(\WW^\smallerordinal \setminus (S^A_0 \cup S^A_1))_{A \in \A})$. 
		Clearly $w'$ is a completed $\smallerordinal+1$-biworld. In order to show that $w \leqp w'$, one needs to prove that for all $\delta < \smallerordinal$, $(w)_\delta \leqp w'$. 
		For this it is enough to show that for all successor ordinals $\delta+1 < \smallerordinal$, $(w)_{\delta+1} \leqp w'$. In more detail, we need to show that for every $\delta < \smallerordinal$ and every $\agent\in \A$, $\agent^{(\world)_{\delta+1}} =\{(\otherworld)_\delta\mid\otherworld\in S^\agent_0 \cup S^\agent_1\}$ 
		and $\impossible{\agent}{(\world)_{\delta+1}} =\{(\otherworld)_\delta \mid \otherworld\in \WW^\smallerordinal \setminus (S^\agent_0 \cup S^\agent_1)\}$. So fix $\delta< \smallerordinal$ and fix $\agent \in \A$. We prove the two required statements one after the other:
		
		\begin{compactenum}
			\item \label{A_set} $\agent^{(\world)_{\delta+1}} =\{(\otherworld)_\delta\mid\otherworld\in S^\agent_0 \cup S^\agent_1\}$.

			If $v \in S^A_0\cup S^A_1$, then $v \in \evalLim{\agent}{\world}$, so $(v)_\delta \in A^{(w)_{\delta+1}}$.  Additionally, we need to show that for any $u \in A^{(w)_{\delta+1}}$, there is an $\smallerordinal$-biworld $v \geqp u$ such that $v \in S^A_0 \cup S^A_1$. So fix $u \in A^{(w)_{\delta+1}}$. By Proposition \ref{prop:ext_to_limit}, there exists $\tilde{u}\in\evalLim{\agent}{\world}$ such that $\tilde{u}\geqp u$.
			 If $\tilde{u}\in S_0^A$ or $\tilde{u}\in S_1^A$, then $\tilde{u}$ is the desired extension of $u$. If $\tilde{u}\notin S_0^A\cup S_1^A$, then $\tilde{u}\in\evalLim{\agent}{\world}\cap\evalLimBar{\agent}{\world}$ and for all $\delta$ there exists $\evensmallerordinal$ such that $\delta<\evensmallerordinal<\smallerordinal$ and $\tilde{u}^\ast_{\evensmallerordinal+1}\notin\impossible{\agent}{\world_{\evensmallerordinal+2}}$. Then we can apply Proposition \ref{prop:ext_to_limit} to $\tilde{u}^\ast_{\evensmallerordinal+1}$ to obtain an \biworld{\smallerordinal} $\tilde{u}'\geqp \tilde{u}^\ast_{\evensmallerordinal+1}\geqp u$. Since $\tilde{u}^\ast_{\evensmallerordinal+1}\notin\impossible{\agent}{\world_{\evensmallerordinal+2}}$, $\tilde{u}'$ is in $S_0^A$, as desired.

			\item $\impossible{\agent}{(\world)_{\delta+1}} =\{(\otherworld)_\delta \mid \otherworld\in \WW^\smallerordinal \setminus (S^\agent_0 \cup S^\agent_1)\}$.
			
			By definition of $S^\agent_0$,  for any $v \in \WW^\smallerordinal \setminus (S^A_0 \cup S^A_1)$, we have $v \in \evalLimBar{\agent}{\world}$, \ie $(v)_\delta \in \bar A^{(w)_{\delta+1}}$.  Additionally, we need to show that for any $u \in \bar A^{(w)_{\delta+1}}$, there is an \biworld{\smallerordinal} $v \geqp u$ such that $v \in \WW^\smallerordinal \setminus (S^A_0 \cup S^A_1)$, \ie $v \notin S^A_0$ and $v \notin S^A_1$. So fix $u \in \bar A^{(w)_{\delta+1}}$. We follow a technique similar to the one adopted for case~{\ref{A_set}} above. By the precision order of limit biworlds, for each $\evensmallerordinal<\gamma<\smallerordinal$ 
			and for each $\otherworld\in\impossible{\agent}{\world_{\evensmallerordinal+1}}$ there exists $\otherworld'\in\impossible{\agent}{\world_{\gamma+1}}$ such that $\otherworld'\geqp\otherworld$. Hence, for each $\eta<\smallerordinal$, we can define a function $f_{\bar{A}}^\eta\colon W_{\bar{A}}^\eta\to \impossible{\agent}{\world_{\eta+1}}$, where
			$W_{\bar{A}}^\eta:=\bigcup_{\gamma<\eta}\impossible{\agent}{\world_{\gamma+1}}$, such that for each $\otherworld\in W_{\bar{A}}^\eta$, $f_{\bar{A}}^\eta(\otherworld)\in\impossible{\agent}{\world_{\eta+1}}$ and  $f_{\bar{A}}^\eta(\otherworld)\geqp\otherworld$. Since \smallerordinal is a countable limit ordinal, we can choose a strictly increasing sequence $(\rho_n)_{n<\omega}$ of ordinals which has limit \smallerordinal. W.l.o.g.\ we assume $ \rho_0>\delta$.
			Then, we define $\tilde{u}\in\WW^\smallerordinal$ as follows: 
			\begin{align}\label{going-up_construction2}
				\tilde{u}_\evensmallerordinal:=
				\begin{cases}
					\restr{u}{\evensmallerordinal} & \text{ if }\evensmallerordinal\leq\delta
					\\  \restr{f_{\bar{A}}^{\rho_0}(u)}{\evensmallerordinal} & \text{ if }\delta<\evensmallerordinal\leq \rho_0
					\\ \restr{f_{\bar{A}}^{\rho_{n+1}}(\tilde{u}_{\rho_n})}{\evensmallerordinal} & \text{ if }\rho_n<\evensmallerordinal\leq \rho_{n+1}
				\end{cases}.
			\end{align}
			Clearly, $\tilde{u}\in\evalLimBar{\agent}{\world}$ and $\tilde{u}\geqp u$. Hence, $\tilde{u}\notin S_0^A$. If $\tilde{u}\notin S_1^A$, then $\tilde{u}$ is the desired extension of $u$. 
			If $\tilde{u}\in S_1^A$, then there exists $\evensmallerordinal<\smallerordinal$ such that $\tilde{u}^\ast_{\evensmallerordinal+1}\in\impossible{\agent}{\world_{\evensmallerordinal+2}}$. Notice that, since $\tilde{u}^\ast_{\evensmallerordinal+1}$ is completed, we have $\tilde{u}^\ast_{\evensmallerordinal+1}\in\impossible{\agent}{\world_{\evensmallerordinal+2}}\setminus \possible{\agent}{\world_{\evensmallerordinal+2}}$.
			Then we can apply to $\tilde{u}^\ast_{\evensmallerordinal+1}$ a construction analogous to  \eqref{going-up_construction2}, to obtain an \biworld{\smallerordinal} $\tilde{u}'\geqp \tilde{u}^\ast_{\evensmallerordinal+1}\geqp u$. Clearly, $\tilde{u}'$ is in $\evalLimBar{\agent}{\world}\setminus\evalLim{\agent}{\world}$, hence not in $S_0^A\cup S_1^A$, as desired.
		\end{compactenum}
		
	\end{compactenum}
	So far, we have shown that for any $\smallerordinal<\ordinal$ and for each \biworld{\smallerordinal} \world there exists a completed \biworld{\smallerordinal{+}1} $\world'\geqp\world$. We want to show that for any $\smallerordinal<\ordinal$ and any \biworld{\smallerordinal} \world, there exists a completed $\ordinal$-biworld $\world'\geqp\world$. We prove this by induction on \ordinal:
	\begin{compactenum}
		\item $\ordinal=\ordinal'+1$ is a successor ordinal. In this case we include the initial step of the induction, namely $\ordinal'=\smallerordinal$. Let \world be an \smallerordinal-biworld. By induction hypothesis, there exists a $\ordinal'$-biworld $\world''\geqp\world$ (for the initial step we have $\world'':=\world$). By the first part of the proof, there exists a completed $\ordinal$-biworld $\world'\geqp\world''\geqp\world$.
		\item \ordinal is a limit ordinal. Let \world be an \biworld{\smallerordinal}. By induction hypothesis, for each ordered pair of ordinals $(\gamma, \delta)$ such that $\smallerordinal\leq \gamma<\delta<\ordinal$ there exists a function $f_\gamma^\delta \colon \WW^\gamma\to\WW^\delta$ such that, for each $\otherworld\in\WW^\gamma$, $f_\gamma^\delta(\otherworld)\geqp\otherworld$ and $f_\gamma^\delta(\otherworld)$ is completed. By the construction of the completed biworld $\world'\geqp\world$ in the first part of the proof of this Lemma, not only $f_\gamma^\delta(\otherworld)$ is completed, but also $\possible{\agent}{f_\gamma^\delta(\otherworld)}\cap\impossible{\agent}{f_\gamma^\delta(\otherworld)}=\emptyset$ for any agent $\agent\in\A$. Since \smallerordinal is a countable limit ordinal, we can choose a strictly increasing sequence of ordinals $(\rho_n)_{n<\omega}$ which has limit~\ordinal. W.l.o.g.\ we assume $\smallerordinal< \rho_0$.
		Then, we define $\world'\in\WW^\ordinal$ as follows: 
		\begin{align*}
			\world'_\evensmallerordinal:=
			\begin{cases}
				\restr{\world}{\evensmallerordinal} & \text{ if }\evensmallerordinal\leq\smallerordinal
				\\  \restr{f_{\smallerordinal}^{\rho_0}(\world)}{\evensmallerordinal} & \text{ if }\smallerordinal<\evensmallerordinal\leq \rho_0
				\\ \restr{f_{\rho_n}^{\rho_{n+1}}(\world'_{\rho_n})}{\evensmallerordinal} & \text{ if }\rho_n<\evensmallerordinal\leq \rho_{n+1}
			\end{cases}.
		\end{align*}
		Clearly, $\world'\geqp\world$. Since $\possible{\agent}{f_\smallerordinal^{\rho_0}(\otherworld)}\cap\impossible{\agent}{f_\smallerordinal^{\rho_0}(\otherworld)}=\emptyset$ for any agent $\agent\in\A$, we have that $\evalLim{\agent}{\world'}\cap\evalLimBar{\agent}{\world'}=\emptyset$ for all $\agent\in\A$. By Proposition \ref{prop:empty_uparrow_then_completed}, $\world'$ is completed, as desired.
	\end{compactenum}
\end{proof}
}

\extended{
\begin{lemma}\label{lem:cc}
	If completability holds at $\ordinal+2$, then \complCondHoldsAt[\mu +1].
\end{lemma}
\begin{proof}
	Let \world be a \biworld{\ordinal{+}1}. By Lemma \ref{lem:empty_then_completed}, if for each $\agent\in\A$, $\possible{\agent}{\world}\cap\impossible{\agent}{\world}=\emptyset$, then \world is completed. 
	
	To prove the other direction needed for the completedness condition, assume there exist $\agent_\ast\in\A$ and  $\otherworld\in\possible{\agent}{\world}_\ast\cap\impossible{\agent}{\world}_\ast$. By definition of biworld, 
	\otherworld must be incompleted, \ie there exist $\otherworld_1, \otherworld_2 \in \WW^{\ordinal+1}$ such that $\otherworld_1, \otherworld_2 \geqp\otherworld$. We will show that we have two distinct \biworld{\ordinal{+}2}s $\world_1\neq\world_2$ such that $\world_1, \world_2\geqp\world$, \ie \world is not completed. By completability at $\ordinal+2$, there exists a completed \biworld{\ordinal{+}2} $\world_1\geqp\world$. Note that for $i \in \{1,2\}$, $v_i$ is in precisely one of the two sets $\possible{\agent}{\world_1}_\ast$ and $\impossible{\agent}{\world_1}_\ast$. Now we define $\world_2$ to be the \biworld{\ordinal{+}2} with $\objworld{\world_2}:=\objworld{\world_1}$, for each $\agent\in\A\setminus\{\agent_\ast\}$, $\possible{\agent}{\world_2}:=\possible{\agent}{\world_1}$ and $\impossible{\agent}{\world_2}:=\impossible{\agent}{\world_1}$, and the sets $\possible{\agent}{\world_2}_\ast$ and $\impossible{\agent}{\world_2}_\ast$ defined as follows:

	\begin{equation*}
	\possible{\agent}{\world_2}_\ast:=
	\begin{cases}
			\possible{\agent}{\world_1}_\ast\setminus\{\otherworld_2\}  & \text{ if } \otherworld_1, \otherworld_2\in \possible{\agent}{\world_1}_\ast 
		\\ 
		\possible{\agent}{\world_1}_\ast\cup\{\otherworld_2\} & \text{ if } \otherworld_1, \otherworld_2\in \impossible{\agent}{\world_1}_\ast 
		\\ \possible{\agent}{\world_1}_\ast\cup\{\otherworld_2\}\setminus\{\otherworld_1\} & \text{ if } \otherworld_1 \in \possible{\agent}{\world_1}_\ast  , \otherworld_2 \in \impossible{\agent}{\world_1}_\ast 
		\\ \possible{\agent}{\world_1}_\ast\cup\{\otherworld_1\}\setminus\{\otherworld_2\} & \text{ if } \otherworld_1 \in \impossible{\agent}{\world_1}_\ast, \otherworld_2 \in \possible{\agent}{\world_1}_\ast 
	\end{cases}
\end{equation*} 
	\begin{equation*}
 \impossible{\agent}{\world_2}_\ast:=
	\begin{cases}
		 \impossible{\agent}{\world_1}_\ast\cup\{\otherworld_2\}& \text{ if } \otherworld_1, \otherworld_2\in \possible{\agent}{\world_1}_\ast 
		\\  \impossible{\agent}{\world_1}_\ast\setminus\{\otherworld_2\}) & \text{ if } \otherworld_1, \otherworld_2\in \impossible{\agent}{\world_1}_\ast 
		\\  \impossible{\agent}{\world_1}_\ast\cup\{\otherworld_1\} \setminus\{\otherworld_2\} & \text{ if } \otherworld_1 \in \possible{\agent}{\world_1}_\ast, \otherworld_2 \in \impossible{\agent}{\world_1}_\ast 
		\\  \impossible{\agent}{\world_1}_\ast \cup\{\otherworld_2\} \setminus\{\otherworld_1\}) & \text{ if } \otherworld_1 \in \impossible{\agent}{\world_1}_\ast, \otherworld_2 \in \possible{\agent}{\world_1}_\ast 
	\end{cases}
\end{equation*} 

	It is easy to see that $\world_2\geqp\world$, as $\otherworld_1$ and $\otherworld_2$ are extensions of $\otherworld\in\possible{\agent}{\world}_\ast\cap\impossible{\agent}{\world}_\ast$. Since $\world_1\neq\world_2$, \world is incompleted, as required.
\end{proof}
}

\extended{
\begin{lemma}\label{lem:well-behaved}
	For every ordinal $\mu$, biworlds are well-behaved until $\mu$.
\end{lemma}
\begin{proof}
	We prove this by induction over $\mu$:
	\begin{compactitem}
		\item $\mu = 0$: Trivial.
		\item $\mu = \mu'+1$ is a successor ordinal. The induction hypothesis is that biworlds are well-behaved until $\mu'$, \ie that restrictability, completability and monotonicity of completedness hold at all ordinals smaller than $\mu'$ and the completedness condition holds at all successor ordinals smaller than $\pl{\mu'}$. Now Lemma~\ref{lem:restrict} implies that restrictability holds at $\mu'$. This together with Lemma~\ref{lem:completability} implies that completability holds at $\mu'$. It now follows from Lemma~\ref{lem:monotonicity} that monotonicity of completedness holds at $\mu'$. Note that since $\pl{\mu} = \pl{\mu'}$, the induction hypothesis implies that the completedness condition holds at all successor ordinals smaller than $\pl{\mu}$. Thus biworlds are well-behaved until $\mu$, as required.
		\item $\mu$ is a limit ordinal. The induction hypothesis is that for every $\alpha<\mu$, biworlds are well-behaved until $\alpha$, \ie that restrictability, completability and monotonicity of completedness hold at all ordinals smaller than $\alpha$ and the completedness condition holds at all successor ordinals smaller than $\pl{\alpha}$. Since this holds for all $\alpha<\mu$, it directly follows that restrictability, completability and monotonicity of completedness hold at all ordinals smaller than $\mu$. Note that $\pl{\mu}=\mu$. So in order to conclude that biworlds are well-behaved until $\mu$, we still need to show that the completedness condition holds at all successor ordinals smaller than $\mu$. Let $\beta +1 < \mu$. Now the induction hypothesis with $\alpha = \beta+3$ implies that completability holds at $\beta+2$. So by Lemma~\ref{lem:cc}, the completedness condition holds at $\beta + 1$, as required.
	\end{compactitem}
\end{proof}

Now the following theorem directly follows from Lemma~\ref{lem:well-behaved}.
}

\begin{theorem}\label{thm:four_properties}
	Let $\alpha \leq \ordinal'\leq\ordinal<\beta$ be ordinals, let $\world$ be a \biworld{\ordinal}, and $\world'$ be a \biworld{\ordinal'}. The following hold:
	\begin{compactenum}
		\item \label{case:restriction_is_biworld} The restriction $\restr{\world}{\smallerordinal}$ is an \biworld{\smallerordinal}.
		\item \label{case:ext_completed_is_completed} If $\world'$ is completed and $\world\geqp\world'$, then $\world$ is completed.
		\item \label{case:competed_ext} There exists a completed \biworld{\beta} $\otherworld$ such that $\otherworld\geqp\world$.
		\item \label{case:completed_biworld_char} If $\ordinal$ is a successor ordinal, then $\world$ is completed if and only if for each $\agent \in \A$, $\possible{\agent}{\world}\cap\impossible{\agent}{\world}=\emptyset$. 
	\end{compactenum}
\end{theorem}


\ignore{
	\begin{proof}[Proof Sketch]
		The proof can be done for all the four statements simultaneously by transfinite induction over \ordinal. We denote by $\pl{\mu}$ the largest limit ordinal smaller or equal to an ordinal $ \mu$. We need some intermediate concepts to express the statements up to a certain depth:
		\begin{compactenum}
			\item ($\mathrm{R}_\ordinal$)  For every \biworld{\ordinal} $\world$, and every $\alpha \leq \mu$, $\restr{\world}{\alpha}$ is a biworld.
			\item ($\mathrm{MoC}_\ordinal$) For each \biworld{\ordinal} $\world$ and each $\alpha \leq \mu$, if $\world|_\alpha$ is completed,  then $\world$ is completed. 
			\item ($\mathrm{C}_\ordinal$) For every $\alpha<\mu$, and every \biworld{\smallerordinal} $u$, there exists a completed \biworld{\ordinal} $u'$ with $u'\geqp u$. 
			\item ($\mathrm{CC}_\ordinal$) For each \biworld{\ordinal} $\world$ with \ordinal a successor, $\world$ is completed if and only if for each agent $A$, $A^w\cap \bar{A}^w = \emptyset$. 
		\end{compactenum}
		 \cref{fig:relations_lemmas} provides a visual representation of some of the relations contained in the omitted lemmas. The successor ordinal case can be easily shown using the relations in \cref{fig:relations_lemmas}, and observing that  $pl(\mu)=pl(\mu+1)$. For the limit ordinal case, we need an additional implication not represented in \cref{fig:relations_lemmas} for the sake of simplicity, namely $\mathrm{C}_{\mu+2} \implies \mathrm{CC}_{\mu+1}$.
	\end{proof}
}

\cref{thm:four_properties} directly implies an important corollary:



\begin{corollary}\label{corollary:completed_char}
	A \biworld{\ordinal} $\world$ is completed if and only if there exists exactly one \biworld{\ordinal{+}1} $\otherworld$ such that $\otherworld\geqp\world$. In this case, also $\otherworld$ is completed.
\end{corollary}
\extended{
\begin{proof}
	If a \biworld{\ordinal} $\world$ is completed, then it has at most one extension at $\ordinal+1$, by definition of completedness. By  \cref{case:competed_ext} of \cref{thm:four_properties} with $\alpha=\ordinal+1$, $\world$ has at least one extension, proving one direction. Conversely, if there is exactly one \biworld{\ordinal{+}1} $\otherworld$ such that $\otherworld\geqp\world$, then $\world$ cannot be incompleted, by definition. 
	Notice that $\otherworld$ has to be completed by \cref{case:ext_completed_is_completed} of \cref{thm:four_properties}.
\end{proof}
}

\section{The logic \cocael}\label{sec:twoValuations}
\label{sec:vals}

In this section we define the syntax and semantics of the logic $\cocael$ that allows to speak about knowledge, common knowledge and only knowing. More specifically, we first define a three-valued model semantics, where the value of a $\cocael$ formula is either true ($\ltrue$), false ($\lfalse$), or unknown ($\lunkn$). 
We consider two orders on these truth values: the \emph{precision order} given by $\lunkn\leqp\lfalse$ and $\lunkn\leqp\ltrue$, and the \emph{truth order} given by $\lfalse\leqt\lunkn\leqt\ltrue$.
We will write $v^{-1}$ for the \emph{inverse} of the truth value $v$, defined by $\lfalse^{-1}=\ltrue$, $\ltrue^{-1}=\lfalse$, and $\lunkn^{-1}=\lunkn$. We will show that our three-valued semantics is precision-monotonic, in the sense that more precise biworlds give more precise results.
Moreover, we will show that in a biworld of depth at least $\omega^2$, every formula evaluates to either true or false. This fact prompts us to define an alternative two-valued semantics. In more detail, since dealing with biworlds with a limit ordinal depth may be complicated and counter-intuitive, we will move our focus to the smallest (for the sake of simplicity) successor ordinal at which biworlds evaluate all formulas (and all sets of formulas) as true or false, namely $\omega^2+1$. In addition, we restrict to completed \biworld{\omega^2{+}1}s: by Item \ref{case:completed_biworld_char} in Theorem \ref{thm:four_properties}, they have the intuitive property that $A^w$ and $\bar A^w$ are disjoint, which corresponds to the intuition that these two sets represent the biworlds deemed possible and the biworlds deemed impossible by $A$. This motivates an alternative semantic characterization of our logic through a canonical Kripke structure consisting of completed $\omega^2{+}1$-biworlds\footnote{For any $\mu\geq \omega^2$, all the results would hold if we considered a canonical Kripke structure consisting of completed \biworlds{\ordinal}. Taking all completed \biworlds{\mu} for all ordinals $\mu$ would also work, but some uniqueness results would be lost, as we would have many worlds representing the same object.}, which we show to coincide with the three-valued semantics on the relevant biworlds. Furthermore, we show that $\cocael$ is semantically well-behaved. In Section~\ref{sec:relatedWork}, we will show that no previously proposed semantic approach leads to a logic that is semantically well-behaved in the way specified in this section.

To say that an agent $A$ only knows $\phi$ can be viewed as a conjunction of the statement that $A$ knows $\phi$, denoted $\knows{A} \phi$, and the statement that $A$ knows at most $\phi$, denoted $M_A \phi$. Given a set $G$ of agents, we write $E_G \phi$ for the statement that every agent in $G$ knows $\phi$, and we write $C_G \phi$ for the statement that $\phi$ is common knowledge within the set $G$ of agents. The language $\cocael$ extends propositional logic with these modal operators as follows:

\begin{definition}\label{ccaelsyntax}
	We define the language $\cocael$ by structural induction with the standard recursive rules of propositional logic, augmented with:
	\vspace{-2mm}
	\begin{align*}
		&\knows{A}(\psi) \in \cocael \text{ if } \psi\in\cocael \text{ and } A\in \Agents\\
		&M_A(\psi) \in \cocael \text{ if } \psi\in\cocael \text{ and } A\in \Agents\\
		&E_G(\psi) \in \cocael \text{ if } \psi\in\cocael \text{ and } G\subseteq \Agents\\
		&C_G(\psi) \in \cocael \text{ if } \psi\in\cocael \text{ and } G\subseteq \Agents
	\end{align*}
	\vspace{-5.5mm}\\
	We use $\onlyKnows{A}  \phi$ as syntactic sugar for $\knows{A} \phi \land M_A \phi$. 
\end{definition}

In the introduction, we have already discussed what an intuitive Kripke semantics for $\onlyKnows{A}  \phi$ would be. Adapting these ideas to the representation of $\onlyKnows{A}  \phi$ as $\knows{A} \phi \land M_A \phi$, we can easily see that the correct way to define the Kripke semantics for $M_A \phi$ is as follows:
\begin{compactitem}
 \item $\mc{K}, w\models M_A \phi$ if for every world $w'\in W$ such that $\mc{K},w'\models \phi$, we have $(w,w')\in R_A$.  
\end{compactitem}

In order to explain why this is a good semantic characterization of ``knowing at most $\phi$'', we will sketch a proof that the only way in which $M_A \phi$ and $\knows{A} \psi$ can both be true is when $\phi$ entails $\psi$: Assume $M_A\phi$ and $\knows{A}\psi$ are true in a world $w$. We want to show that $\phi$ entails $\psi$, \ie that $\psi$ is true in every world $w'$ in which $\phi$ is true. Assume $w'$ is a world in which $\phi$ is true. But the assumption that $M_A\phi$ is true in $w$ together with the above definition of the Kripke semantics for $M_A\phi$ implies that $(w,w')\in R_A$. This together with the assumption that $\knows{A}\psi$ is true in $w$ implies that $\psi$ is true in $w'$, as required. 

In preparation for the upcoming discussion of a three-valued semantics for $\cocael$, note that if we write $(M_A \phi)^{\mc{K},w} = \Tr$ and $(M_A \phi)^{\mc{K},w} = \Fa$ for $\mc{K}, w\models M_A \phi$ and $\mc{K}, w \not\models M_A \phi$ respectively, 
the above characterization of the semantics of $M_A \phi$ is equivalent to the following:
\vspace{-2mm}
\begin{align*}
	(M_A \phi)^{\mc{K},w} = \glb_{{\leq_t}} \{(\phi^{\mc{K},w'})^{-1} \mid w'\notin R_A^w\}
\end{align*}
\vspace{-5mm}


Let us now turn to the three-valued valuation of formulas of $\cocael$. All parts of this definition 
are precisely what one would expect when applying a Kleene-style three-valued semantic approach to logics with a common knowledge operator, taking into account the above rewording of the Kripke semantics of $M_A \phi$.
\begin{definition}\label{ccaelsemsuc}
	Given a formula $\phi\in \cocael$ and a \biworld{\ordinal} $w$, we define the three-valued valuation function $\phi^w$ by induction on $\ordinal$ and the structure of $\phi$:
	\vspace{-1mm}
	\begin{equation*}
		\begin{split}
			P^w &=\Tr  \text{ if }  P\in \objworld{w} \text{ and $P^w = \Fa$ otherwise}\\
			(\phi \land \psi)^w &=glb_{{\leq_t}} (\phi^{w},\psi^{w})\\
			(\neg \phi)^w &=(\phi^{w})^{-1}\\
			(\knows{A}\phi)^{w}&=
			\begin{cases}
				\Un & \text{if $\ordinal=0$}\\
				\glb_{{\leq_t}} \{\phi^{w'} \mid w'\in A^w\} & 	\text{if }\ordinal=\ordinal'+1 \\
				\lub_{\leq_p}\{(\knows{A}\varphi)^{(w)_{\ordinal'}} \mid \ordinal'<\ordinal \} & \begin{split}
					\text{if }&\ordinal\textnormal{ is a limit}\extended{\\& \textnormal{ordinal}}
				\end{split}
			\end{cases} \\
			(M_A\phi)^{w}&=
			\begin{cases}
				\Un & \text{if $\ordinal=0$}\\
				\glb_{{\leq_t}} \{(\phi^{w'})^{-1} \mid w'\in \bar A^w\} & 
					\text{if } \ordinal=\ordinal'+1  \\
				\lub_{\leq_p}\{(M_A\varphi)^{(w)_{\ordinal'}} \mid \ordinal'<\ordinal \}  &\begin{split}
					\text{if }&\ordinal \textnormal{ is a limit}\extended{\\& \textnormal{ordinal}}
				\end{split} 
			\end{cases}\\
			(E_G \phi)^w&=       \glb_{\leq_t}\{(\knows{A}\phi)^{w} \mid A\in G\}\\  
			(C_G \phi)^w&=        \glb_{\leq_t}\{(E_G^k\phi)^{w} \mid k\geq 1\}
		\end{split}
	\end{equation*}%
%
%
%
%
%
\noindent where we define $E_G^k\f$ inductively as $E_G^0\f=\f$, $E_G^{k+1}\f=E_G(E_G^k\f)$. 	\label{ccaelsemlim}
We say a \biworld{\ordinal}  \emph{$w$ satisfies a formula $\phi$} (notation: $w\models\phi$) if $\phi^w=\Tr$.  
A \biworld{\ordinal} $w$ satisfies, or is a model of, a theory if it satisfies all formulas in that theory.
We say a \biworld{\ordinal} $w$ \emph{resolves} a formula $\phi$ if $\phi^w\neq \Un$.
\end{definition}

\ignore{
	\bart{
		My problem is: to evaluate this formula, we only look at the objective interpretations that are nested at most two levels deep... Where: in order to obtain those levels we sometimes need to make a step to smaller succesor ordinals. But in essence, we never look deeper. Seems so strange that it would work... 
		I'll try to work out a different example, a bit simpler. Namely, I'll try to show that the following formula is satisfiable:
		\[(M_A\ltrue) \land \knows{A}\knows{A}p\]
		\ie 
		I only have trivial konwledge (first conjunct) and yet I know that I know $p$ (non-trivial knowledge). 
		Intuitively, this should be unsatisfiable. If I can construct a 
		All the $0$-biworlds are $p:= \{p\}$ and $\bar p := \emptyset$.
		
		All $1$-wolds are  \begin{align*}
			p_1&=  (p, \emptyset, \{p,\bar p\}),\\
			p_2&=  (p, \{p\}, \{p,\bar p\}), \\
			p_3&=  (p, \{p\}, \{\bar p\}),\\ 
			p_4&=  (p, \{\bar p\}, \{p,\bar p\}),\\
			p_5&=  (p, \{\bar p\}, \{p\}),\\ 
			p_6&=  (p, \{p,\bar p\}, \{p,\bar p\}),\\
			p_7&=  (p, \{p,\bar p\}, \{\bar p\}),\\
			p_8&=  (p, \{p,\bar p\}, \{p\}),\\
			p_9&=  (p, \{p,\bar p\}, \emptyset),\\
			\\
			\bar p_1&=  (\bar p, \emptyset, \{p,\bar p\}),\\
			\bar p_2&=  (\bar p, \{p\}, \{p,\bar p\}), \\
			\bar p_3&=  (\bar p, \{p\}, \{\bar p\}),\\
			\bar p_4&=  (\bar p, \{\bar p\}, \{p,\bar p\}),\\
			\bar p_5&=  (\bar p, \{\bar p\}, \{p\}),\\
			\bar p_6&=  (\bar p, \{p,\bar p\}, \{p,\bar p\}),\\
			\bar p_7&=  (\bar p, \{p,\bar p\}, \{\bar p\}),\\
			\bar p_8&=  (\bar p, \{p,\bar p\}, \{p\}),\\
			\bar p_9&=  (\bar p, \{p,\bar p\}, \emptyset),
		\end{align*}
		
		When is $M_A\lfalse$ true in a $2$-biworld $w$? Observe that in any world $w'$ $\lfalse^{w'}=\lfalse$ and hence $(\lfalse^{w'})^{-1}=\ltrue$.
		Hence for any world $w$, if $\bar A^w \neg \emptyset$, then $M_A\lfalse = \ltrue$. This is suspicious! 
	}
}

The following proposition asserts that the three-valued valuation is $\leqp$-monotonic:

\begin{proposition}\label{prop:respower}
	For every pair $w, w'$ of  biworlds such that $w\leq_p w'$ and every formula $\f\in\cocael$, we have $\f^{w}\leqp\f^{w'}$.
\end{proposition}
\extended{
\begin{proof}
	We prove this proposition by induction over $\phi$. We distinguish the following cases:
	
	\begin{compactitem}
		\item $\phi=P$. Since $\objworld{\world'}=\restr{\world'}{0}=\restr{\restr{\world}{d(\world)}}{0}=\restr{\world}{0}=\objworld{\world}$, we have $\phi^\world=\phi^{\objworld{\world}}=\phi^{\objworld{\world'}}=\phi^{\world'}$.
		\item $\phi=\psi_1 \land \psi_2$. Then 
		\begin{align*}
			\phi^w  = \glb_{\leq_t}(\psi_1^\world,\psi_2^\world) \leqp \glb_{\leq_t}(\psi_1^{\world'},\psi_2^{\world'}) = \phi^{\world'},
		\end{align*}
		where the central inequality comes from the induction hypothesis and $\glb_{\leq_t}$ being $\leqp$-monotone.
		\item $\phi=\neg \psi$. Then $\phi^w  = (\psi^w)^{-1} \leqp (\psi^{\world'})^{-1} = \phi^{\world'}$, where the central inequality comes from the induction hypothesis and $(\cdot)^{-1}$ being $\leqp$-monotone.
		\item $\phi=\knows{A} \psi$. We prove this case by an induction on $d(\world')$:
        \begin{compactenum}
            \item $d(\world')=d(\world)$. Then $\world=\world'$.
            \item $d(\world')$ is a successor ordinal. Here we distinguish three cases: 
            \begin{compactenum}
                \item $d(\world) = 0$. Then $\phi^\world = \Un \leqp \phi^{\world'}$.
                \item \label{case:succ} $d(\world)$ is a successor ordinal. By the definition of restriction, every biworld in $\possible{\agent}{\world}$ has an extension in $\possible{\agent}{\world'}$ and the restriction of every biworld in $\possible{\agent}{\world'}$ is in $\possible{\agent}{\world}$. Hence, by the induction hypothesis on $\phi$, and $\glb_{\leq_t}$ being $\leqp$-monotone, we have $\phi^\world=\glb_{\leq_t}\{\psi^\otherworld\mid\otherworld\in\possible{\agent}{\world}\}\leqp \glb_{\leq_t}\{\psi^\otherworld\mid\otherworld\in\possible{\agent}{\world'}\}=\phi^{\world'}$.
                \item $d(\world)$ is a limit ordinal. Then by the limit ordinal case of the definition of $(\knows{A} \phi)^w$, there is a successor ordinal $\mu' < d(\world)$ such that $(\knows{A} \psi)^{w|_{\mu'}} = (\knows{A} \psi)^w$. Then by case~\ref{case:succ}, $\phi^\world \leqp \phi^{\world'}$.
            \end{compactenum}
            \item $d(\world')$ is a limit ordinal. Then by the limit ordinal case of the definition of $(\knows{A} \phi)^w$, there is an ordinal $\mu' < d(\world')$ such that $(\knows{A} \psi)^{w'|_{\mu'}} = (\knows{A} \psi)^{w'}$. By the induction hypothesis, $\phi^\world \leqp (\knows{A} \psi)^{w'|_{\mu'}} = \phi^{w'}$.
        \end{compactenum}
		\item $\phi=M_A \psi$. The proof is analogous to the one for the case  $\phi=\knows{A} \psi$. 
		\item $\phi=E_G \psi$. Then 
		\begin{align*}
			\phi^w &= (E_G \psi)^w = \glb_{\leq_t}\{(\knows{A}\psi)^{w} \mid A\in G\} \\&\leqp \glb_{\leq_t}\{(\knows{A}\psi)^{\world'} \mid A\in G\} = (E_G \psi)^{\world'} = \phi^{\world'},
		\end{align*}
		where the central inequality comes from the induction hypothesis and $\glb_{\leq_t}$ being $\leqp$-monotone.
		\item $\phi=C_G \psi$. Then 
		\begin{align*}
			\phi^\world &= (C_G \psi)^w = \glb_{\leq_t}\{(E_G^k\psi)^{w} \mid k\in\mathbb{N}\} \\&\leqp \glb_{\leq_t}\{(E_G^k\psi)^{\world'} \mid k\in\mathbb{N}\} = (C_G \psi)^{\world'} = \phi^{\world'},
		\end{align*}
		where the central inequality comes from the induction hypothesis and $\glb_{\leq_t}$ being $\leqp$-monotone.
	\end{compactitem}
\end{proof}
}

The notion of the \emph{modal depth} of a formula allows us to specify conditions for a formula to be resolved by a biworld:

\begin{definition}[Modal depth]\label{def:MD}
	The \emph{modal depth} $\MD(\phi)$ of a formula $\phi\in\cocael$ is defined by inductively as follows:
	\begin{compactitem}
		\item $\MD(P)=0$ for every propositional atom $P$
		\item $\MD(\neg \phi) = \MD(\phi)$ 
		\item $\MD(\phi  \land \psi) = 
		max(\MD(\phi),\MD(\psi))$
		\item $\MD(\knows{A}\phi)= \MD(\phi)+1$
		\item $\MD(M_A\phi)=\MD(\phi)+1$ 
		
		\item $\MD(E_G\phi)=\MD(\phi)+1$
		
		\item $\MD(C_G\phi)=\MD(\phi)+\omega$, which is the smallest limit ordinal greater than $\MD(\phi)$.
	\end{compactitem}
\end{definition}

Note that the modal depth of any formula in $\cocael$ is less than $\omega^2$.
Every formula of a given modal depth is resolved at any biworld of at least this depth:

\begin{theorem}\label{thm:resolvingmodaldepth}
	If $w$ is a \biworld{\ordinal} and $\phi \in \cocael$ is a formula such that $\MD(\phi) \leq \ordinal$, then $w$ resolves $\phi$.
\end{theorem}
\extended{
\begin{proof}
    We prove this proposition by a simultaneous induction over $\mu$ and the complexity of $\phi$. The induction hypothesis is that if either $\alpha < \mu$ and $\psi$ is any formula or $\cocael$ or $\alpha = \mu$ and $\psi$ is a subformula of $\phi$, then any $\alpha$-biworld resolves $\psi$. For the purpose of this induction (and some similar induction proofs later), we consider $\knows{A} \psi$ to be a subformula of $E_G \psi$ whenever $A \in G$, and we consider $E_G^k \psi$ to be a subformula of $C_G \psi$ for any $k \geq 1$. It is enough to show that this induction hypothesis together with the assumptions that $w$ is a \biworld{\ordinal} and $\MD(\phi) \leq \ordinal$ implies that $w$ resolves $\phi$. For this, we distinguish the following cases:
	\begin{compactitem}
		\item $\phi=P$. Clearly, $\phi^\world\neq \Un$.
		\item $\phi=\psi_1 \land \psi_2$.  By \cref{def:MD}, we have $\MD(\psi_1),\MD(\psi_2)\leq \ordinal$. 
		Hence, by induction hypothesis, $(\psi_1 \land \psi_2)^\world= \glb_{\leq_t}(\psi_1^\world, \psi_2^\world)\neq \Un$. 
		\item $\phi=\neg \psi$. By \cref{def:MD}, we have $\MD(\psi)=\MD(\phi)\leq \ordinal$. By induction hypothesis, $\psi^\world\neq \Un$, so $(\neg\psi)^\world=(\psi^\world)^{-1}\neq \Un$. 
		\item $\phi=\knows{A} \psi$. Then $\MD(\phi) \neq 0$, so $\mu \neq 0$. We distinguish two cases:
		\begin{compactenum}
		 \item $\mu$ is a successor ordinal. By the induction hypothesis, any world in $A^w$ resolves $\psi$. This together with the successor ordinal case of the definition of $(\knows{A} \phi)^w$ implies that $(\knows{A}\psi)^\world\neq \Un$, as required.
		 \item $\mu$ is a limit ordinal. Since $\MD(\phi)$ is a successor ordinal, $\MD(\phi) < \mu$, so by induction hypothesis, $w|_{\MD(\phi)}$ resolves $\phi$. Then by Proposition~\ref{prop:respower}, $w$ resolves $\phi$, as required.
		\end{compactenum}
		\item $\phi=M_A \psi$. The proof is analogous to the one for the case $\phi=\knows{A} \psi$.
		\item $\phi=E_G \psi$. By the induction hypothesis, $(\knows{A}\psi)^\world\neq \Un$ for any $A \in G$. Hence, $\phi^\world\neq \Un$, as desired. 
		\item $\phi=C_G \psi$. By \cref{def:MD}, we have $\MD(E_G^k\psi)=\MD(\psi)+k\leq \MD(\psi)+\omega=\MD(\phi)\leq \ordinal$. By induction hypothesis, $(E_G^k\psi)^\world\neq \Un$ for all $k\in\natnrs$. Hence, $\phi^\world\neq \Un$, as desired. 
	\end{compactitem}
\end{proof}
}



As explained at the beginning of this section, 
by \cref{case:completed_biworld_char} of Theorem~\ref{thm:four_properties}, a \biworld{\mu{+}1} $\world$ precisely captures the knowledge of every agent $A$ iff $w$ is completed. 
Hence, Theorem~\ref{thm:resolvingmodaldepth} combined with the fact that the modal depth of any formula in $\cocael$ is less than $\omega^2$ motivates focusing on completed \biworlds{\omega^2{+}1}, as $\omega^2+1$ is the first successor ordinal greater than the modal depth of all formulas.


\begin{definition} 
	A \emph{world} is a completed \biworld{\omega^2{+}1}.
\end{definition}

\begin{definition}
	The $\omega^2{+}1$\emph{-completed Kripke structure} $\K := (U,(R_A)_{A\in\A})$ is the Kripke structure whose underlying world set $U$ is the set of all worlds, and whose accessibility relations $R_A$ are given by  
	\vspace{-1mm}
	\[ R_A  = \{ (w,w') \in U^2  \mid  w'|_{\omega^2} \in A^w  \}.\] 
	\vspace{-5mm}\\
	Instead of $(w,w') \in R_A$, we sometimes write $w R_A w'$.
\end{definition}

Interpreting formulas in this canonical Kripke structure $\K$ in the standard way (with the above specified semantics for $M_A \phi$) amounts to a two-valued valuation of $\cocael$.

\extended{Given a world $w$, we write $R_A^w$ for the set of all worlds $R_A$-accessible from $w$. }

\ignore{
The following proposition is important for seeing that the above definition of the accessibility relation is indeed sensible:

\begin{proposition}\label{prop:reduction_to_worlds}
	Let $\phi\in \cocael$, $w\in U^{\omega^2}$, and $\agent \in \A$. The formula $\phi$ is true in every world in $\evalLim{\agent}{\world}$ if and only if it is true in every biworld in $\evalLim{\agent}{\world}$.
\end{proposition}

\marcos{I found it important to put the proposition here, but for this it was necessary to separate the proposition from its proof in the technical report. I thought hard about this decision and think that it's the best we can do here.}

\extended{Before we can prove this proposition, we first need to introduce some notation and prove a lemma.}



Given a world $w$, we write $R_A^w$ for the set of all worlds accessible from $w$. In other words,
$$R_A^w = \evalLim{\agent}{\world} \cap U^{\omega^2}$$
Note that by post-completedness of $\world$ it follows that
$$R_A^w =\{v \mid v\in U^{\omega^2}\textnormal{ and $\exists\ordinal < \omega^2$ s.t.\ $(v)_\ordinal \notin \bar A^{(w)_{\ordinal+1}}$} \}.$$
We will often make use of this alternative characterization of $R_A^w$.

In order to establish the relationship beteween the two-valued and the three-valued valuation, we first need to prove the following lemma that for $w' >_p w$ establishes links between $R_A^{w'}$ on the one hand and $A^w$ and $\bar A^w$ on the other hand:

\begin{lemma}
	\label{lem:shiftA}
	Suppose $\ordinal$ is a successor ordinal less than $\omega^2$, $w$ is a \biworld{\ordinal}, and $w'$ is a world such that $w' \geq_p w$. Then 
	\begin{compactenum}[(a)]
		\item \label{case:ext_in_limit_world} for every $\ordinal{-}1$-biworld $v \in A^w$, there is a world $v' \in R_A^{w'}$ such that $v' \geq_p v$;
		\item for every world $v' \in R_A^{w'}$, there is a $\ordinal{-}1$-biworld $v \in A^w$ such that $v' \geq_p v$;
		\item for every $\ordinal{-}1$-biworld $v \in \bar A^w$, there is a world $v' \notin R_A^{w'}$ such that $v' \geq_p v$;
		\item for every world $v' \notin R_A^{w'}$, there is a $\ordinal{-}1$-biworld $v \in \bar A^w$ such that $v' \geq_p v$.
	\end{compactenum}
\end{lemma}
\extended{
	\begin{proof}
		\begin{compactenum}[(a)]
			\item Suppose $v \in A^w$. We distinguish two cases:
			\begin{compactenum}[1.]
				\item $\otherworld \in \bar A^w$. In this case, $v \in A^w \cap \bar A^w$, so by \cref{case:completed_biworld_char} of \cref{thm:four_properties}, $w$ is not completed. Since $w'$ is post-completed, there is a $\ordinal' < \omega^2$ such that $(w')_{\ordinal'}$ is completed. By  \cref{case:ext_completed_is_completed} of \cref{thm:four_properties}, we can choose $\ordinal'$ to be a successor ordinal. Then $A^{(w')_{\ordinal'}} \cap \bar A^{(w')_{\ordinal'}} = \emptyset$. Since $(w')_\ordinal = w$ is not completed, \cref{case:ext_completed_is_completed} of \cref{thm:four_properties} implies that $\ordinal' > \ordinal$. So $(w')_{\ordinal'} \geq_p w$, so by the definition of the precision relation, $A^{\restr{(w')_{\ordinal'}}{\ordinal}} = A^w$. 
				Since $v \in A^w$, there exists a $v^* \geqp v$ such that $v^* \in A^{(w')_{\ordinal'}}$. Since $A^{(w')_{\ordinal'}} \cap \bar A^{(w')_{\ordinal'}} = \emptyset$, we get that $v^* \notin \bar A^{(w')_{\ordinal'}}$. By Corollary \ref{corollary:post-completed_extension}, there is a world $v' \geqp v^*$. Then $v' \geqp v$ and $(v')_{\ordinal'} = v^* \notin \bar A^{(w')_{\ordinal'}}$. So by the alternative characterization of $R_A^{w'}$, it follows that $v' \in R_A^{w'}$, as required.
				\item $v \notin \bar A^w$. By Corollary \ref{corollary:post-completed_extension}, there is a world $v'$ such that $v' \geq_p v$. Since $v = (v')_{\ordinal-1}$, it follows by the alternative characterization of $R_A^{w'}$ that $v' \in R_A^{w'}$.
			\end{compactenum}
			\item Let $v' \in R_A^{w'}$. Then, there is some $\ordinal' < \omega^2$ such that $(v')_{\ordinal'} \notin \bar A^{(w')_{\ordinal'+1}}$. Define $v := (v')_{\ordinal-1}$. Clearly $v' \geqp v$. So it is now enough to prove that $v \in A^w$. For this we distinguish two cases:
			\begin{compactenum}[1.]
				\item $\ordinal-1 \leq \ordinal'$. By the definition of biworlds, $(v')_{\ordinal'} \notin \bar A^{(w')_{\ordinal'+1}}$ implies that $(v')_{\ordinal'} \in A^{(w')_{\ordinal'+1}}$. By definition of restriction, this implies that $v = (v')_{\ordinal-1} \in A^{(w')_{\ordinal}} = A^w$, as required.
				\item $\ordinal-1 > \ordinal'$. By definition of restriction,  $(v')_{\ordinal'} \notin \bar A^{(w')_{\ordinal'+1}}$ implies that $v = (v')_{\ordinal-1} \notin \bar A^{(w')_\ordinal} = \bar A^w$. By the definition of biworlds, this implies that $v \in A^w$, as required.
			\end{compactenum}
			\item Suppose $v \in \bar A^w$. We distinguish two cases:
			\begin{compactenum}[1.]
				\item $v \in A^w$. In this case, $v \in A^w \cap \bar A^w$, so by \cref{case:completed_biworld_char} of \cref{thm:four_properties}, $v$ is not completed and $w$ is not completed. Since $w'$ is post-completed, there is a $\ordinal' < \omega^2$ such that $(w')_{\ordinal'}$ is completed. By  \cref{case:ext_completed_is_completed} of \cref{thm:four_properties}, we can choose $\ordinal'$ to be a successor ordinal. Then $A^{(w')_{\ordinal'}} \cap \bar A^{(w')_{\ordinal'}} = \emptyset$. Since $(w')_\ordinal = w$ is not completed, \cref{case:ext_completed_is_completed} of \cref{thm:four_properties} implies that $\ordinal' > \ordinal$. So $(w')_{\ordinal'} \geq_p w$, so by the definition of the precision relation, $\bar A^{\restr{(w')_{\ordinal'}}{\ordinal}} = \bar A^w$. 
				Since $v \in \bar A^w$, there exists a $v^* \geqp v$ such that $v^* \in \bar A^{(w')_{\ordinal'}}$. Since $A^{(w')_{\ordinal'}} \cap \bar A^{(w')_{\ordinal'}} = \emptyset$, we get that $v^* \notin A^{(w')_{\ordinal'}}$. By Corollary \ref{corollary:post-completed_extension}, there is a world $v' \geqp v^*$. Then $v' \geqp v$ and $(v')_{\ordinal'} = v^* \notin A^{(w')_{\ordinal'}}$. So by the definition of $R_A^{w'}$, it follows that $v' \notin R_A^{w'}$, as required.
				\item $v \notin A^w$. By Corollary \ref{corollary:post-completed_extension}, there is a world $v'$ such that $v' \geq_p v$. Since $v = (v')_{\ordinal-1}$, it follows by the definition of $R_A^{w'}$ that $v' \notin R_A^{w'}$.
			\end{compactenum}
			\item Let $v' \notin R_A^{w'}$ be post-completed. Then, there is some $\ordinal' < \omega^2$ such that $(v')_{\ordinal'} \notin A^{(w')_{\ordinal'+1}}$. Define $v := (v')_{\ordinal-1}$. Clearly $v' \geqp v$. So it is now enough to prove that $v \in \bar A^w$. For this we distinguish two cases:
			\begin{compactenum}[1.]
				\item $\ordinal-1 \leq \ordinal'$. By the definition of biworlds, $(v')_{\ordinal'} \notin A^{(w')_{\ordinal'+1}}$ implies that $(v')_{\ordinal'} \in \bar A^{(w')_{\ordinal'+1}}$. By definition of restriction, this implies that $v = (v')_{\ordinal-1} \in \bar A^{(w')_{\ordinal}} = \bar A^w$, as required.
				\item $\ordinal-1 > \ordinal$. By definition of restriction,  $(v')_{\ordinal'} \notin A^{(w')_{\ordinal'+1}}$ implies that $v = (v')_{\ordinal-1} \notin A^{(w')_\ordinal} = \bar A^w$. By the definition of biworlds, this implies that $v \in \bar A^w$, as required.
			\end{compactenum}
		\end{compactenum}
	\end{proof}
}

\extended{
Now we are ready to prove Proposition~\ref{prop:reduction_to_worlds}:
\renewcommand*{\proofname}{Proof of Proposition~\ref{prop:reduction_to_worlds}}
	\begin{proof}
		Suppose $\phi^{\world'}=\Tr$ for all worlds $\world'\in\evalLim{\agent}{\world}\cap U^{\omega^2}$, and let $\otherworld\in \evalLim{\agent}{\world}$. By Theorem \ref{thm:resolvingmodaldepth}, $\restr{\otherworld}{\MD(\phi)+1}$ resolves $\phi$. Moreover, by \cref{def:new:arrow_sets}, $\restr{\otherworld}{\MD(\phi)+1}\in\possible{\agent}{\world_{\MD(\phi)+2}}$. By  \cref{case:ext_in_limit_world} of \cref{lem:shiftA}, there exists a world $\otherworld'\in \evalLim{\agent}{\world}\cap U^{\omega^2}$ such that $\otherworld'\geqp\restr{\otherworld}{\MD(\phi)+1}$. By Proposition \ref{prop:respower} and our assumption on worlds in $\evalLim{\agent}{\world}\cap U^{\omega^2}$, we have $\phi^{\restr{\otherworld}{\MD(\phi)+1}}\leqp \phi^{\otherworld'}=\Tr$. Since $\restr{\otherworld}{\MD(\phi)+1}$ resolves $\phi$, we have $ \phi^{\restr{\otherworld}{\MD(\phi)+1}}=\Tr$. By Proposition \ref{prop:respower}, $\phi^{\restr{\otherworld}{\MD(\phi)+1}}\leqp \phi^\otherworld$. Hence, we have $\phi^\otherworld=\Tr$ as desired. The other direction is trivial.
	\end{proof}
\renewcommand*{\proofname}{Proof}
}
}

\extended{
The following lemma is of central importance, as it establishes a relationship between the two-valued and the three-valued valuation and can be used to prove multiple important theorems.
%

\begin{lemma}
	\label{lem:valrel}
	If $\phi \in \cocael$, $w$ is an $\MD(\phi)$-biworld and ${w' \geq_p w}$ is a world, then $\phi^w = \phi^{\K,w'}$.
\end{lemma}
	\begin{proof}
		By Proposition~\ref{prop:respower}, $\phi^w = \phi^{w'}$. So all that is left to be shown is that $\phi^{w'} = \phi^{\K,w'}$. We prove this proposition by induction over the complexity of $\phi$. So suppose as induction hypothesis that the lemma holds for all subformulas $\psi$ of $\phi$. In order to prove that the lemma holds for $\phi$, we distinguish the following cases:
		\begin{compactitem}
			\item $\phi$ is $P$. Then $\phi^{w'} = \phi^{\objworld{w'}} = \phi^{\K,w'}$.
			\item $\phi$ is $\psi_1 \land \psi_2$. Then 
			\begin{align*}
				\phi^{w'} &= (\psi_1 \land \psi_2)^{w'}\\
				&= \glb_{\leq_t}(\psi_1^{w'},\psi_2^{w'})\\
				&= \glb_{\leq_t}(\psi_1^{\K,w'},\psi_2^{\K,w'}) \textnormal{ by induction hypothesis}\\
				&= (\psi_1 \land \psi_2)^{\K,w'}\\
				&= \phi^{\K,w'}.
			\end{align*}
			\item $\phi$ is $\neg \psi$. Then $\phi^{w'} = (\neg \psi)^{w'} = (\psi^{w'})^{-1} = (\psi^{\K,w'})^{-1} = (\neg \psi)^{\K,w'} = \phi^{\K,w'}$.
			\item $\phi$ is $\knows{A} \psi$. Then $\MD(\phi)$ is a successor ordinal, so $\phi^{w'} = (\knows{A} \psi)^{w'} = \glb_{\leq_t}\{\psi^{w^*} \mid w^* \in A^{w'}\}$. By \cref{case:competed_ext} of Theorem~\ref{thm:four_properties}, every $\omega^2$-biworld $w^*$ can be extended to a completed $\omega^2{+}1$-biworld $w^+$, and by Proposition~\ref{prop:respower}, $\psi^{w^*} = \psi^{w^+}$. Thus $\phi^{w'} = \glb_{\leq_t}\{\psi^{w^+} \mid w^+ \in R_A^{w'} \}$, which by the induction hypothesis equals $(\knows{A} \psi)^{\K,w'}$, as required.
			\item The proof is analogous to the one for the case $\phi = \knows{A} \psi$.
			\item $\phi$ is $E_G \psi$. Then 
			\begin{equation*}
				\begin{split}
					\phi^w &= (E_G \psi)^w \\
					&= \glb_{\leq_t}\{(\knows{A}\psi)^{w} \mid A\in G\} \\
					&= \glb_{\leq_t}\{(\knows{A}\psi)^{\K,w'} \mid A\in G\}\;\;\begin{split}
						&\textnormal{by induction}\\&\textnormal{hpothesis}
					\end{split}\\
					&= (E_G \psi)^{\K,w'} \\
					&= \phi^{\K,w'}.
				\end{split}
			\end{equation*}
			\item $\phi$ is $C_G \psi$. Then 
			\begin{equation*}
				\begin{split}
					\phi^w &= (C_G \psi)^w \\
					&= \glb_{\leq_t}\{(E_G^k\psi)^{w} \mid k\in\mathbb{N}\} \\
					&= \glb_{\leq_t}\{(E_G^k\psi)^{\K,w'} \mid k\in\mathbb{N}\} \;\;\begin{split}
						&\textnormal{by induction}\\&\textnormal{hpothesis}
					\end{split}\\
					&= (C_G \psi)^{\K,w'} \\
					&= \phi^{\K,w'}.
				\end{split}
			\end{equation*}
		\end{compactitem}
	\end{proof}

Combining Lemma \ref{lem:valrel} with Proposition \ref{prop:respower}, we directly get the following theorem, which}
\conf{The following theorem} 
tells us that the two-valued and three-valued valuations fully coincide on worlds\extended{, and thus expresses the relationship between the two-valued and the three-valued valuation in a simpler way than Lemma \ref{lem:valrel} does}:


\begin{theorem}
	\label{thm:valcoinc}
	If $\phi \in \cocael$ and $w\in U$, then $\phi^w = \phi^{\K,w}$.
\end{theorem}

\extended{Finally, Lemma~\ref{lem:valrel} can also be used to prove the following theorem, which}
\conf{The next theorem}
is of central importance to show that our semantics generally avoids a problem that some previous accounts of only knowing and common knowledge had, namely the problem that $\onlyKnows{A}  \neg C_G p$ is not satisfiable in those accounts, even though it should be (see Section~\ref{sec:relatedWork} for a discussion of this problem in other accounts). The following theorem shows that no such problems can arise in the $\omega^2{+}1$-completed Kripke structure $\K$:

\begin{theorem}\label{thm:O_satisfiable}
	Let $\agent$ be an agent. For every formula $\phi \in \cocael$, there is a world $w$ such that $(\onlyKnows{A}  \phi)^w = (\onlyKnows{A}  \phi)^{\K,w} = \Tr$. Moreover, if $\world_1$ and $\world_2$ are two such worlds, then ${\agent}^{\world_1}={\agent}^{\world_2}$ and $\bar{\agent}^{\world_1}=\bar{\agent}^{\world_2}$.
\end{theorem}
\extended{
	\begin{proof}
		By Theorem \ref{thm:valcoinc}, $(\onlyKnows{A}  \phi)^w = (\onlyKnows{A}  \phi)^{\K,w}$ for any world $w$. Hence it is enough to show that there is a world $w$ such that $(\onlyKnows{A}  \phi)^{\K,w} = \Tr$. 
		
		Let $v$ be an arbitrary $0$-biworld. Define
		\begin{align*}
			W &:= \{ u|_{\omega^2} \mid u \in U \textnormal{ and } \phi^{\K,u} = \Tr \},\\
			\overline{W} &:= \WW^{\omega^2} \setminus W,\\
			v' &:= (v,(W)_{A \in \A},(\overline{W})_{A \in \A}).
		\end{align*}
		Clearly $v'$ is a world. It follows directly from the definition of $v'$ 
		that $(\knows{A} \phi)^{\K,v'} = \Tr$ and that $(M_A \phi)^{\K,v'} = \Tr$, \ie that $(\onlyKnows{A}  \phi)^{\K,v'} = \Tr$. 
		
		For the second statement, let $\world_1,\world_2$ be worlds such that $(\onlyKnows{A} \phi)^{\world_1}=(\onlyKnows{A} \phi)^{\world2}=\Tr$. By the definition of the only knowing operator, we  have $(\knows{A}\phi)^{\world_i}=(M_A\phi)^{\world_i}=\Tr$ for $i\in\{1,2\}$. This means that $\phi^\otherworld=\Tr$ for all $\otherworld\in {\agent}^{\world_1}\cup{\agent}^{\world_2}$, and  $\phi^\otherworld=\Fa$ for all $\otherworld\in \bar{\agent}^{\world_1}\cup\bar{\agent}^{\world_2}$. Since ${\agent}^{u}\cup\bar{\agent}^{u}=\WW^{\omega^2}$ for any world $u$, we must have 
		\begin{align*}
			{\agent}^{\world_1}&={\agent}^{\world_2}=\{v\in\WW^{\omega^2}\mid \phi^v=\Tr\}
			\\ \bar{\agent}^{\world_1}&=\bar{\agent}^{\world_2}=\{v\in\WW^{\omega^2}\mid \phi^v=\Fa\}.
		\end{align*}
	\end{proof}
}

We will now show that the two-valued valuation of $\cocael$ (and thus by Theorem~\ref{thm:valcoinc} also the three-valued valuation, when restricted to suitable biworlds) gives rise to a sensible entailment relation between formulas of $\cocael$. We define this relation as follows:

\begin{definition}
	Let $\phi\in \cocael$ be a formula, and $\Gamma\subseteq\cocael$ be a set of formulas. We write $\Gamma\models \phi$ if $\phi^{\K, \world}=\Tr$ for every world $\world$  such that $\psi^{\K,\world}=\Tr$ for all $\psi\in\Gamma$.
\end{definition}

Note that this is definition of the entailment relation does not coincide with the standard way of defining the entailment with respect to a Kripke semantics, as in our case we fix a canonical Kripke structure rather than quantifying over all Kripke structures.
The fact that this entailment relation behaves in a sensible way is captured by the properties listed in the following theorem:

\begin{theorem}\label{thm:properties}
	Let $\phi, \psi \in \cocael$ two formulas, $\Gamma, \Gamma'\subset \cocael$ two sets of formulas, $\agent \in \A$ an agent, and $G\subseteq\A$ a non-empty set of agents. Then, the following properties 
	hold:
	\begin{compactenum}
		\item \label{Prop} (Prop) For each propositional tautology $\phi$, we have $\models \phi$.
		\item \label{MP} (MP) $\phi, \phi \Rightarrow \psi \models \psi$.
		\item  (Mono) If $\Gamma \models \phi$, then $\Gamma, \psi \models \phi$. 
		\item \label{Cut} (Cut) If $\Gamma \models \phi $ and $\Gamma', \phi \models \psi$, then $\Gamma, \Gamma' \models\psi$.
		\item \label{K} (K) $\models (\knows{A}(\phi\Rightarrow\psi)\land \knows{A}\phi)\Rightarrow \knows{A}\psi$.
		\item \label{Nec} (Nec) If $\models\phi$, then $\models \knows{A}\phi$.
		\item \label{M} (M) If $\phi\not\models\psi$, then $ M_A\phi \models \neg \knows{A}\psi$.
		\item (O) \label{O} $\onlyKnows{A} \phi\not\models \false$. 
		\item \label{fixed_point_axiom} (Fixed point axiom) $\models C_G \phi \iff E_G(\phi\wedge C_G\phi)$.
		\item \label{Induction_rule} (Induction rule) If $ \phi\models E_G(\phi\wedge\psi)$, then $ \phi\models C_G\psi$.
	\end{compactenum}
\end{theorem}
\extended{
	\begin{proof}
		\begin{compactenum}
			\item Let $\phi$ be any propositional tautology with atoms $P_1, \ldots, P_k$, and $\xi_1, \ldots, \xi_k$ be formulas in $\cocael$. For every $i\in[1,k]$, we can substitute $P_i$ with $\xi_i$ in $\phi$ to obtain $\phi_\xi$. Let $w$ be a world. Since for every formula $\psi\in\cocael$ the evaluation $\psi^{\K, w}$ cannot be undefined, for every $i\in[1,k]$, we can assign to $P_i$ the truth value given by $\xi_i^{\K,w}$. Since $\phi$ is a propositional tautology, it must be true with the above truth assignments. Hence, $\phi_\xi^{\K,w}=\Tr$ by the definition of the evaluation in the Kripke structure.
			\item Let $\world \in U^{\omega^2}$ such that $\phi^{\K,\world}=\Tr$ and $(\phi \Rightarrow \psi)^{\K,\world}=\Tr$. Notice that $(\phi \Rightarrow \psi)^{\K,\world}=\left((\phi\land\neg\psi)^{\K,\world}\right)^{-1}$. Hence, $\glb_{\leq_t}(\phi^{\K,\world}, (\neg\psi)^{\K,\world})=\Fa$. Since $\phi^{\K,\world}=\Tr$, we have $(\psi^{\K,\world})^{-1}=(\neg\psi)^{\K,\world}=\Fa$, \ie $\psi^{\K,\world} = \Tr$, as desired.
			\item Clear by the definition of logical consequence relation.
			\item Let $\world\in U$ such that $\xi^{\K,\world}=\Tr$ for all $\xi\in \Gamma\cup\Gamma'$. Since $\Gamma \models \phi $, we have $\phi^{\K,\world}=\Tr$. Then, $\psi^{\K,\world}=\Tr$ because $\Gamma', \phi \models \psi$.
			\item Let $\world \in U$. We want to show that $\left((\knows{A}(\phi\Rightarrow\psi)\land \knows{A}\phi)\Rightarrow \knows{A}\psi\right)^{\K,\world}=\Tr$. Notice that
			\begin{equation*}
				\begin{split}
					V:&=((\knows{A}(\phi\Rightarrow\psi)\land \knows{A}\phi)\Rightarrow \knows{A}\psi)^{\K,\world} \\
					&= \left(\neg((\knows{A}(\neg(\phi\land\neg\psi))\land \knows{A}\phi)\land \neg \knows{A}\psi)\right)^{\K,\world}\\
					&=\left(\glb_{\leq_t}\left(
					S, ((\knows{A}\psi)^{\K,\world})^{-1}
					\right)\right)^{-1},
				\end{split}
			\end{equation*}
			where $S:=\glb_{\leq_t}((\knows{A}(\neg(\phi\land\neg\psi)))^{\K,\world}, (\knows{A}\phi)^{\K,\world})$.
			If $(\knows{A}\psi)^{\K,\world}=\Tr$, then $V=\Tr$, as desired. Now, we assume $(\knows{A}\psi)^{\K,\world}=\Fa$. Observe that $V$ reduces to just $V=\left(\glb_{\leq_t}((\knows{A}(\neg(\phi\land\neg\psi)))^{\K,\world}, (\knows{A}\phi)^{\K,\world})\right)^{-1}$. If $(\knows{A}\phi)^{\K,\world}=\Fa$, then again $V=\Tr$, as desired. So, we suppose $(\knows{A}\phi)^{\K,\world}=\Tr$, \ie $\phi^{\K,\otherworld}=\Tr$ for all $\otherworld\in R_A^\world$. Hence, we have that
			\begin{align*}
				(\knows{A}(\neg(\phi &\land\neg\psi)))^{\K,\world}=  \{(\neg(\phi\land\neg\psi))^{\K,\otherworld}\mid \otherworld\in R_A^{\world}\}
				\\&= \{(\glb_{\leq_t}(\phi^{\K,\otherworld}, (\psi^{\K,\otherworld})^{-1}))^{-1}\mid \otherworld\in R_A^{\world}\}
				\\&= \{((\psi^{\K,\otherworld})^{-1})^{-1}\mid \otherworld\in R_A^{\world}\}
				\\&=\{\psi^{\K,\otherworld}\mid \otherworld\in R_A^{\world}\}=(\knows{A}\psi)^{\K,\world}=\Fa,
			\end{align*}
			which implies $V=\Tr$, as desired.
			\item Let $\world\in U$. By definition, $(\knows{A}\phi)^{\K,\world}=glb_{{\leq_t}} \{\phi^{\K,w'}|w'\in R_A^\world\} $. Since $\phi^{\K,\world'}=\Tr$ for any $\world'\in U$ by hypothesis, we have $(\knows{A}\phi)^{\K,\world}=\Tr$, as desired.
			\item  Since $\phi\not\models\psi$, there exists $\world^\ast\in U$ such that $\phi^{\K,\world^\ast}=\Tr$ and $\psi^{\K, \world^\ast}=\Fa$. Let $ \world\in U$ be such that  $(M_A\phi)^{\K,\world}=\Tr$. Then, $R_A^w \supseteq \{w \in U \mid \phi^{\K, \world} = \Tr\}$, and in particular $\world^\ast\in R_A^\world$. Hence, $(\knows{A}\psi)^{\K,\world}=\glb_{\leq_t}\left\{\psi^{\K,\world'}\mid \world'\in R_A^\world\right\}=\Fa$.
			\item By Theorem \ref{thm:O_satisfiable}, there exists a world $\world$ such that $(\onlyKnows{A} \phi)^{\K,\world}=\Tr$. Since $\false^{\K,\world}=\Fa$, we have $\onlyKnows{A} \phi\not\models \false$.
			\item Let $\world \in U$. By Definition \ref{ccaelsemsuc}, we have the following double implications: 
			\begin{align*}
				(E_G(&\phi\wedge C_G\phi))^{\K,\world}=\Tr 
				\\ \iff& \forall \agent\in G, \; \left( \knows{A}(\phi\wedge C_G\phi )\right)^{\K,\world} =\Tr
				\\  \iff & \forall \agent \in G, \; \forall \world'\in\evalLim{\agent}{\world}, \; (\phi\wedge C_G\phi)^{\K,\world'}=\Tr
				\\ \iff& \forall A\in G, \; \forall \world'\in \evalLim{\agent}{\world}, \; \phi^{\K,\world'}=\Tr 
				\\ &\text{and } (C_G\phi)^{\K,\world'}=\Tr
				\\  \iff & \forall A\in G,\; \forall \world'\in \evalLim{\agent}{\world}, \; \forall k\geq 1,\; (\knows{A}\phi)^{\K,\world'}=\Tr 
				\\ & \text{and }  (E_G^k\phi)^{\K,\world'}=\Tr
				\\  \iff & \forall A\in G, \; \forall k\geq 2,\; (E_G\phi)^{\K,\world}=\Tr 
				\\ & \text{and }  (\knows{A}(E_G^{k-1}\phi))^{\K,\world}=\Tr
				\\  \iff &  \forall k\geq 2,\; (E_G\phi)^{\K,\world}=\Tr, \text{ and }  (E_G^{k-1}\phi)^{\K,\world}=\Tr
				\\ \iff & (C_G\phi)^{\K,\world}=\Tr.
			\end{align*}
			\item Let $\world\in U$ such that $\phi^{\K,\world}=\Tr$. We say that a world $\otherworld\in U^{\omega}$ is \emph{$1$-reachable from} \world if $\otherworld\in R^\world_A$ for some $A\in G$. For $k>1$, we say that a world $\otherworld\in U^{\omega}$ is \emph{$k{+}1$-reachable from} \world if \otherworld is 1-reachable from a world $\otherworld' $ that is $k$-reachable from \world. 
			
			We show by induction on $k$ that for all $k\in\mathbb{N}_{\geq 1}$, for all $k$-reachable worlds \otherworld from \world, $(\psi\land\varphi)^{\K,\otherworld}=\Tr$. In particular, this proves that $E^k_G(\psi)^{\K,\world}=\Tr$ for all $k>0$, \ie $(C_G\psi)^{\K,\world}=\Tr$.
			\begin{compactitem}
				\item $k=1$.  Let \otherworld be a world 1-reachable from \world. By hypothesis, $\phi^{\K,\world}=\Tr$ and  $ \varphi\models E_G(\varphi\wedge\psi)$, thus $(E_G(\varphi\wedge\psi))^{\K,\world}=\Tr$. By definition of $E_G$ and of 1-reachability, we must have $(\varphi\wedge\psi)^{\K,\otherworld}=\Tr$, as desired.
				\item $k=k'+1$. Let \otherworld be a world k-reachable from \world. By definition of $k'$-reachability, there exists a world $\otherworld'$ $k'$-reachable from \world such that \otherworld is 1-reachable from $\otherworld'$. By induction hypothesis, $(\varphi\wedge\psi)^{\K,\otherworld'}=\Tr$, and $\varphi^{\K,\otherworld'}=\Tr$ in particular. Thus, by hypothesis, we get $(E_G(\varphi\wedge\psi))^{\K,\otherworld'}=\Tr$. By the definition of the valuation of the operator $E_G$, it follows that $(\varphi\wedge\psi)^{\K,\otherworld}=\Tr$, as desired.
			\end{compactitem}
		\end{compactenum}
	\end{proof}
}

Properties \ref{Prop}, \ref{K}, \ref{Nec}, \ref{M}, and \ref{O} of \cref{thm:properties} ensure that the semantics properly captures the intended meaning of the only-knowing operator. In more detail, properties \ref{Prop}, \ref{K}, \ref{Nec}, and \ref{M}  imply that for every $\phi$ and $\psi$, either $\onlyKnows{A} \phi \limplies \knows{A}\psi$ or $\onlyKnows{A} \phi \limplies \lnot \knows{A}\psi$. 
More specifically, the former holds when $\psi$ is entailbed by $\phi$, the latter otherwise. This means that $\onlyKnows{A} \varphi$ completely determines the agent's knowledge. Property \ref{O} states that for any formula $\phi\in\cocael$, there is a world in which $\onlyKnows{A} \phi$ holds, i.e.\ it is 
possible that an agent knows $\phi$ and knows nothing beyond $\phi$.

The following example discusses the construction of a world $w$ that shows that $\onlyKnows{A} \neg C_G p$ is satisfiable (in line with item~8 of Theorem~\ref{thm:properties}), something that previous attempts at combining only knowing and common knowledge failed at, even though it should intuitively be the case.

\begin{example}\label{ex:cgp_satisfied}
	Consider the setting of \cref{ex:compl_notpostcompl_omegabiw}. We want to construct a world \world that satisfies $\onlyKnows{A} \neg C_G p$, where $G=\{A\}$. By the definition of the only knowing operator, we must have
	\vspace{-4.5mm}
	\begin{align*}
		{\agent}^{\world}=&\{\world'\in\WW^{\omega^2}\mid (C_Gp)^{\world'}=\Fa\}
		\\ \bar{\agent}^{\world}=&\{\world'\in\WW^{\omega^2}\mid (C_Gp)^{\world'}=\Tr\}.
	\end{align*}
	\vspace{-6mm}\\
	We first have to find the above sets.
	
	Notice that, by \cref{thm:resolvingmodaldepth}, we have $\possible{\agent}{\world}\cup\impossible{\agent}{\world}=\WW^{\omega^2}$, as supposed. By \cref{prop:respower}, \cref{thm:resolvingmodaldepth}, and the fact that any \biworld{\omega^2} is the extension of some \biworld{\omega}, we can reduce to finding the set of \biworlds{\omega} satisfying $C_Gp$. It is not hard to see that $v$ from \cref{ex:compl_notpostcompl_omegabiw} satisfies  $C_Gp$. Moreover, since $v$ is completed, by Corollary~\ref{corollary:completed_char}, $v$ has a unique extension $v'$ to depth $\omega^2$, and $v'$. In an analogous fashion, we can build another \biworld{\omega^2} $u'$ satisfying $C_Gp$ by considering the unique extension of a completed \biworld{\omega} $u:=(u_\alpha)_{\alpha<\omega}$ defined as
	\vspace{-1.5mm}
		\begin{equation*}
		u_\alpha:=\begin{cases}
			\{\emptyset\} & \text{if }\alpha=0
			\\ (u_0, \{v_{\alpha'}\}, \WW^{\alpha'}) & \text{if }\alpha=\alpha'+1
		\end{cases}
	\end{equation*}
	\vspace{-3mm}\\
	Intuitively, both $v'$ and $u'$ are worlds in which $p$ is common knowledge (for the only agent \agent), but $p$ is true in the objective world of $v'$ and false in the one of $u'$. In particular, in $u'$ the agent \agent is not truthful. 
	
	We claim\footnote{The claim can be proven by induction, and it can already be seen to hold true by writing down all the eighteen \biworlds{1}, and reasoning on the conditions that make $C_Gp$ satisfied. We omit this reflection for the sake of conciseness.} $v$ and $u$ are the only \biworlds{\omega} satisfying $C_Gp$. Hence, $A^w$ must be $\WW^{\omega^2}\setminus\{v',u'\}$ and $\bar A^w$ must be $\{v', u'\}$,
and we define $w:=(\{p\},\WW^{\omega^2}\setminus\{v',u'\},\{v', u'\})$.

Notice that by \cref{thm:O_satisfiable}, $w$ is unique up to change of objective world. In other words, the world defined as $w$ but with $\{\emptyset\}$ as objective world $w_0$ satisfies $\onlyKnows{A} \neg C_Gp$ too.
\end{example}

\ignore{
\section{Truthfulness and Introspection}\label{sec:TIworlds}
In this section we present a modification of the semantics of $\cocael$ that is based on the presupposition that agents are truthful and (both positively and negatively) introspective. Treating this as a presupposition of the logic amounts to treating it as common knowledge common knowledge that all agents satisfy these properties. In a similar way one could also incorporate only some of these properties into the logic, e.g.\ only positive and negative introspection for a formalization of belief rather than knowledge.

In this section we focus on the two-valued valuation and therefore work only with the elements of the Kripke structure $\K$, \ie the worlds. We start with a definition of truthfulness as well as positive and negative introspection:
\begin{definition}
    Let $w$ be a world. 
    \begin{compactitem}
     \item $w$ is called \emph{truthful} if for all $\agent\in\A$, $\world\in\evalLim{\agent}{\world}$;
     \item $w$ is called \emph{positively introspective} if for each $\agent\in\A$, for each $\world'\in\evalLim{\agent}{\world}$, $\evalLim{\agent}{\world'}\subseteq \evalLim{\agent}{\world}$;
     \item $w$ is called \emph{negatively introspective} if for each $\agent\in\A$, for each $\world'\in\evalLim{\agent}{\world}$, $ \evalLim{\agent}{\world}\subseteq \evalLim{\agent}{\world'}$.
    \end{compactitem}
\end{definition}

The following theorem establishes that the definitions of these three properties imply the expected characterization of these properties in terms of $\cocael$ formulas:

\begin{theorem}
	Let \world be a world, $\agent \in\A$ an agent, and $\phi \in \cocael$ a formula. Then the following statements hold:
	\begin{compactenum}
		\item If \world is truthful, then $\world$ satisfies  $\knows{A}\phi \Rightarrow \phi $.
		\item If \world is positively introspective, then $\world$ satisfies $\knows{A}\phi \Rightarrow \knows{A}\knows{A}\phi$.
		\item If \world is negatively introspective, then $\world$ satisfies $\neg \knows{A}\phi \Rightarrow \knows{A}\neg \knows{A}\phi$.
	\end{compactenum}
\end{theorem}
\extended{
\begin{proof}
		\begin{compactenum}
			\item Suppose \world is truthful. If $(\knows{A} \phi)^{\K,w}=\Tr$, then $\phi^{\K,\world'}=\Tr$ for all $\world'\in\evalLim{\agent}{\world}$. Since $w$ is truthful, $\world\in\evalLim{\agent}{\world}$. Hence, $\phi^{\K,\world}=\Tr$. 
			\item Suppose \world is positively introspective. If $ (\knows{A}\phi)^{\K,\world}=\Tr$, then for each $\world'\in\evalLim{\agent}{\world}$, $\phi^{\K,\world'}$. Let $\world'  \in\evalLim{\agent}{\world}$ and $\world'' \in \evalLim{\agent}{\world'}$. By positively introspection of \world,  $\world''\in\evalLim{\agent}{\world}$. Hence, $\phi^{\K,\world''}$. Since $\world'$ and $\world''$ are chosen arbitrarily, we get $(\knows{A} \knows{A} \phi)^{\K, \world}=\Tr$.
			\item Suppose \world is negatively introspective. If $(\neg \knows{A}\phi)^{\K,\world}=\Tr$, then there exists $\world'\in\evalLim{\agent}{\world}$, such that $\phi^{\K,\world'}=\Fa$. By negative introspection of \world, for each  $\world''\in\evalLim{\agent}{\world}$, $\world'\in\evalLim{\agent}{\world}\subseteq\evalLim{\agent}{\world''}$. Hence,  $(\knows{A} \neg \knows{A} \phi)^{\K,\world}$, as desired.
		\end{compactenum}
\end{proof}
}

The properties listed in \cref{thm:properties} hold true in the Kripke structure $\K$, whereas the truthfulness, positive introspection, and negative introspection of every agent are not always true for a world of $\K$. We would like to restrict the Kripke structure $\K$ to a substructure $\K_\mathit{TI}$ which contains only worlds in which all agents are truthfull, positively introspective, and negatively introspective.

\begin{definition}[Truthful-introspective world] 
	A world is called \emph{truthful-introspective} if it is truthful, positively introspective, and negatively introspective.
\end{definition}

\begin{definition}[Recursively truthful-introspective world] \label{def:recur_TI_worlds}
	A world $w$ is called \emph{recursively truthful-introspective} if $w$ is truthful-introspective and all the worlds that are reachable from $w$ through the union $\bigcup_{A\in\A} R_A$ of all accessibility relations are truthful-introspective.
\end{definition}

We now define $\KripkeStructureTI:=(U^{\mathit{TI}}, (R^\mathit{TI}_A)_{A\in\A})$ to be the Kripke substructure of $\K$, where the underlying world set $U^\mathit{TI}$ is the set of all recursively truthful-introspective worlds, and the accessibility relations are just the ones coming from  $\K_\mathit{TI}$ being a substructure, \ie for all $\agent\in\A$, $R^\mathit{TI}_A:=R_A\cap (U^\mathit{TI}\times U^\mathit{TI})$. 	

Now we can define a modified entailment relation that takes into account truthfulness and introspection:

\begin{definition}
	For $\phi\in \cocael$ and $\Gamma\subseteq\cocael$, we write $\Gamma \entailsti \phi$ if $\phi^{\KripkeStructureTI, \world}=\Tr$ for every world $\world$  such that $\psi^{\KripkeStructureTI,\world}=\Tr$ for all $\psi\in\Gamma$.
\end{definition}

Note that incorporating truthfulness and introspection into the logic in this way modifies the meaning of the modalities $M_A$ and $\onlyKnows{A} $, because knowing at most $\phi$ means that any $\psi$ that one knows must be entailed by $\phi$ and the entailment relation is different now that truthfulness and introspection are hard-coded into the logic: $\psi$ may be entailed by $\phi$ together with truthfulness and introspection even though it was not entailed by $\phi$ itself in the original semantics of $\cocael$ without truthfulness and introspection. In other words, when $\phi \entailsti \psi$ but $\phi \not \models \psi$, then $M_A \phi \not \entailsti \neg \knows{A} \phi$ even though $M_A \phi \models \neg \knows{A} \phi$.

\marcos{If space is an issue, we could remove the following two propositions and the paragraphs preceding them from the paper.}

The following proposition ensures that for a world $w$ in $\KripkeStructureTI$, the worlds that are $R^\mathit{TI}_A$-accessible from $w$ are precisely the same as the worlds that are $R_A$-accessible from~$w$:

\begin{proposition}\label{prop:good_worlds}
	Let $\world\in U$, $\agent\in\A$, and $\world'\in\evalLim{\agent}{\world}\cap U$. If $\world$ is recursively truthful-introspective, then so is $\world'$.
\end{proposition}
\extended{
\begin{proof}
	Since $\world'$ is reachable by $\world$ in $\K$, every world reachable by $\world'$ is also reachable by $\world$. Hence, every world reachable by $\world'$ in $\K$ is truthful-introspective, as desired.
\end{proof}
}

As in standard Kripke semantics for epistemic logic, truthfulness together with positive and negative introspection means that the accessibility relations are equivalence relations:


\begin{proposition}
	The accessibility relations of $\K_\mathit{TI}$ are equivalence relations.
\end{proposition}
\extended{
\begin{proof}
	Let $\world, u, v\in U^\mathit{TI}$, and $\agent\in\A$.
	\begin{compactitem}
		\item \underline{Reflexivity}: by truthfulness of $\world$,  $w\in\evalLim{\agent}{\world}$, \ie  $(\world, \world)\in R^\mathit{TI}_\agent$.
		\item \underline{Simmetry}: suppose $(\world,\otherworld)\in R^\mathit{TI}_\agent$. By negative introspection and truthfulness of $\world$, $\world\in\evalLim{\agent}{\world}\subseteq \evalLim{\agent}{v}$. Hence, $(v,\world)\in R^\mathit{TI}_\agent$. 
		\item \underline{Transitivity}: suppose $(u, v), (v,w)\in R^\mathit{TI}_\agent$. This means $v\in\evalLim{\agent}{u}$ and $w\in\evalLim{\agent}{v}$. By positive introspection of $u$, $\evalLim{\agent}{v}\subseteq \evalLim{\agent}{u}$. Hence, $w\in\evalLim{\agent}{u}$, \ie $(u,w)\in R^\mathit{TI}_\agent$, as desired.
	\end{compactitem}
\end{proof}
}


The following theorem establishes that the properties that we established for $\models$ in Theorem~\ref{thm:properties} also hold for $\entailsti$:

\marcos{I'm pretty sure that (O) also holds for $\entailsti$, so it should be included for completeness. However, as far as I can see, proving (O) for $\entailsti$ requires defining TI-biworlds (in which truthfulness and introspection are built in from the start) and reproving a lot of the results about biworlds, the three-valued valuation and its relation to the two-valued valuation for this semantics with truthfulness and introspection built in from the start. Do you see any possibilities for shortcuts? If not, how should we proceed about this issue?}

\begin{theorem}
		Let $\phi, \psi \in \cocael$ be two formulas, $\Gamma, \Gamma'\subset \cocael$ two sets of formulas, $\agent \in \A$ an agent, and $G\subseteq\A$ a non-empty set of agents. Then, the following properties hold: 
	\begin{compactenum}
		\item \label{PropTI} (Prop) For each propositional tautology $\phi$, we have $\entailsti \phi$.
		\item \label{MPTI} (MP) $\phi, \phi \Rightarrow \psi \entailsti \psi$.
		\item  (Mono) If $\Gamma \entailsti \phi$, then $\Gamma, \psi \entailsti \phi$. 
		\item \label{CutTI} (Cut) If $\Gamma \entailsti \phi $ and $\Gamma', \phi \entailsti \psi$, then $\Gamma, \Gamma' \entailsti\psi$.
		\item \label{KTI} (K) $\entailsti (\knows{A}(\phi\Rightarrow\psi)\land \knows{A}\phi)\Rightarrow \knows{A}\psi$.
		\item \label{NecTI} (Nec) If $\entailsti\phi$, then $\entailsti \knows{A}\phi$.
		\item \label{MTI} (M) If $\phi\not\entailsti\psi$, then $ M_A\phi \entailsti \neg \knows{A}\psi$.
		\item \label{fixed_point_axiomTI} (Fixed point axiom) $\entailsti C_G \phi \iff E_G(\phi\wedge C_G\phi)$.
		\item \label{Induction_ruleTI} (Induction rule) If $ \phi\entailsti E_G(\phi\wedge\psi)$, then $ \phi\entailsti C_G\psi$.
	\end{compactenum}
\end{theorem}
\extended{
	\begin{proof}
		\begin{compactenum}
			\item Clear by \cref{Prop} of \cref{thm:properties} and the fact that $\KripkeStructureTI$ is a substructure of $\K$.
			\item Analogous to the proof of \cref{MP} of \cref{thm:properties}.
			\item Clear by the definition of logical consequence relation.
			\item Analogous to the proof of \cref{MP} of \cref{thm:properties}.
			\item Clear by \cref{K} of \cref{thm:properties} and the fact that $\K_\mathit{TI}$ is a substructure of $\K$.
			\item Let $\world\in U^\mathit{TI}$. By definition of $\KripkeStructureTI$, $(\knows{A}\phi)^{\K_\mathit{TI},\world}=(\knows{A}\phi)^{\K,\world}=glb_{{\leq_t}} \{\phi^{\K,w'}|w'\in R_A^\world\} $. Since $\phi^{\K_\mathit{TI},\world'}=\Tr$ for any $\world'\in U^\mathit{TI}$ by hypothesis, and $R_A^\world\subseteq U^\mathit{TI}$, we have $(\knows{A}\phi)^{\K_\mathit{TI},\world}=\Tr$, as desired.
			\item  Analogous to the proof of \cref{M} of \cref{thm:properties}. 
			\item Clear by \cref{fixed_point_axiom} of \cref{thm:properties} and the fact that $\K_\mathit{TI}$ is a substructure of $\K$.
			\item By \cref{prop:good_worlds}, the proof is analogous to the one of \cref{Induction_rule} in \cref{thm:properties}.
		\end{compactenum}
	\end{proof}
}

}

\section{Truthfulness and Introspection} \label{sec:TIworlds}
In the previous three sections we have defined and described the construction of a structure of worlds that is rich enough to allow to formally define the semantics of only knowing and common knowledge in a way that matches basic intuitions about these logical modalities in a precisely specified way. A major challenge in designing this construction was to ensure that we have enough worlds to describe all logically possible epistemic states. For this reason, we decided to keep the construction as general as possible, \ie not to unnecessarily limit the set of worlds.

However, there are certain properties of the knowledge modality that are often taken for granted in epistemic logic and that require limiting the set of worlds. In particular, the following properties are often assumed to hold:

\begin{compactitem}
		\item Truthfulness: $\knows{A}\phi \Rightarrow \phi$ is satisfied in every world for every agent $A$.
		\item Positive introspection: $\knows{A}\phi \Rightarrow \knows{A} \knows{A}\phi$ is satisfied in every world for every agent $A$.
		\item Negative introspection: $\neg \knows{A}\phi \Rightarrow \knows{A}\neg \knows{A}\phi$ is satisfied in every world for every agent $A$.
\end{compactitem}

While all three of these properties are commonly assumed in epistemic logic, there are specific issues about ensuring the first or the third in a logic with a modality $\onlyKnows{A}  \phi$ for only knowing or a modality $M_A \phi$ for knowing at most.

In the case of truthfulness, there is an issue concerning the formula $\onlyKnows{A}  (\knows{A} p \lor \knows{A} q)$. 
Given the definition of the modality $\onlyKnows{A}$, this entails $\knows{A} (\knows{A} p \lor \knows{A} q)$, which by truthfulness entails $\knows{A} p \lor \knows{A} q$, so either $\knows{A} p$ or $\knows{A} q$ has to be true. But since $A$ only knows $\knows{A} p \lor \knows{A} q$ and since $\knows{A} p \lor \knows{A} q$ entails neither $p$ nor $q$, neither $p$ nor $q$ can be known, a contradiction. Thus $\onlyKnows{A}  (\knows{A} p \lor \knows{A} q)$ cannot be satisfiable in a logic with truthfulness, which means that principle~(O) from Theorem~\ref{thm:properties} cannot hold in such a logic (assuming principles (Prop), (Cut) and (M) do hold).

In the case of negative introspection, a more severe problem arises. Suppose $M_A q$ is true in some world $w$. Since $q$ does not entail $p$, by (M) this should entail that $\neg \knows{A} p$ is true in $w$, so by negative introspection, $\knows{A} \neg \knows{A} p$ is true in $w$. But since $q$ does not entail $\neg \knows{A} p$, principle (M) also allows us to conclude that $\neg \knows{A} \neg \knows{A} p$ is true in $w$, a contradiction. Thus $M_A q$ is not satisfiable, and similarly no formula of the form $M_A \phi$ or $\onlyKnows{A}  \phi$ is satisfiable for any satisfiable $\phi$. It should be stressed that this problem is not caused by our semantic approach to only knowing, but is a direct consequence of basic properties that only knowing has been assumed to satisfy also in other papers.
\begin{remark} \label{remark:wrong}
	A reader familiar with the literature on only knowing, of which we give an overview in Section~\ref{sec:relatedWork}, may wonder why this problem was not identified in previous papers, e.g.\ \citet{ai/BelleL15}, who define a semantics for only knowing in a logic with negative introspection. What the authors did not realize is that by enforcing negative introspection in their logic, they actually made all statements of the form $\onlyKnows{A}  \phi$ unsatisfiable (for any satisfiable formula $\phi$). They wrongly claim on page~5 that $O_i (p \land C p)$ is satisfiable, but the alleged proof is wrong. If one defines $V^1  =\{w \mid p,q \in w, w \in \mathcal{W}\}$, $V^{k+1} = \{(w,V^k,\dots,V^k) \mid w \in V^1\}$, $f'(i,k) = V^k$ and $w' = \{p\}$, then one can easily see that $f' \notin f_i^{w'}$ and $f',w' \models p \land C p$, contradicting their claim that $f,w \models O_i (p \land Cp)$.
\end{remark}
 In order to avoid this problem, one would need to make use of autoepistemic logic \mycite{AEL}, or some multi-agent version thereof \cite{ijcai/Lakemeyer93,wocfai/PermpoontanalarpJ95,VlaeminckVBD/KR2012,ijcai/HertumCBD16} in the definition of the semantics of $M_A \phi$: Intuitively, $M_A \phi$ should be true in a world in which all worlds that satisfy all formulas entailed by the autoepistemic theory $\{\phi\}$ are accessible. 
 This would give rise to a logic in which a variant of principle~(M) with a negated autoepistemic entailment in the place of the negated entailment is satisfied. We leave it to future work to develop the details of such a theory and investigate whether it behaves as intended.

Given that it is somewhat problematic to incorporate truthfulness and negative introspection in a logic with only knowing, whereas no similar problems arise for positive introspection, we describe how the semantic framework from the previous sections can be used to define a logic of only knowing and common knowledge in which positive introspection is ensured.

\begin{definition}
 A world $w$ is called \emph{positively introspective (PI)} if for each $\agent\in\A$ and for any worlds 
   $w', w''$, $w'' R_A w' R_A w$, implies $w'' R_A w$. A world $w$ is called \emph{recursively PI} if $w$ is positively introspective and all the worlds that are reachable from $w$ through the union $\bigcup_{A\in\A} R_A$ of all accessibility relations are positively introspective.
\end{definition}

We now define $\K_{\mathit{PI}}:=(U^{\mathit{PI}}, (R^\mathit{PI}_A)_{A\in\A})$ to be the Kripke substructure of $\K$, where the underlying world set $U^\mathit{PI}$ is the set of all recursively PI worlds, and the accessibility relations are just the ones coming from  $\K_\mathit{PI}$ being a substructure, \ie for all $\agent\in\A$, $R^\mathit{PI}_A:=R_A\cap (U^\mathit{PI}\times U^\mathit{PI})$. 	

Now, we can define a modified entailment relation that takes into account positive introspection:

\begin{definition}
	For $\phi\in \cocael$ and $\Gamma\subseteq\cocael$, we write $\Gamma \entailsti \phi$ if $\phi^{\K_{\mathit{PI}}, \world}=\Tr$ for every world $\world$  such that $\psi^{\K_{\mathit{PI}},\world}=\Tr$ for all $\psi\in\Gamma$.
\end{definition}

Note that incorporating 
positive introspection into the logic in this way modifies the meaning of the modalities $M_A$ and $\onlyKnows{A} $, because knowing at most $\phi$ means that any $\psi$ that one knows must be entailed by $\phi$ and the entailment relation is different now that truthfulness and introspection are hard-coded into the logic: $\psi$ may be entailed by $\phi$ together with truthfulness and introspection even though it was not entailed by $\phi$ itself in the original semantics of $\cocael$ without truthfulness and introspection. In other words, when $\phi \entailsti \psi$ but $\phi \not \models \psi$, then $M_A \phi \not \entailsti \neg \knows{A} \phi$ even though $M_A \phi \models \neg \knows{A} \phi$.

The following theorem establishes that positive introspection does indeed hold in this logic and that the properties that we established for $\models$ in Theorem~\ref{thm:properties} also hold for $\entailsti$. 

\begin{theorem}
\label{thm:PI}
		Let $\phi, \psi \in \cocael$ be two formulas, $\Gamma, \Gamma'\subset \cocael$ two sets of formulas, $\agent \in \A$ an agent, and $G\subseteq\A$ a non-empty set of agents. Then, the following properties hold: 
	\begin{compactenum}
        \item[1-10.] All properties mentioned in Theorem~\ref{thm:properties} with $\models$ replaced by $\entailsti$.
        \item[11.] (PI) $\entailsti \knows{A}\phi \Rightarrow \knows{A} \knows{A}\phi$ \bart{THIS SHOULD BE LAST. To make numbering the same}
	\end{compactenum}
\end{theorem}
\extended{
	\begin{proof}
	 All properties apart from (PI) and (O) can be established in a way analogous to the proof of Theorem~\ref{thm:properties}. 
	 \begin{compactitem}
	  \item (PI) Suppose \world is positively introspective and suppose $ (\knows{A}\phi)^{\K_{\mathit{PI}},\world}=\Tr$. Suppose $w'' R_A w' R_A w$. Since \world is positively introspective, it follows that $w'' R_A w$, \ie that 
	  $ \phi^{\K_{\mathit{PI}},\world''}=\Tr$. Since $w''$ was arbitrary, we can conclude that $(\knows{A} \phi)^{\K_{\mathit{PI}},\world'}=\Tr$. Since $w'$ was arbitrary, this in turn implies that $(\knows{A} \knows{A} \phi)^{\K_{\mathit{PI}},\world}=\Tr$, as required.
	  \item (O') For establishing this result, we need a result similar to the first statement of Theorem~\ref{thm:O_satisfiable}: For every agent $\agent$ and every formula $\phi \in \cocael$, 
	  there is a world $w \in U^\mathit{PI}$ such that $(\onlyKnows{A}  \phi)^{\K_{\mathit{PI}},w} = \Tr$.
	  
	  Let $v$ be an arbitrary $0$-biworld.
		\begin{align*}
			W &:= \{ u|_{\omega^2} \mid u \in U^\mathit{PI} \textnormal{ and } (\phi \land \knows{A} \phi)^{\K,u} = \Tr \},\\
			\overline{W} &:= \WW^{\omega^2} \setminus W,\\
			w &:= (v,(W)_{A \in \A},(\overline{W})_{A \in \A}).
		\end{align*}
		Clearly $w$ is a world. In order to show that $w \in U^\mathit{PI}$, we need to show that $w$ is recursively PI. From the definition of $w$ it is easy to see that for this, it is enough to show that $w$ is positively introspective.
		
		So suppose $A \in \A$ and $w'' R_A w' R_A w$. We need to show that $w'' R_A w$. 
		The fact that $w' R_A w$ means that $w'|_{\omega^2} \in W$, \ie $(\knows{A} \phi)^{\K_{\mathit{PI}},w'} = \Tr$. Since $w'' R_A w'$, it follows that $\phi^{\K_{\mathit{PI}},w''} = \Tr$. Since $w'$ is positively introspective, the fact that $(\knows{A} \phi)^{\K_{\mathit{PI}},w'} = \Tr$ implies that $(\knows{A} \knows{A} \phi)^{\K_{\mathit{PI}},w'} = \Tr$, \ie that $(\knows{A} \phi)^{\K_{\mathit{PI}},w''} = \Tr$. So $(\phi \land \knows{A} \phi)^{\K_{\mathit{PI}},w''} = \Tr$, \ie $w''|_{\omega^2} \in W = A^w$, \ie $w'' R_A w$, as required.
		
		It follows directly from the definition of $w$ 
		that $(\knows{A} \phi)^{\K_{\mathit{PI}},v'} = \Tr$ and that $(M_A \phi)^{\K_{\mathit{PI}},v'} = \Tr$, \ie that $(\onlyKnows{A}  \phi)^{\K_{\mathit{PI}},v'} = \Tr$.
	 \end{compactitem}

	\end{proof}
}


Given that we could develop a theory with positive introspection, one may wonder what happens if one tries to similarly add negative introspection and/or truthfulness. Due to the problems with negative introspection described above, naively adding negative introspection in this way will yield a logic in which $M_A \phi$ is not satisfiable for any $\phi$. Adding truthfulness, on the other hand, does not cause such problems. Indeed, truthfulness and positive introspection can meaningfully be added together in a way similar to how positive introspection was added in the above definitions, yielding an entailment relation $\models_\mathit{TPI}$. Apart from (O), the properties of Theorem~\ref{thm:PI} will still hold. We conjecture that the following weakening of (O') holds in this context:

(O') If $\phi$ is a formula not involving a modality with subscript $A$, then $\onlyKnows{A}  \phi \not \models_\mathit{TPI} \bot$.

The proof of this conjecture is left to future work.

\section{Related Work}\label{sec:relatedWork}

In this paper, we have studied the interplay of common knowledge and only knowing.
 The former concept is quite well-known and has been extensively studied \cite{FaginHMV95,book/MeyerH95} since its first mentions in the philosophical \cite{book/Lewis69}, and the mathematical literature \cite{as/Aumann76}. 
 The latter concept is younger and has been studied intensively more recently. 
 
 \citet{ai/Levesque90} was among the first to introduce the notion of only knowing\footnote{There are several closely-related notions, like ignorance  \cite{nato/HalpernM85,aaai/Konolige82}, minimal knowledge \cite{sLogica/HoekT02}, and total knowledge \cite{aaai/Pratt-Hartmann00}.} by presenting a single-agent logic of belief extended with a novel operator $O$ expressing that the agent's beliefs are exactly the ones implied by the knowledge base and nothing more.
  He intended his logic of only knowing to capture certain types of non-monotonic reasoning patterns, like autoepistemic logic (AEL)~\cite{mo85}. 
 In the 1990s and early 2000s,  single-agent only knowing was successfully studied and implemented \cite{ai/HalpernL95,ai/Rosati00,LevesqueL00}, and \citet{aaai/LakemeyerL05} further revealed its potential by showing that it is also possible to capture default logic (DL) \cite{ai/Reiter80} and a variant of AEL proposed by \citet{Konolige88}. 
 
\marcos{For journal version, we need to reconsider the following paragraph, which seems unplausible after the recent discovery of a big problem in combining negative introspection with only knowing:

 Only knowing does not present particular problems in its formalization if the agent is assumed to be positive and negative introspective. Analogously, assuming mutual introspection for the multi-agent case would allow to avoid any major issue. However, while this assumption about introspection is sensible for single-agent only knowing, it is rather unintuitive for most applications of multi-agent only knowing. Unfortunately, removing the assumption of mutual introspection brings  to the table a new set of issues.}
 
\citet{aaai/Halpern93} and  \citet{ijcai/Lakemeyer93} were among the first to extend Levesque's only knowing to a multi-agent setting. In a joint publication, they (\citeyear{logcom/HalpernL01}) improved upon their independent works with an axiom system satisfying all the desired properties for only knowing. Nevertheless, the proposed axiomatization forces to include directly in the language a validity operator and the resulting semantics is not ``as natural as we might like'', according to the authors. 
 In the early 2000s, first \citet{aiml/Waaler04} alone and then together with Solhaug \cite{tark/WaalerS05} tried another route to generalize Levesque's axioms, without encoding the notion of validity into the language itself. 
 Yet, once again these models and the logic itself feel complex and unnatural.
 
 An in-depth analysis of the issues of all the previously mentioned works on multi-agent only knowing is provided by Belle and Lakemeyer ~(\citeyear{corr/BelleL10,ai/BelleL15}), together with a natural way to avoid such problems. 
 Belle and Lakemeyer proposed to use models limited to a finite depth $k$, \ie models at which beliefs can be nested at most $k$ times, by introducing the concept of $k$-structure to represent an agent's epistemic state. In this way, they successfully extend Levesque's logic, by keeping the original idea of Levesque's worlds and generalizing its features at the same time.
 However, they do not take into account common knowledge. 
 
 \citet{ijcai/AucherB15} proposed a novel formulation comprising both common knowledge and only knowing at level $k$. To achieve this result, they make use of the concept of so-called $k{+}1$-canonical formulas  \cite{jphil/Moss07}, written as conjunctions of a $0$-canonical formula and $k$-canonical formulas nested in knowledge and common knowledge operators. Even though the proposed pointed epistemic models can be fully characterized up to modal depth $k$ by a $k$-canonical formula, such representation at depth $k$ might feel slightly unnatural, as it characterizes the knowledge of an agent only up to level $k$, but it also determines the common knowledge of a group of agents, which is infinitary in nature. 
 When disregarding the conjunction of the common knowledge operators in such formulas, one can see a correspondence between our \biworld{k}s and the proposed $k$-canonical formulas. 
 However, there is an important limitation to this work. 
 Namely, while our $\onlyKnows{\agent}$ operator fully captures the knowledge and ignorance of an agent $\agent$ at any depth, the $\onlyKnows{\agent}^n$ operator proposed by \citet{ijcai/AucherB15} expresses only knowing for an agent $\agent$ just up to level $n$ in the sense that an expression of the form $\onlyKnows{\agent}^1 \phi$ should be read as ``If I disregard all my knowledge deeper than level one, I only know $\phi$''. 
 Formally, as Aucher and Belle pointed out, for any two modal depths $n>m$, if agent $\agent$ only knows $\phi$ at depth $n$, then $\agent$ only knows $\phi$ at depth $m$, \ie  $\onlyKnows{\agent}^n \to \onlyKnows{\agent}^m$, but the reverse implication might not hold.  
 
The limitation of the only knowing operator to a finite depth is overcome in the same year by  \citet{ijcai/BelleL15}.  
They devised an alternative approach to bring common knowledge and only knowing together, the semantic structures of which are close to our $\omega$-biworlds (except that there is no representation of the set $\evalLimBar{A}{\omega}$). 
However, as explained in Remark~\ref{remark:wrong}, due to hard-coding negative introspection in their logic, all non-trivial formulas of the form $O_A\phi$ are unsatisifiable, which goes against what we expressed as Property \ref{O} in \cref{thm:properties}. Hence, such a semantics does not properly capture the intended meaning of the only-knowing operator.


\citet{phd/VanHertum16} observed that in \cite{ijcai/BelleL15}, the formula $\onlyKnows{A}\neg Cp$ cannot be satisfied, and this holds even if negative introspection is dropped (to avoid the problem mentioned in Remark~\ref{remark:wrong}). \citet{phd/VanHertum16} attempted to overcome this problem by proposing four different semantics, but he does not entirely fullfill his purpose. In more detail, two of the proposed semantics (Section 6.3 of \cite{phd/VanHertum16}) do not allow for arbitrary nesting of the only knowing operator and thus they do not cover the whole $\cocael$ language. A third semantics (Section 6.2.2 of \cite{phd/VanHertum16}) seems to solve the problem regarding the satisfiability of formulas like $\onlyKnows{A} \neg Cp$, but it is not precision-monotonic.
This is rather problematic: in a given $\mu$-world in such semantics, $O_A \phi$ might be true, but if we give more precise information (extending it to depth $\mu+1$), the formula might become false. Hence, one could only define what an agent only knows ``up to a certain depth'', as in \cite{ijcai/AucherB15}. \marcos{I feel like this problem is a problem inside of Pieter's theoretical framework, but it should be possible to state a problem caused by it that is accessible to people who do not know about the details of Pieter's approach like the precision ordering between worlds.}
Finally, the semantics presented in Section 6.2.1 makes use of $\lambda$-canonical Kripke structures, for any fixed limit ordinal $\lambda$, to define a two-valued valuation for $\cocael$.  
In particular, if $\lambda=\omega$, then the proposed semantics corresponds to the semantics of \citet{ijcai/BelleL15} with some minor adjustments. 
 The hope of the author was that by choosing a large enough limit ordinal $\lambda$, the satisfiability issue would be solved. Unfortunately, there exists formulas, such as $\onlyKnows{A} \neg \onlyKnows{A} \bot$ formalizing the statement ``all \agent knows is that their knowledge is consistent'', that are not satisfiable in any $\lambda$-canonical Kripke structure. Even though the formula just mentioned may seem like a corner case, its satisfiability guarantees something quite desirable, anmely that our semantic structures (the biworlds) are rich enough to allow the agents to identify precisely those worlds in which they have consistent knowledg
 As shown in Theorem \ref{thm:O_satisfiable}, such satisfiability problems do not occur with our definition of worlds. 
 
 \section{Conclusion}\label{sec:conclusion}
 
 We defined a multi-agent epistemic logic with common knowledge and only knowing operators, which successfully encodes these notions under the same framework. 
 
First, we introduced the novel concept of \biworld{\ordinal} for countable ordinals \ordinal, which approximates not only the worlds that an agent deems possible, but also those deemed impossible. This duality proved to be fundamental to successfully  deal with the only knowing operator in a multi-agent setting.
Moreover, we have shown that the proposed new definitions are indeed sensible, as they satisfy the properties one would expect (\cref{thm:four_properties}).
 
Second, we defined the language $\cocael$, extending propositional logic with the modal operators $\knows{A}$, $M_A$, $E_G$, and $C_G$, and a three-valued model semantics for it. Furthermore, we defined a canonical Kripke structure over completed \biworlds{\omega^2{+}1}, and the two-valued semantics obtained from it is shown to coincide with the model semantics. 
This allowed us to prove several desirable properties (\cref{thm:properties}) the resulting logic satisfies. In particular, we showed that for any formula $\varphi$, there is a unique state-of-mind of a given agent in which they only know $\varphi$.

Finally, we have considered how our framework can be extended to satisfy properties like truthfulness, positive introspection and negative introspection. For positive introspection we have shown some positive results, whereas for truthfulness and negative introspection, we have identified certain problems that arise when combining them with only knowing. We have motivated the need for further research related to these problems.

Another line of future work that we envisage is to generalize the construction of biworlds so that it becomes applicable to other areas of research. More concretely, this amounts to the construction of a set-theoretic universe in which there exists a universal set, similarly as in topological set theory. If one takes this alternative set theory as one's metatheory, the incorrect definition from the introduction of this paper could very easily be turned into a correct definition.

\conf{
  \newpage
 \todo{THIS COMMENT SHOULD BE AT THE START OF PAGE  10}
}
\section*{Acknowledgements}
We are grateful to Marc Denecker and Pieter Van Hertum for the fruitful discussions and feedback on earlier versions of this work. 
This work was partially supported by  Fonds Wetenschappelijk Onderzoek -- Vlaanderen (project G0B2221N) and by the Flemish Government under the ``Onderzoeksprogramma Artificiële Intelligentie (AI) Vlaanderen''

\bibliographystyle{kr}
\bibliography{idp-latex/krrlib}

\begin{thebibliography}{}

\bibitem[\protect\citeauthoryear{Aucher and Belle}{2015}]{ijcai/AucherB15}
Aucher, G., and Belle, V.
\newblock 2015.
\newblock Multi-agent only knowing on planet {K}ripke.
\newblock In Yang, Q., and Wooldridge, M.~J., eds., {\em Proceedings of the
  Twenty-Fourth International Joint Conference on Artificial Intelligence,
  {IJCAI} 2015, Buenos Aires, Argentina, July 25-31, 2015},  2713--2719.
\newblock {AAAI} Press.

\bibitem[\protect\citeauthoryear{Aumann}{1976}]{as/Aumann76}
Aumann, R.~J.
\newblock 1976.
\newblock Agreeing to disagree.
\newblock {\em The Annals of Statistics} 4(6):1236--1239.

\bibitem[\protect\citeauthoryear{Belle and Lakemeyer}{2010}]{corr/BelleL10}
Belle, V., and Lakemeyer, G.
\newblock 2010.
\newblock Multi-agent only-knowing revisited.
\newblock {\em CoRR} abs/1009.2041.

\bibitem[\protect\citeauthoryear{Belle and Lakemeyer}{2015a}]{ijcai/BelleL15}
Belle, V., and Lakemeyer, G.
\newblock 2015a.
\newblock Only knowing meets common knowledge.
\newblock In Yang and Wooldridge \shortcite{ijcai/2015},  2755--2761.

\bibitem[\protect\citeauthoryear{Belle and Lakemeyer}{2015b}]{ai/BelleL15}
Belle, V., and Lakemeyer, G.
\newblock 2015b.
\newblock Semantical considerations on multiagent only knowing.
\newblock {\em Artif. Intell.} 223:1--26.

\bibitem[\protect\citeauthoryear{Fagin \bgroup et al\mbox.\egroup
  }{1995}]{FaginHMV95}
Fagin, R.; Halpern, J.~Y.; Moses, Y.; and Vardi, M.~Y.
\newblock 1995.
\newblock {\em Reasoning About Knowledge}.
\newblock MIT Press.

\bibitem[\protect\citeauthoryear{Fagin, Halpern, and
  Vardi}{1991}]{jacm/FaginHV91}
Fagin, R.; Halpern, J.~Y.; and Vardi, M.~Y.
\newblock 1991.
\newblock A model-theoretic analysis of knowledge.
\newblock {\em J. {ACM}} 38(2):382--428.

\bibitem[\protect\citeauthoryear{Halpern and Lakemeyer}{1995}]{ai/HalpernL95}
Halpern, J.~Y., and Lakemeyer, G.
\newblock 1995.
\newblock Levesque's axiomatization of only knowing is incomplete.
\newblock {\em Artif. Intell.} 74(2):381--387.

\bibitem[\protect\citeauthoryear{Halpern and
  Lakemeyer}{2001}]{logcom/HalpernL01}
Halpern, J.~Y., and Lakemeyer, G.
\newblock 2001.
\newblock Multi-agent only knowing.
\newblock {\em J. Log. Comput.} 11(1):41--70.

\bibitem[\protect\citeauthoryear{Halpern and Moses}{1985}]{nato/HalpernM85}
Halpern, J.~Y., and Moses, Y.
\newblock 1985.
\newblock Towards a theory of knowledge and ignorance: Preliminary report.
\newblock In Apt, K.~R., ed., {\em Logics and Models of Concurrent Systems},
  volume~13 of {\em NATO ASI Series},  459--476.
\newblock Springer Berlin Heidelberg.

\bibitem[\protect\citeauthoryear{Halpern}{1993}]{aaai/Halpern93}
Halpern, J.~Y.
\newblock 1993.
\newblock Reasoning about only knowing with many agents.
\newblock In Fikes, R., and Lehnert, W.~G., eds., {\em Proceedings of the 11th
  National Conference on Artificial Intelligence. Washington, DC, USA, July
  11-15, 1993},  655--661.
\newblock {AAAI} Press / The {MIT} Press.

\bibitem[\protect\citeauthoryear{Hintikka}{1962}]{Hintikka1962-HINKAB}
Hintikka, J.
\newblock 1962.
\newblock {\em Knowledge and Belief}.
\newblock Ithaca: Cornell University Press.

\bibitem[\protect\citeauthoryear{Konolige}{1982}]{aaai/Konolige82}
Konolige, K.
\newblock 1982.
\newblock Circumscriptive ignorance.
\newblock In Waltz, D.~L., ed., {\em Proceedings of the National Conference on
  Artificial Intelligence, Pittsburgh, PA, USA, August 18-20, 1982},  202--204.
\newblock {AAAI} Press.

\bibitem[\protect\citeauthoryear{Konolige}{1988}]{Konolige88}
Konolige, K.
\newblock 1988.
\newblock On the relation between default and autoepistemic logic.
\newblock {\em Artif. Intell.} 35(3):343--382.

\bibitem[\protect\citeauthoryear{Lakemeyer and
  Levesque}{2005}]{aaai/LakemeyerL05}
Lakemeyer, G., and Levesque, H.~J.
\newblock 2005.
\newblock Only-knowing: Taking it beyond autoepistemic reasoning.
\newblock In Veloso, M.~M., and Kambhampati, S., eds., {\em Proceedings, The
  Twentieth National Conference on Artificial Intelligence and the Seventeenth
  Innovative Applications of Artificial Intelligence Conference, July 9-13,
  2005, Pittsburgh, Pennsylvania, {USA}},  633--638.
\newblock {AAAI} Press / The {MIT} Press.

\bibitem[\protect\citeauthoryear{Lakemeyer}{1993}]{ijcai/Lakemeyer93}
Lakemeyer, G.
\newblock 1993.
\newblock All they know: {A} study in multi-agent autoepistemic reasoning.
\newblock In Bajcsy, R., ed., {\em Proceedings of the 13th International Joint
  Conference on Artificial Intelligence. Chamb{\'{e}}ry, France, August 28 -
  September 3, 1993},  376--381.
\newblock Morgan Kaufmann.

\bibitem[\protect\citeauthoryear{Levesque and Lakemeyer}{2000}]{LevesqueL00}
Levesque, H.~J., and Lakemeyer, G.
\newblock 2000.
\newblock {\em The logic of knowledge bases}.
\newblock {MIT} Press.

\bibitem[\protect\citeauthoryear{Levesque}{1990}]{ai/Levesque90}
Levesque, H.~J.
\newblock 1990.
\newblock All {I} know: A study in autoepistemic logic.
\newblock {\em Artif. Intell.} 42(2-3):263--309.

\bibitem[\protect\citeauthoryear{Lewis}{1969}]{book/Lewis69}
Lewis, D.~K.
\newblock 1969.
\newblock {\em Convention: A Philosophical Study}.
\newblock Cambridge, MA, USA: Wiley-Blackwell.

\bibitem[\protect\citeauthoryear{Meyer and {van der
  Hoek}}{1995}]{book/MeyerH95}
Meyer, J.~C., and {van der Hoek}, W.
\newblock 1995.
\newblock {\em Epistemic logic for {AI} and computer science}, volume~41 of
  {\em Cambridge tracts in theoretical computer science}.
\newblock Cambridge University Press.

\bibitem[\protect\citeauthoryear{Moore}{1985}]{mo85}
Moore, R.~C.
\newblock 1985.
\newblock Semantical considerations on nonmonotonic logic.
\newblock {\em Artif. Intell.} 25(1):75--94.

\bibitem[\protect\citeauthoryear{Moss}{2007}]{jphil/Moss07}
Moss, L.~S.
\newblock 2007.
\newblock Finite models constructed from canonical formulas.
\newblock {\em J. Philos. Log.} 36(6):605--640.

\bibitem[\protect\citeauthoryear{Permpoontanalarp and
  Jiang}{1995}]{wocfai/PermpoontanalarpJ95}
Permpoontanalarp, Y., and Jiang, J.~Y.
\newblock 1995.
\newblock On multi-agent autoepistemic reasoning.
\newblock In {\em {WOCFAI}},  307--318.

\bibitem[\protect\citeauthoryear{Pratt{-}Hartmann}{2000}]{aaai/Pratt-Hartmann00}
Pratt{-}Hartmann, I.
\newblock 2000.
\newblock Total knowledge.
\newblock In Kautz, H.~A., and Porter, B.~W., eds., {\em Proceedings of the
  Seventeenth National Conference on Artificial Intelligence and Twelfth
  Conference on on Innovative Applications of Artificial Intelligence, July 30
  - August 3, 2000, Austin, Texas, {USA}},  423--428.
\newblock {AAAI} Press / The {MIT} Press.

\bibitem[\protect\citeauthoryear{Reiter}{1980}]{ai/Reiter80}
Reiter, R.
\newblock 1980.
\newblock A logic for default reasoning.
\newblock {\em Artif. Intell.} 13(1-2):81--132.

\bibitem[\protect\citeauthoryear{Rosati}{2000}]{ai/Rosati00}
Rosati, R.
\newblock 2000.
\newblock On the decidability and complexity of reasoning about only knowing.
\newblock {\em Artif. Intell.} 116(1-2):193--215.

\bibitem[\protect\citeauthoryear{van~der Hoek and
  Thijsse}{2002}]{sLogica/HoekT02}
van~der Hoek, W., and Thijsse, E.
\newblock 2002.
\newblock A general approach to multi-agent minimal knowledge: With tools and
  samples.
\newblock {\em Stud Logica} 72(1):61--84.

\bibitem[\protect\citeauthoryear{{Van Hertum} \bgroup et al\mbox.\egroup
  }{2016}]{ijcai/HertumCBD16}
{Van Hertum}, P.; Cramer, M.; Bogaerts, B.; and Denecker, M.
\newblock 2016.
\newblock Distributed autoepistemic logic and its application to access
  control.
\newblock In Kambhampati, S., ed., {\em Proceedings of the Twenty-Fifth
  International Joint Conference on Artificial Intelligence, {IJCAI} 2016, New
  York, NY, USA, 9-15 July 2016},  1286--1292.
\newblock {IJCAI/AAAI} Press.

\bibitem[\protect\citeauthoryear{{Van Hertum}}{2016}]{phd/VanHertum16}
{Van Hertum}, P.
\newblock 2016.
\newblock {\em New Language Constructs and Inferences for the Knowledge Base
  Paradigm: {A} Business and Multi-agent Perspective ; Taaluitbreidingen en
  nieuwe inferenties voor het kennisbank paradigma vanuit een Business en
  Multi-Agent perspectief}.
\newblock Ph.D. Dissertation, Katholieke Universiteit Leuven, Belgium.

\bibitem[\protect\citeauthoryear{Vlaeminck \bgroup et al\mbox.\egroup
  }{2012}]{VlaeminckVBD/KR2012}
Vlaeminck, H.; Vennekens, J.; Bruynooghe, M.; and Denecker, M.
\newblock 2012.
\newblock Ordered {E}pistemic {L}ogic: {S}emantics, complexity and
  applications.
\newblock In Brewka, G.; Eiter, T.; and McIlraith, S.~A., eds., {\em Principles
  of Knowledge Representation and Reasoning: Proceedings of the Thirteenth
  International Conference, KR 2012, Knowledge Representation and Reasoning,
  Rome, 10-14 July 2012},  369--379.
\newblock AAAI Press.

\bibitem[\protect\citeauthoryear{Waaler and Solhaug}{2005}]{tark/WaalerS05}
Waaler, A., and Solhaug, B.
\newblock 2005.
\newblock Semantics for multi-agent only knowing: extended abstract.
\newblock In van~der Meyden, R., ed., {\em Proceedings of the 10th Conference
  on Theoretical Aspects of Rationality and Knowledge (TARK-2005), Singapore,
  June 10-12, 2005},  109--125.
\newblock National University of Singapore.

\bibitem[\protect\citeauthoryear{Waaler}{2004}]{aiml/Waaler04}
Waaler, A.
\newblock 2004.
\newblock Consistency proofs for systems of multi-agent only knowing.
\newblock In Schmidt, R.~A.; Pratt{-}Hartmann, I.; Reynolds, M.; and Wansing,
  H., eds., {\em Advances in Modal Logic 5, papers from the fifth conference on
  "Advances in Modal logic," held in Manchester, UK, 9-11 September 2004},
  347--366.
\newblock King's College Publications.

\bibitem[\protect\citeauthoryear{Yang and Wooldridge}{2015}]{ijcai/2015}
Yang, Q., and Wooldridge, M., eds.
\newblock 2015.
\newblock {\em Proceedings of the Twenty-Fourth International Joint Conference
  on Artificial Intelligence, {IJCAI} 2015, Buenos Aires, Argentina, July
  25-31, 2015}. {AAAI} Press.

\end{thebibliography}

\end{document}